\documentclass[12pt]{article}

\usepackage[margin=0.7in]{geometry}
\usepackage{amsmath,amsthm,amssymb,amsfonts,mathtools}

\newtheorem{theorem}{Theorem}[section]
\newtheorem{lemma}[theorem]{Lemma}
\newtheorem{corollary}[theorem]{Corollary}

\newtheorem{definition}[theorem]{Definition}
\newtheorem{remark}[theorem]{Remark}

\newtheorem{assumption}[theorem]{Assumption}

\DeclareMathOperator*{\argmin}{arg\,min}
\DeclareMathOperator*{\lexmin}{lex\,min}
\DeclareMathOperator*{\arglexmin}{arg\,lex\,min}

\usepackage{graphicx}

\usepackage{algorithm, algpseudocode}
\algtext*{EndFor}
\algtext*{EndWhile}
\algtext*{EndIf}

\usepackage{hyperref}
\hypersetup{hidelinks=true}

\usepackage{fancyvrb}

\begin{document}

\thispagestyle{empty}
{\large\noindent IEEE Copyright Notice}
\vspace{3mm}

\noindent
\fbox{
\begin{minipage}{\linewidth}
\textcircled{c} 2025 IEEE. Personal use of this material is permitted.  Permission from IEEE must be obtained for all other uses, in any current or future media, including reprinting/republishing this material for advertising or promotional purposes, creating new collective works, for resale or redistribution to servers or lists, or reuse of any copyrighted component of this work in other works.
\end{minipage}
}

\vspace{2cm}

{\large\noindent A part of this work has been accepted to be published in the IEEE Transactions on Automatic Control.}

\vspace{2cm}

{\large\noindent Cite as:}
\vspace{3mm}

\noindent
\fbox{
\begin{minipage}{\linewidth}
    S. An, D. Lee, and G. Park, "Ultimate boundedness and output convergence of prioritized output tracking control under nonsmooth and imperfect feedback linearization," \textit{IEEE Transactions on Automatic Control}, 2025, doi: 10.1109/TAC.2025.3547538.
\end{minipage}
}

\vspace{2cm}

{\large\noindent BibTeX:}

\begin{Verbatim}[frame=single]
@article{an2025tac,
    title={Ultimate Boundedness and Output Convergence of Prioritized Output 
Tracking Control Under Nonsmooth and Imperfect Feedback Linearization},
    author={An, Sang-ik and Lee, Dongheui and Park, Gyunghoon},
    journal={IEEE Transactions on Automatic Control},
    year={2025},
    publisher={IEEE},
    doi={10.1109/TAC.2025.3547538}
\end{Verbatim}

\newpage

\title{Input-Output Feedback Linearization Preserving Task Priority for Multivariate Nonlinear Systems Having Singular Input Gain Matrix}

\author{%
    Sang-ik An, Dongheui Lee, and Gyunghoon Park
    \thanks{
        This article was presented in part at the 61st IEEE Conference on Decision and Control, Kanc\'{u}n, Mexico, December 2022 \cite{An2022} and submitted in part to the 62nd IEEE Conference on Decision and Control, 2023 \cite{An2023}.
        {\it(Corresponding Author: Gyunghoon Park.)}
    }
    \thanks{Sang-ik An and Gyunghoon Park are with School of Electrical and Computer Engineering, University of Seoul, South Korea (e-mail: S.An.Robots@gmail.com; gyunghoon.park@uos.ac.kr).}
    \thanks{Dongheui Lee is with Autonomous Systems at Institute of Computer Technology, Faculty of Electrical Engineering and Information Technology, TU Wien, Austria and also with Institute of Robotics and Mechatronics, German Aerospace Center (DLR), Germany (e-mail: dongheui.lee@tuwien.ac.at).}
}

\maketitle

\begin{abstract}

We propose an extension of the input-output feedback linearization for a class of multivariate systems that are not input-output linearizable in a classical manner.
The key observation is that the usual input-output linearization problem can be interpreted as the problem of solving simultaneous linear equations associated with the input gain matrix: thus, even at points where the input gain matrix becomes singular, it is still possible to solve a part of linear equations, by which a subset of input-output relations is made linear or close to be linear.
Based on this observation, we adopt the task priority-based approach in the input-output linearization problem.
First, we generalize the classical Byrnes-Isidori normal form to a prioritized normal form having a triangular structure, so that the singularity of a subblock of the input gain matrix related to lower-priority tasks does not directly propagate to higher-priority tasks.
Next, we present a prioritized input-output linearization via the multi-objective optimization with the lexicographical ordering, resulting in a prioritized semilinear form that establishes input output relations whose subset with higher priority is linear or close to be linear.
Finally, Lyapunov analysis on ultimate boundedness and task achievement is provided, particularly when the proposed prioritized input-output linearization is applied to the output tracking problem.
This work introduces a new control framework for complex systems having critical and noncritical control issues, by assigning higher priority to the critical ones.

\end{abstract}

\section{Introduction}
\label{sec:introduction}

Sigularity is a major cause of complexity in extending the input-output feedback linearization from a single-input single-output (SISO) system to a multi-input multi-output (MIMO) system.
A SISO system is input-output linearizable locally if and only if the system has a relative degree at a point, i.e., a scalar-valued input gain function remains to be nonzero locally  \cite{Brockett1978,Isidori1981,Byrnes1984,Byrnes1988,Khalil2015}.
A generalization of this concept to the MIMO system has been made by employing the vector relative degree defined with a matrix-valued input gain function \cite{Isidori1995} 
(while it has been further extended with the nonlinear structure algorithm \cite{Silverman1969, Hirschorn1979, Singh1981, Isidori2017}).
An essential part of the MIMO extension is to solve simultaneous linear equations related with the input gain matrix at each state. 
In this point of view, nonsingularity of the input gain matrix is understood as a necessary requirement for the input-output linearizability of the system, introducing a fundamental limitation on practical use of the input-output linearization as illustrated in related works.

In this article, we propose an extension of the input-output linearization without presuming the nonsingularity of the input gain matrix.  
An observation we made is that even at some point in the state space where the input gain matrix becomes singular, the simultaneous linear equations can still be solved partially in a sense that at least the relation between a subset of input and output variables is made linear or close to be linear.
This enlightens a possibility to generalize the input-output linearization, and to apply well-developed linear control theories for a part of non-linearizable nonlinear systems.

At this point a natural question arises: which subset of input-output relations should be chosen to be linearized? 
As a key to answer the question, the notion of the {\it priority} comes into the picture of our work. 
Prioritization is a strategy that distributes limited resources to multiple tasks.
In the context of control engineering, each task can be interpreted as  a job that a decision maker has to fulfill with a dynamical system, and its success is directly connected to the output controllability of the system associated with a variable termed a {\it task variable}, with which the task is fully represented.
In this setup, singularity of the input gain matrix may violate the output controllability of two task variables as the subblock of the  matrix associated with the task variables does not have full rank: it means that achieving these two tasks simultaneously can be impossible. 
In order to deal with these incompatible tasks, we introduce priority relations on the tasks, by which output controllability of higher-priority tasks is preserved first.
It is worth noting that the task priority-based approach has been studied intensively in the robotics society, such as 
prioritized inverse kinematics (PIK) \cite{Nakamura1987, Chiaverini1997, An2019}, 
constrained PIK \cite{Mansard2009, Kanoun2012},
task switching \cite{Lee2012, An2015}, 
prioritized control \cite{Khatib1987,Ott2015}, 
prioritized optimal control \cite{Nakamura1987a, Geisert2017}, 
learning prioritized tasks \cite{Saveriano2015, Calinon2018}, 
etc.
Nonetheless, most of the aforementioned studies have been conducted mainly focusing on practical problems, while the theoretical aspect of priority has not been understood well.
With this kept in mind, we in this work pursue to make a bridge between two successful studies on the input-output linearization and the priority-based approaches in a theoretical point of view.  
As a result, it is seen in this article that the input-output linearization can be extended to systems that are not input-output linearizable via the proposed approach with task priority, and furthermore a concrete theoretical background is provided for the task priority-based control of robotic systems.

The rest of this article is organized as follows.
In Section~\ref{sec:prioritized_control_problem_and_prioritized_normal_form}, we propose an extension of the Byrnes-Isidori normal form with the task priority, called the {\it prioritized normal form} throughout this article, by applying a specific LQ factorization to the input gain matrix.
Then Section~\ref{sec:prioritized_input_output_linearization} presents a priority-based generalization of the input-output linearization, in both canonical and generalized approaches. 
Finally in Section~\ref{sec:prioritized_tracking_control}, the prioritized output tracking control problem is addressed on top of the prioritized input-output linearization technique. 
We will observe that the singularity of the input gain matrix causes two inevitable issues in the tracking control: the factorization of the input gain matrix becomes nonsmooth at some points, and the linearization performed is imperfect. 
This work focuses on providing mathematical analysis on ultimate boundedness and higher-priority task achievement by using the Lyapunov's method, even in the presence of nonsmooth orthogonalization and imperfect linearization. 

\subsection{Notations}
\label{subsec:notations}

For a given matrix $M\in\mathbb{R}^{a\times b}$, 
$\sigma_{\rm max}(M)$ and $\sigma_{\rm min}(M)$ denote the maximum and minimum singular values of $M$, respectively, ${\cal R}(M)$ and ${\cal N}(M)$ are the range and null spaces of $M$, respectively, and $M^+$ represents the Moore-Penrose pseudoinverse of $M$. 
Together with $c\in [0,\infty]$, the extended damped pseudoinverse $M^{+(c)} \in \mathbb{R}^{b\times a}$ of $M$ is given as 
\begin{equation*}
    M^{+(c)} = \begin{dcases*} M^+, &$c=0$ \\ M^T(MM^T+c^2I_a)^{-1}, & $0<c<\infty$ \\ 0, & $c=\infty$. \end{dcases*}
\end{equation*}
For a block diagonal matrix $M = {\rm diag}( M_{11},\dots, M_{kk} )$ and a vector $c=(c_1,\dots,c_k) \in [0,\infty]^k$, 
$$M^{\oplus(c)}={\rm diag}( M_{11}^{+(c_1)},\dots, M_{kk}^{+(c_k)} ).$$
For sufficiently smooth functions $f$, $h$, and $G$, let $L_f^jh(x) = (\partial (L_f^{j-1}h)/\partial x)f(x)$, $L_GL_f^jh(x) = (\partial (L_f^jh)/\partial x)G(x)$, and $L_f^0h(x) = h(x)$.
${\rm int}(\mathcal{S})$, ${\rm bd}(\mathcal{S})$, ${\rm cl}(\mathcal{S})$, and $\overline{\rm co}(\mathcal{S})$ are the interior, the boundary, the closure, and the convex closure of a set $\mathcal{S}$, respectively, and $\mathcal{B}_n$ is the open unit ball in $\mathbb{R}^n$.

Throughout this paper, the subscript $i:j$ or $i:j,i':j'$ will be mainly employed to designate a part of the entirety in the domain of discourse.
For example, if $u = u_1 + \cdots + u_k$, $v = {\rm col}(v_1,\dots,v_k)$, and $L = [L_{ab}]_{1\le a,b\le k}$, then $u_{i:j} = u_i + \cdots + u_j$, $v_{i:j} = {\rm col}(v_i,\dots,v_j)$, and $L_{i:j,i':j'} = [L_{ab}]_{i\le a\le j, i'\le b\le j'}$, respectively.
Sometimes, we define individuals $\bullet_i = \bullet_{i:i}$ first and then construct the entirety $\bullet = \bullet_{1:k}$ later by specifying the meaning of $1:k$.
We also use the subscript $i:j$ to designate a variable or a function whose definition is closely connected to a part of the entirety.
For example, if $\xi = {\rm col}(\xi_1,\dots,\xi_k)$, we define a variable $\eta_{1:i}$ and a function $\Phi_{1:i}(x)$ satisfying $(\eta_{1:i},\xi_{1:i}) = \Phi_{1:i}(x)$.
Without any definition, $\bullet_{i:j}$ represents ${\rm col}(\bullet_i,\dots,\bullet_j)$; otherwise, we always specify the meaning of $\bullet_{i:j}$.
If we need to consider an initial segment of an index set $\mathcal{I} = \{1,\dots,k\}$, we define $\mathcal{I}_0 = \{1,\dots,i_0\} \subset \mathcal{I}$ and write $\bullet_0 = \bullet_{1:i_0}$ or $\bullet_0^c = \bullet_{i_0+1:k}$ or $\bullet_0 = \bullet_{1:i_0,1:i_0}$ or $\bullet_0^c = \bullet_{i_0+1:k,i_0+1:k}$ for the simplicity of notation.

\section{Prioritized Control Problem and Prioritized Normal Form}
\label{sec:prioritized_control_problem_and_prioritized_normal_form}

We begin this section by introducing the priotized control problem for a dynamical system and its related concepts.
A \textit{task} ${\mathfrak{T}}$ is a job, usually written as a set of statements, a decision maker has to fulfill.
In order to complete a task, a decision maker builds a dynamical system ${\mathfrak{S}}$ and develops a control strategy $\mathfrak{C}$ for an \textit{input variable} $u$ of $\mathfrak{S}$, i.e., a process to find the control input of ${\mathfrak{S}}$ for ${\mathfrak{T}}$.
For a given task $\mathfrak{T}$, the associated \textit{task variable} $y_{\mathfrak{T}}$ is an ordered set of outputs of a dynamical system ${\mathfrak{S}}$ with which the task $\mathfrak{T}$ can be represented (e.g., to minimize an objective function $\pi(y_{\mathfrak{T}})$). (Hereinafter, we drop the subscript $\mathfrak{T}$ in $y_{\mathfrak{T}}$ if obvious.)
We here do not specify what the task is, and selection of $\mathfrak{T}$ is left to a decision maker.
When $\mathfrak{T}$ is determined {\it a priori}, we say that a dynamical system $\mathfrak{S}$ is \textit{square} / \textit{redundant} / \textit{deficient} if the number of inputs is equal to / greater than / less than the number of outputs.

In this article, we are interested in a particular scenario that a decision maker needs to accomplish multiple interconnected tasks ${\mathfrak{T}}_1,\dots,{\mathfrak{T}}_{{k}}$ simultaneously for a single dynamical system $\mathfrak{S}$. 
This leads to a difficulty in controller design, as achieving two among the tasks could be impossible due to the lack of control resources. 
A remedy is to introduce the notion of {\it priority} for the tasks, and to distribute available resources according to the priority. 
In what follows, a task $\mathfrak{T}_i$ is said to \textit{have higher priority} than another one $\mathfrak{T}_j$, if it is pursued to accomplish $\mathfrak{T}_i$ first, and then  $\mathfrak{T}_j$ if still possible. 
\textit{Unprioritized tasks} or \textit{tasks without priority} $({\mathfrak{T}}_1,\dots,{\mathfrak{T}}_{{k}})$ are multiple tasks that do not have priority relations between tasks.
On the other hand, \textit{prioritized tasks} or \textit{tasks with priority} $[{\mathfrak{T}}_1,\dots,{\mathfrak{T}}_{{k}}]$ are multiple tasks that have priority relations between tasks, listed in an ascending order that $\mathfrak{T}_1$ has the highest priority and $\mathfrak{T}_{{k}}$ has the lowest priority.

Conceptually, the problem of our main interest can be formulated as follows (while its detailed version will be introduced in Section~\ref{sec:prioritized_tracking_control} after the tasks are specified to the output tracking problem):

{\bf Priotized Control Problem}: \textit{
For a given dynamical system $\mathfrak{S}$ and priotized tasks $[\mathfrak{T}_1,\dots,\mathfrak{T}_k]$, the priotized control problem is to find a control strategy ${\mathfrak{C}}$ that consists of partial control strategies ${\mathfrak{C}}_1,\dots,{\mathfrak{C}}_k$ for $\mathfrak{T}_1,\dots,\mathfrak{T}_k$, respectively, and satisfies the following:
\begin{enumerate}
    \item {\bf (Dependence)}: A control strategy ${\mathfrak{C}}_i$ for $\mathfrak{T}_i$ does not affect whether or not the tasks $\mathfrak{T}_1,\dots,\mathfrak{T}_{i-1}$ with higher priority are achieved;
    \item {\bf (Representation)}: Under the dependence property, ${\mathfrak{C}}_i$ uses maximum available control resources but no more than $\mathfrak{T}_i$ actually needs.
\end{enumerate}
}

In particular, the representation property claims that, doing nothing or unnecessary things for $\mathfrak{T}_i$ is not preferred, in order to not only complete a task $\mathfrak{T}_i$ but also deal with those with lower priorities.
On the other hand, a control strategy that solves the priotized control problem is called a {\it priotized control} and denoted as $\mathfrak{C} = [\mathfrak{C}_1,\dots,\mathfrak{C}_k]$.
As the first step of constructing a priotized control $\mathfrak{C}$, we study the Byrnes-Isidori normal form of redundant systems.

\subsection{Normal Form of Square or Redundant Systems}
\label{subsec:normal_form_of_square_or_redundant_systems}

We consider a multivariate nonlinear input-affine system
\begin{subequations}\label{eqn:dynamical_system_with_single_task}
\begin{align}
    \dot{x} &= f(x) + G(x)u \\
    y &= h(x)
\end{align}
\end{subequations}
where $f:\mathbb{R}^n\to\mathbb{R}^n$, $G:\mathbb{R}^n\to\mathbb{R}^{n\times {m}}$, and $h:\mathbb{R}^n\to\mathbb{R}^{p}$ are sufficiently smooth on $\mathbb{R}^n$ and ${p}\le {m}\le n$ (so the system is square or redundant).
Without loss of generality, $f(0) = 0$ and $h(0) = 0$. 
An usual assumption for \eqref{eqn:dynamical_system_with_single_task} is the well-defineness of the vector relative degree.

\begin{definition}
    \label{def:relative_degree}
    \cite[\S5.1]{Isidori1995}
    The system \eqref{eqn:dynamical_system_with_single_task} has a \textit{(vector) relative degree} $({r}_1,\dots,{r}_{p})\in\mathbb{N}^{p}$ at $x_0\in\mathbb{R}^n$ if 
    \begin{enumerate}
        \item $L_GL_f^jh_i(x) = 0$ for all $1\le i\le {p}$, for all $0\le j\le {r}_i-2$, and for all $x$ in a neighborhood of $x_0$;
        \item the ${p}\times {m}$ matrix 
        \begin{equation}
            J(x) = \mathrm{col}(L_GL_f^{{r}_1-1}h_1(x),\dots,L_GL_f^{{r}_{p}-1}h_{p}(x))
            \label{eqn:input_gain_matrix}
        \end{equation}
        is nonsingular at $x = x_0$. 
        $\hfill\square$
    \end{enumerate}
\end{definition}

We call ${r}_1 + \cdots + {r}_{p}$ as the \textit{total relative degree}.

\begin{lemma}
    \label{lem:diffeomorphism}
    \cite[Lemma 5.1.1, Proposition 5.1.2, Remark 5.1.3]{Isidori1995}
    Assume that the system \eqref{eqn:dynamical_system_with_single_task} has a relative degree $({r}_1,\dots,{r}_{p})$ at $x_0\in\mathbb{R}^n$.
    Denote $\phi_i^j = L_f^{j-1}h_i$ and $d\phi_i^j = \partial\phi_i^j/\partial x$ and define a distribution $\Delta = \mathrm{span}\{g_1,\dots,g_{{m}}\}$
    where $g_i$ is the $i$-th column of $G$.
    Then, we have
    \begin{enumerate}
        \item ${r} = {r}_1 + \cdots + {r}_{p} \le n$;
        \item there is a diffeomorphism $$\Phi(x) = \mathrm{col}(\phi_1(x),\dots,\phi_n(x))$$ on a neighborhood of $x_0$ satisfying $$(\phi_{n-r+1},\dots,\phi_{n}) = (\phi_1^1,\dots,\phi_1^{{r}_1},\dots,\phi_{p}^1,\dots,\phi_{p}^{{r}_{p}});$$
        \item if $f(x_0) = 0$ and $h(x_0) = 0$, then $\Phi$ can be chosen in a way that $\Phi(x_0) = 0$;
        \item if ${p}={m}$, ${r}<n$, and $\Delta$ is involutive near $x_0$,
        then $\Phi$ can be chosen in a way that $L_G\phi_i(x) = 0$ for all $1\le i\le n-r$ and all $x$ near $x_0$.
        $\hfill\square$
    \end{enumerate}
\end{lemma}

\begin{proof}
    See Appendix \ref{app:proof_of_lemma_about_diffeomorphism}
\end{proof}

Assume that the dynamical system \eqref{eqn:dynamical_system_with_single_task} has a relative degree $({r}_1,\dots,{r}_{p})$ at $0$ and define
\begin{equation}
    \kappa(x) = \mathrm{col}(L_f^{{r}_1}h_1(x),\dots,L_f^{{r}_{p}}h_{p}(x)).
    \label{eqn:kappa_function}
\end{equation}
Let $\Phi(x)$ be a diffeomorphism given by Lemma \ref{lem:diffeomorphism} and $\mathcal{D}\subset\mathbb{R}^n$ be a neighborhood of $0$ in which the relative degree and the diffeomorphism are defined.
We define new coordinates of the system as
\begin{align*}
    z = \begin{bmatrix} \eta \\ \xi \end{bmatrix} = \begin{bmatrix} \phi_{1:n-r}(x) \\ \phi_{n-r+1:n}(x) \end{bmatrix} = \Phi(x).
\end{align*}
The system \eqref{eqn:dynamical_system_with_single_task} is transformed into the (Byrnes-Isidori) normal form as (for brevity, we will utilize a notation $a[\eta,\xi] = a[z] = a(\Phi^{-1}(z))$ for a function $a$ defined on $\mathcal{D}$):
\begin{subequations}\label{eqn:normal_form}
\begin{align}
    \dot{\eta} &= {f}_\eta[\eta,\xi] + {G}_\eta[\eta,\xi]u \\
    \dot{\xi} &= A\xi + B\left( \kappa[\eta,\xi] + J[\eta,\xi]u \right) \\
    y &= C\xi
\end{align}
\end{subequations}
where $(A,B,C)$ is a canonical form of chains of integrators,%
\footnote{
    The canonical form representation is given as:
    \begin{align*}
        A_i &= \begin{bmatrix} 
            0 & 1 & 0 & \cdots & 0 & 0 \\ 
            0 & 0 & 1 & \cdots & 0 & 0 \\
            \vdots & \vdots & \vdots & \ddots & \vdots & \vdots \\
            0 & 0 & 0 & \cdots & 1 & 0 \\
            0 & 0 & 0 & \cdots & 0 & 1 \\
            0 & 0 & 0 & \cdots & 0 & 0
        \end{bmatrix} \in \mathbb{R}^{r_i\times r_i} &
        B_i &= \begin{bmatrix} 0 \\ \vdots \\ 0 \\ 1 \end{bmatrix} \in \mathbb{R}^{r_i\times 1} &
        C_i &= \begin{bmatrix} 1 & 0 & \cdots & 0 \end{bmatrix} \in \mathbb{R}^{1\times r_i} \\
        A &= \begin{bmatrix} A_1 & 0 & \cdots & 0 \\ 0 & A_2 & \cdots & 0 \\ \vdots & \vdots & \ddots & \vdots \\ 0 & 0 & \cdots & A_p \end{bmatrix}\in\mathbb{R}^{r\times r} &
        B &= \begin{bmatrix} B_1 & 0 & \cdots & 0 \\ 0 & B_2 & \cdots & 0 \\ \vdots & \vdots & \ddots & \vdots \\ 0 & 0 & \cdots & B_p \end{bmatrix}\in\mathbb{R}^{r\times p} &
        C &= \begin{bmatrix} C_1 & 0 & \cdots & 0 \\ 0 & C_2 & \cdots & 0 \\ \vdots & \vdots & \ddots & \vdots \\ 0 & 0 & \cdots & C_p \end{bmatrix}\in\mathbb{R}^{p\times r}.
    \end{align*}
}
${f}_\eta = L_f\phi_{1:n-r}$, and ${G}_\eta = L_G\phi_{1:n-r}$.
We assume $\Phi(0) = 0$; thus, ${f}_\eta(0) = {f}_\eta[0] = 0$ and $\kappa(0) = \kappa[0] = 0$.

\subsection{System Redundancy and Control of Internal Dynamics}

In this subsection, we discuss how the redundancy of a system provides an additional control resource to be used for other purposes.
For this, we for now assume that the system \eqref{eqn:dynamical_system_with_single_task} is redundant (i.e., $p<m$), and set our task as stabilization of the output $y(t)$.
Suppose that the stabilizing task is completely achieved; that is, $y(t)\equiv 0$.
Then by definition, $\xi(t) \equiv 0$ as long as ${x}_\eta(t) \coloneqq \Phi^{-1}(\eta(t),0)$ is in ${\mathcal{D}}$.
Then, the (output-zeroing) input $u_\eta(t)$ needs to satisfy
\begin{equation}
    J({x}_\eta(t))u_\eta(t) = -\kappa({x}_\eta(t)).
    \label{eqn:j_u_equals_minus_kappa}
\end{equation}
If $J(x_\eta(t))$ is nonsingular, the solution of \eqref{eqn:j_u_equals_minus_kappa} is given as
\begin{equation}
    u_\eta(t)= -J^+({x}_\eta(t))\kappa({x}_\eta(t)) + N({x}_\eta(t))u_f
    \label{eqn:u_to_find_zero_dynamics}
\end{equation}
where $u_f\in\mathbb{R}^{{m}}$ is a free variable, and 
\begin{equation}\label{eqn:null_space_projector}
    N(x) = I_{{m}} - J^+(x)J(x)
\end{equation} 
is the {\it null-space projector} with respect to $J(x)$. 
Applying \eqref{eqn:u_to_find_zero_dynamics} to the internal dynamics of \eqref{eqn:normal_form} as $u=u_\eta$ gives 
\begin{equation}
    \dot{\eta}(t) = {f}_\eta^\circ[\eta(t),0] + ({G}_\eta N)[\eta(t),0]u_f
    \label{eqn:zero_dynamics}
\end{equation}
(which is actually the zero dynamics of \eqref{eqn:dynamical_system_with_single_task})
for ${x}_\eta(t) \in {\mathcal{D}}$ where ${f}_\eta^\circ = {f}_\eta - {G}_\eta J^+\kappa$.
It should be noted that $u_f$ that appears in \eqref{eqn:u_to_find_zero_dynamics} is a free variable and plays no role in achieving \eqref{eqn:j_u_equals_minus_kappa}. 
This is a unique characteristic of the redundant system, providing an opportunity for achieving extra tasks (possibly defined on $\eta$) by utilizing $u_f$ as an {\it additional} control resource.

At this point, with Lemma~\ref{lem:diffeomorphism} taken into account, one may wonder whether ${r} = n$ or there is a particular diffeomorphism satisfying $L_G\phi_i(x) = 0$  for a redundant system, for both of which cases $u_f$ disappears in the resultant normal form; in other words, $G_\eta^\circ(x)=0$.
In fact, we can easily exclude such possibility under a mild assumption.
This will be seen by deriving the zero dynamics on the original coordinates.
Let $y(t) = 0$ for all $t$ satisfying $x(t)\in {\mathcal{D}}$.
Then, the system must evolve on the manifold
\begin{equation}
    \mathcal{Z} = \{x\in {\mathcal{D}} : L_f^jh_i(x) = 0,\,1\le i\le {p},\,0\le j\le {r}_i-1\},
    \label{eqn:manifold_of_zero_dynamics}
\end{equation}
which is smooth and has dimension $n-{r}$.
Note that $f(0) = 0$ and $h(0) = 0$ imply $0\in \mathcal{Z}$.
The state feedback $u=k(x)$ that makes $x$ be on $\mathcal{Z}$ must have the form
\begin{equation}
    k(x) = -J^+(x)\kappa(x) + N(x)u_f
    \label{eqn:u_to_find_zero_dynamics2}
\end{equation}
with the free variable $u_f\in\mathbb{R}^{{m}}$, by which the system \eqref{eqn:dynamical_system_with_single_task} (in the original coordinate) becomes
\begin{equation}
    \dot{x} = f^\circ(x) + (GN)(x)u_f
    \label{eqn:zero_dynamics_original_coordinate}
\end{equation}
where $f^\circ = f - GJ^+\kappa$.
It is easily checked that $f^\circ(x) + (GN)(x)u_f$ is tangent to $\mathcal{Z}$ for all $x\in \mathcal{Z}$ and $u_f\in\mathbb{R}^{{m}}$.%
\footnote{
    Since $L_f^jh_i(x) = 0$ for all $1\le i\le {p}$, for all $0\le j\le {r}_i-1$, and for all $x\in \mathcal{Z}$, we have
    \begin{align*}
        dL_f^jh_i(x)(f^\circ(x) + (GN)(x)u_f) 
        &= dL_f^jh_i(x)\left[f(x) - G(x)J^+(x)\kappa(x) + G(x)N(x)u_f\right] \\
        &= L_f^{j+1}h_i(x) + L_GL_f^jh_i(x)\left[ -J^+(x)\kappa(x) + N(x)u_f \right] 
        = 0
    \end{align*}
    for all $1\le i \le {p}$, $0\le j\le {r}_i-2$, and all $x\in \mathcal{Z}$.
    Also,
    \begin{align*}
        &\begin{bmatrix} dL_f^{{r}_1-1}h_1(x) \\ \vdots \\ dL_f^{{r}_{p}-1}h_{p}(x) \end{bmatrix}(f^\circ(x)+(GN)(x)u_f) 
        = \begin{bmatrix} dL_f^{{r}_1-1}h_1(x) \\ \vdots \\ dL_f^{{r}_{p}-1}h_{p}(x) \end{bmatrix}[f(x)-G(x)J^+(x)\kappa(x) + G(x)N(x)u_f] \\
        &= \begin{bmatrix} L_f^{{r}_1}h_1(x) \\ \vdots \\ L_f^{{r}_{p}}h_{p}(x) \end{bmatrix} + \begin{bmatrix} L_GL_f^{{r}_1-1}h_1(x) \\ \vdots \\ L_GL_f^{{r}_{p}-1}h_{p}(x) \end{bmatrix}[-J^+(x)\kappa(x) + N(x)u_f] 
        = \kappa(x) + J(x)[-J^+(x)\kappa(x) + N(x)u_f] \\
        &= \kappa(x) - J(x)J^+(x)\kappa(x) + J(x)N(x)u_f 
        = 0
    \end{align*}
    for all $x\in{\cal Z}$.
    Thus, 
    \begin{equation*}
        dL_f^jh_i(x)(f^\circ(x)+(GN)(x)u_f) = 0
    \end{equation*}
    for all $1\le i\le {p}$, all $0\le j\le {r}_i-1$, all $x\in \mathcal{Z}$.
    Let $x_1,x_2\in \mathcal{Z}$.
    Then,
    \begin{align*}
        0 = L_f^jh_i(x_1) 
        = L_f^jh_i(x_2) + dL_f^jh_i(x_2)(x_1-x_2) + r(x_1-x_2) 
        = dL_f^jh_i(x_2)(x_1-x_2) + r(x_1-x_2)
    \end{align*}
    and
    \begin{equation*}
        \lim_{x_1\to x_2}dL_f^jh_i(x_2)\frac{x_1-x_2}{\|x_1-x_2\|} = -\lim_{x_1\to x_2}\frac{r(x_1-x_2)}{\|x_1-x_2\|} = 0.
    \end{equation*}
    Thus, $dL_f^jh_i(x)$ is normal to $\mathcal{Z}$ for all $x\in \mathcal{Z}$.
    It follows that $f^\circ(x) + (GN)(x)u_f$ is tangent to $\mathcal{Z}$ for all $x\in \mathcal{Z}$ and $u_f\in\mathbb{R}^m$.
}
Thus, every solution $x(t)$ of \eqref{eqn:zero_dynamics_original_coordinate} initiated in $\mathcal{Z}$ stays on $\mathcal{Z}$ until it leaves $\mathcal{D}$.

\begin{lemma}
    \label{lem:condition_for_GN_equal_zero}
    Let $\Phi(x)$ be a diffeomorphism in Lemma \ref{lem:diffeomorphism}.
    Then, for $x\in \mathcal{D}$, we have $d\phi_{n-r+1:n}(x)(GN)(x) = 0$ and
    \begin{align*}
        L_G\phi_{1:n-r}(x) = 0 \implies (GN)(x) = 0 \iff (G_\eta N)(x) = 0.
    \end{align*}
    Assume $G(x)$ is nonsingular on $\mathcal{D}$.
    Then,
    \begin{equation*}
        r = n \implies p = m \iff (GN)(x) = 0
    \end{equation*}
    for each $x\in \mathcal{D}$.
    $\hfill\square$
\end{lemma}
\begin{proof}
    See Appendix \ref{app:proof_of_lemma_about_condition_for_GN_equal_zero}.
\end{proof}

Lemma~\ref{lem:condition_for_GN_equal_zero} implies that, as long as $G(x)$ is nonsingular on $\mathcal{D}$ and the system is redundant, we have $r<n$, $(GN)(x) \neq 0$, and $(G_\eta N)(x)\neq0$; in other words, the extra control input $u_f$ will not be vanished and  is able to affect the internal part of the system.

In the rest of this subsection, we clarify the meaning of $u_f$ in solving \eqref{eqn:zero_dynamics} or \eqref{eqn:zero_dynamics_original_coordinate}.
When we solve \eqref{eqn:zero_dynamics}, the value of $u_f$ should be specified for each $t$; that is, $u_f$ is in fact a function of $t$ as $u_f(t)$.
Since $u_f$ remains undetermined, candidates for $u_f(t)$ are not unique but infinitely many.
Considering the solvability of \eqref{eqn:zero_dynamics}, we restrict $u_f(t)$ to be a (Lesbesgue) measurable function satisfying $u_f(t)\in \mathcal{U}_f$ almost everywhere for some compact set $\mathcal{U}_f\subset\mathbb{R}^{{m}}$.
Then, for every $\eta^\circ$ sufficiently close to zero, there exists a (Carath\'{e}odory) solution $\eta(t)$ of \eqref{eqn:zero_dynamics} satisfying $\eta(0) = \eta^\circ$ \cite[\S1]{Filippov1988}.%
\footnote{
    A Carath\'{e}odory solution of a differential equation $\dot{x} = f(t,x)$ is a function $x(t)$ defined on an open or closed interval $l$ such that $x(t)$ is absolutely continuous on each closed interval $[\alpha,\beta]\subset l$ and satisfies $\dot{x}(t) = f(t,x(t))$ almost everywhere.
}

The existence of the solution can be extended to a {\it set of free variables}, with help of {\it differential inclusion}. Indeed, under the same assumptions on $u_f(t)$ and $\eta^\circ$, an absolutely continuous function $\eta(t)$ is a solution of \eqref{eqn:zero_dynamics} if and only if $\eta(t)$ is a solution of the differential inclusion
\begin{equation}
    \dot{\eta}(t) \in {\mathcal{F}}_\eta^\circ[\eta(t),0] = \bigcup_{u_f\in {\mathcal{U}_f}}\{{f}_\eta^\circ[\eta(t),0] + ({G}_\eta N)[\eta(t),0]u_f\}
    \label{eqn:zero_dynamics_inclusion}
\end{equation}
at least locally \cite[Theorem 2.3, \S4.2]{Smirnov2002}.
Motivated by the discussions made so far, we call \eqref{eqn:zero_dynamics_inclusion} the zero dynamics of the redundant system \eqref{eqn:dynamical_system_with_single_task} (written in a differential inclusion form).
Note that for the case of square systems, \eqref{eqn:zero_dynamics_inclusion} is the very zero dynamics defined in a conventional way (but in a differential inclusion form), because for $x\in {\mathcal{D}}$ we have
\begin{equation*}
    {p}={m} \implies N(x) = 0 \implies {\mathcal{F}}_\eta^\circ(x) = \{{f}^\circ_\eta(x)\}.
\end{equation*}
Following the same reasoning process, we also have a counterpart of \eqref{eqn:zero_dynamics_inclusion} in the original coordinate
\begin{equation}
    \dot{x} \in \mathcal{F}^\circ(x) = \bigcup_{u_f\in \mathcal{U}_f}\{f^\circ(x) + (GN)(x)u_f\} \quad (x\in \mathcal{Z}).
    \label{eqn:zero_dynamics_inclusion_original_coordinate}
\end{equation}
Obviously, \eqref{eqn:zero_dynamics_inclusion_original_coordinate} holds for square systems because if ${p}={m}$, then ${\mathcal{F}}^\circ(x) = \{f^\circ(x)\}$ for all $x\in {\mathcal{D}}$.
We close this subsection by noting that, representing the internal dynamics as the differential inclusion as above will be useful when we consider multiple tasks with priority in the following sections.

\subsection{Prioritized Normal Form}
\label{subsec:prioritized_normal_form}

We consider a dynamical system $\mathfrak{S}$ defined with a set of tasks ${\mathfrak{T}}_i$, having the form
\begin{subequations}\label{eqn:dynamical_system_with_multiple_tasks}
\begin{align} 
    \dot{x} &= f(x) + G(x)u \\
    y_i &= h_i(x) \quad (1\le i\le {k})
\end{align}
\end{subequations}
where $f:\mathbb{R}^n\to\mathbb{R}^n$, $G:\mathbb{R}^n\to\mathbb{R}^{n\times {m}}$, and $h_i:\mathbb{R}^n\to\mathbb{R}^{{p}_i}$ are sufficiently smooth on $\mathbb{R}^n$ and ${p}_i \le {m}\le n$.
Without loss of generality, $f(0) = 0$ and $h_i(0) = 0$.
Let ${p}_{i:j} = {p}_i + \cdots + {p}_j$, and ${\mathfrak{S}}_{i:j}$ be the system
\begin{align*}
    \dot{x} &= f(x) + G(x)u \\
    y_a &= h_a(x) \quad (i\le a\le j).
\end{align*}
The task variables $y_1,\dots,y_{{k}}$ are assumed to be chosen to represent the associated tasks ${\mathfrak{T}}_1,\dots,{\mathfrak{T}}_{{k}}$, respectively.

In this section, we propose a normal-form like representation of \eqref{eqn:dynamical_system_with_multiple_tasks} with the priority of the tasks taken into account, which is called a \emph{prioritized normal form} of \eqref{eqn:dynamical_system_with_multiple_tasks} in this work.
One way to find a normal form of \eqref{eqn:dynamical_system_with_multiple_tasks} is to combine all tasks as a single task ${\mathfrak{T}} = ({\mathfrak{T}}_1,\dots,{\mathfrak{T}}_{{k}})$, which is unprioritized tasks, with a task variable $y\coloneqq y_{1:{k}}$. 
Then, the normal form \eqref{eqn:normal_form} of \eqref{eqn:dynamical_system_with_single_task} can be found under the assumptions that ${p} = {p}_{1:{k}}\le {m}$ and the system \eqref{eqn:dynamical_system_with_single_task} has a relative degree at $x=0$.
However, we cannot always assume those conditions and even if it is possible, a neighborhood ${\mathcal{D}}\subset \mathbb{R}^n$ of the origin, in which the vector relative degree and the diffeomorphism of \eqref{eqn:dynamical_system_with_single_task} are defined, can be very limited.
With this kept in mind, in this work we deal with a class of nonlinear systems \eqref{eqn:dynamical_system_with_multiple_tasks} satisfying a weaker condition than above, stated in the following definition.

\begin{definition}
    \label{def:prioritizable}
    The system \eqref{eqn:dynamical_system_with_multiple_tasks} is said to be \textit{prioritizable} at $x_0\in\mathbb{R}^n$ if $G(x_0)$ has full rank and each subsystem ${\mathfrak{S}}_i$ has a vector relative degree $({r}_{i1},\dots,{r}_{i{p}_i})\in\mathbb{N}^{{p}_i}$ at $x_0$.
    $\hfill\square$
\end{definition}

From now on, we assume that the tasks are prioritized as $[{\mathfrak{T}}_1,\dots,{\mathfrak{T}}_{k}]$ and the system \eqref{eqn:dynamical_system_with_multiple_tasks} is prioritizable at 0.
Then, there exists a neighborhood ${\mathcal{D}}\subset\mathbb{R}^n$ of $0$ in which the following statements hold:
\begin{enumerate}
    \item $G(x)$ is nonsingular on $\mathcal{D}$;
    \item 
    for each $i$, there exists a neighborhood ${\mathcal{D}}_i\subset{\mathcal{D}}$ of 0 such that, with a total relative degree ${r}_i$ and a coordinate transformation
    \begin{equation*}
        z_i = \begin{bmatrix} \eta_i \\ \xi_i \end{bmatrix} = \begin{bmatrix} \phi_{i,1:n-{r}_i} \\ \phi_{i,n-{r}_i+1:n} \end{bmatrix} = \Phi_i(x)
    \end{equation*}
    defined over $\mathcal{D}_i$, the subsystem ${\mathfrak{S}}_i$ has the following  normal form    on $\mathcal{D}_i$:
    \begin{subequations}\label{eqn:normal_form_of_s_i}
    \begin{align}
        \dot{\eta}_i &= {f}_i[\eta_i,\xi_i] + G_i[\eta_i,\xi_i]u \label{eqn:normal_form_of_s_i_1} \\
        \dot{\xi}_i &= A_i\xi_i + B_i\left(\kappa_i[\eta_i,\xi_i] + J_i[\eta_i,\xi_i]u\right) \label{eqn:normal_form_of_s_i_2} \\
        y_i &= C_i\xi_i \label{eqn:normal_form_of_s_i_3}
    \end{align}
    \end{subequations}
    (hereinafter, for any function $a(x)$ defined on $\mathcal{D}_i$, we use the symbol $a[\eta_i,\xi_i]$ instead of $a(\Phi_i^{-1}(\eta_i,\xi_i))$) with $f_i(x) = L_f\phi_{i,1:n-r_i}(x)$ and $G_i(x) = L_G\phi_{i,1:n-r_i}(x)$;
    \item $L_GL_f^\alpha h_{ij}(x) = 0$ for all $1\le i\le k$, $1\le j\le p_i$, $0\le \alpha\le r_{ij}-2$, and $x\in\mathcal{D}$ where $h_{ij}$ is the $j$-th component of $h_i$.
\end{enumerate}

 It is noted that the subsystem \eqref{eqn:normal_form_of_s_i_2} and \eqref{eqn:normal_form_of_s_i_3}, also known as the partial normal form of ${\mathfrak{S}}_i$, is equivalent to
\begin{equation*}
    \mathrm{col}\left(\frac{d^{{r}_{i1}}y_{i1}}{dt^{{r}_{i1}}},\dots,\frac{d^{{r}_{i{p}_i}}y_{i{p}_i}}{dt^{{r}_{i{p}_i}}}\right) = \kappa_i(x) + J_i(x)u
\end{equation*}
where $y_{ij}$ is the $j$-th component of $y_i$, which holds not only on $\mathcal{D}_i$ but also on $\mathcal{D}$, whenever the first item of Definition \ref{def:relative_degree} is satisfied.
Thus, it is possible to extend the domain of \eqref{eqn:normal_form_of_s_i_2} and \eqref{eqn:normal_form_of_s_i_3} to $\mathcal{D}$ in a sense that \eqref{eqn:normal_form_of_s_i_2} and \eqref{eqn:normal_form_of_s_i_3} are replaced with
\begin{equation}
        \dot{\xi}_i = A_i\xi_i + B_i\left(\kappa_i(x) + J_i(x)u\right),\quad y_i = C_i \xi_i.
    \label{eqn:subsystem_on_D}
\end{equation}
From a geometric point of view, \eqref{eqn:subsystem_on_D} says that $\mathcal{R}(J_i^T(x))$ represents the set of available control input $u$ used to achieve ${\mathfrak{T}}_i$ at $x$, in a sense that all $u\in\mathcal{N}(J_i(x)) = \mathcal{R}(J_i^T(x))^\perp$ do not affect \eqref{eqn:subsystem_on_D} at $x$.
Thus, with no restriction on ${\mathfrak{T}}_i$, the ability of achieving ${\mathfrak{T}}_i$ becomes maximal if $J_i(x)$ has full rank so that ${\rm dim}(\mathcal{R}(J_i^T(x))) = p_i$: similarly, ${\rm dim}(\mathcal{R}(J^T(x))) = p$ guarantees to achieve all the possible tasks. 
For systems of being prioritizable, ${\rm dim}(\mathcal{R}(J^T(x))) = p$ cannot be achieved in general. 
In more details, $J_i(x)$ in \eqref{eqn:subsystem_on_D} is nonsingular on ${\mathcal{D}}_i$ but $J(x)$ is allowed to be singular on any point in $\mathcal{D}$ (possibly including the origin) particularly at which a conflict between two or more tasks takes place.
For example, suppose that ${\mathfrak{T}}_i$ is incompatible with the higher-priority tasks ${\mathfrak{T}}_1,\dots,{\mathfrak{T}}_{i-1}$ at some $x\in {\mathcal{D}}$; in other words,
\begin{equation*}
    \mathrm{rank}(J_{1:i}(x)) < \mathrm{rank}(J_{1:i-1}(x)) + \mathrm{rank}(J_i(x)).
\end{equation*}
Then there are some rows of $J_i(x)$ that are linearly dependent on the rows of $J_{1:i-1}(x)$, which means that
\begin{equation*}
    \mathcal{V} = \mathcal{R}(J_{1:i-1}^T(x))\cap\mathcal{R}(J_i^T(x)) \neq \{0\}.
\end{equation*}
The subspace $\mathcal{V}$ is the set of all shared control inputs that are necessary for achieving all ${\mathfrak{T}}_1,\dots,{\mathfrak{T}}_{i-1}$ and ${\mathfrak{T}}_i$. 
It then follows that both $y_{1:i-1}$ and $y_i$ cannot be controlled simultaneously and arbitrarily at $x$, leading to the failure of achieving all ${\mathfrak{T}}_1$, $\dots$, ${\mathfrak{T}}_i$. 

In order to deal with such a lack of control resources for $[{\mathfrak{T}}_1,\dots,{\mathfrak{T}}_{i}]$, we suggest to focus on the higher-priority tasks ${\mathfrak{T}}_1$, $\cdots$, ${\mathfrak{T}}_{i-1}$, while the lower one ${\mathfrak{T}}_i$ remains out of interest at least at the point $x$. 
In our case, this strategy can be implemented by factorizing $\mathcal{R}(J^T(x))$ on ${\cal D}$ with the priority taken into account. 
For this purpose, we introduce the null-space operator $N_{1:i}(x)$  and the orthogonal projector $P_i(x)$, which are constructed by the following recursive relations:
\begin{align}
N_{1:i} = I_{{m}} - J_{1:i}^+J_{1:i} ~~ (1\leq i \leq k),\quad 
 N_{1:0} = I_{{m}}
\end{align}
and 
\begin{subequations}
\begin{align}
    P_i & = (J_iN_{1:i-1})^+(J_iN_{1:i-1}),~~ (1\leq i \leq k)\\
    P_{{k}+1} & = I_{{m}} - P_{1:{k}}, 
\end{align}
\end{subequations}
where $ P_{i:j} \coloneqq P_i + \cdots + P_j$ for $1\leq i\le j\leq k+1$. 
Let us also denote the rank of $P_i(x)$ as $\rho_i(x)= {\rm rank}(P_i(x))\geq 0$ (which depends on $x$ and thus possibly becomes $0$ at some $x$), and ${\rho}_{i:j} = {\rho}_i + \cdots + {\rho}_j$, and ${\rho} = {\rho}_{1:{k}}$.
Some useful properties on $P_i$, $N_{1:i}$, and $\rho_i$ are summarized as below. 
\begin{lemma}
    For all $x\in {\mathcal{D}}$, the following statements hold:
    \begin{enumerate}
        \label{lem:orthogonal_projectors}
        \item $P_i(x)P_j(x) = 0$ if $i\neq j$;
        \item $ N_{1:i}(x) = I_{{m}} - P_{1:i}(x)$;
        \item $J_i(x) = J_i(x)P_{1:i}(x)$;
        \item $J_i(x)P_i(x) = J_i(x)N_{1:i-1}(x)$;
        \item ${\rho}_i(x) = \mathrm{rank}(J_i(x)P_i(x)) \le {p}_i$;
        \item ${\rho}_{1:i}(x) = \mathrm{rank}(P_{1:i}(x)) = \mathrm{rank}(J_{1:i}(x)) \le {p}_{1:i}$;
        \item There exist a lower triangular matrix $L_e(x)\in\mathbb{R}^{{p}\times {m}}$ and an orthogonal matrix $Q_e(x)\in\mathbb{R}^{{m}\times {m}}$ such that
        \begin{equation}
            \underbrace{\begin{bmatrix} J_1 \\ \vdots \\ J_{{k}}\end{bmatrix}}_{J(x)} 
            = \underbrace{\begin{bmatrix} L_{11} & \cdots & 0 & 0 \\ \vdots & \ddots & \vdots & \vdots \\ L_{{k}1} & \cdots & L_{{k}{k}} & 0\end{bmatrix}}_{L_e(x)=[L_{ij}(x)]}
            \underbrace{\begin{bmatrix}Q_1 \\ \vdots \\ Q_{{k}+1}\end{bmatrix}}_{Q_e(x)},
            \label{eqn:orthogonalization_of_J}
        \end{equation}
        $L_e^T(x)$ is in a row echelon form with positive leading entries, and 
        \begin{equation}
            \begin{split}
                J_iP_j &= L_{ij}Q_j \\
                P_j &= Q_j^TQ_j
            \end{split}
            \quad \left(\begin{array}{c}1\le i\le {k} \\ 1\le j\le {k}+1 \end{array}\right)
            \label{eqn:properties_of_orthogonalization_of_J}
        \end{equation}
        where $L_{ij}(x)\in\mathbb{R}^{{p}_i\times {\rho}_j(x)}$ is the $(i,j)$-th block of $L_e(x)$, $Q_i(x)\in\mathbb{R}^{{\rho}_i(x)\times {m}}$ is the $i$-th row block of $Q_e(x)$, and 
        \begin{equation}
            \label{eqn:definitions_when_r_j_equals_0}
            \begin{split}
                L_{ij}(x)Q_j(x) &\coloneqq 0 \in\mathbb{R}^{{p}_i\times {m}} \\
                Q_j^T(x)Q_j(x) &\coloneqq 0\in\mathbb{R}^{{m}\times {m}}
            \end{split}
            \quad ({\rho}_j(x) = 0).
        \end{equation}
        $\hfill\square$
    \end{enumerate}
\end{lemma}

\begin{proof}
    See Appendix \ref{app:proof_of_lemma_about_orthogonal_projectors}.
\end{proof}

\begin{remark}
    Notice that if ${\rho}_j(x) = 0$, then $L_{ij}(x)$, $1\le i\le {k}$, and $Q_j(x)$ are vanished in the decomposition \eqref{eqn:orthogonalization_of_J}.
    Nonetheless, in order to keep consistency of notations, we define \eqref{eqn:definitions_when_r_j_equals_0} for the case when $\rho_j(x)=0$, so that the structural properties in  \eqref{eqn:orthogonalization_of_J} and \eqref{eqn:properties_of_orthogonalization_of_J} hold for all ${\rho}_j(x)\ge0$.
    $\hfill\square$
\end{remark}

By letting $L=L_{1:k,1:k}$ and $Q=Q_{1:k}$, we have the reduced decomposition $J = LQ$ along with $P \coloneqq J^+J = Q^TQ = P_{1:k}$ and $N \coloneqq I-J^+J = N_{1:k} = P_{k+1}$.
For the later use, we also define $\Omega_{i:j} = \rho_{i:j}^{-1}(p_{i:j})$.
Then, each $\Omega_{1:i}$ is open and $\Omega_{1:k}\subset\Omega_{1:k-1}\subset\cdots\subset\Omega_1\subset{\cal D}$ by Item 6) of Lemma \ref{lem:orthogonal_projectors}.

The orthogonal projectors $P_i(x)$ decompose the input space $\mathbb{R}^{{m}}$ into mutually orthogonal subspaces $\mathcal{R}(P_i(x))$ as
\begin{equation*}
    \mathbb{R}^m = \mathcal{R}(P_1(x)) + \mathcal{R}(P_2(x)) + \cdots + \mathcal{R}(P_{{k}+1}(x)).
\end{equation*}
We now claim that each set $\mathcal{R}(P_i(x))$ satisfies
\begin{subequations}\label{eqn:properties_of_orthogonal_subspaces}
\begin{align}
    \mathcal{R}(J_{1:i}^T(x)) &= \mathcal{R}(J_{1:i-1}^T(x)) + \mathcal{R}(P_i(x)),\label{eqn:properties_of_orthogonal_subspaces_1}\\
    \mathcal{R}(J_{1:i-1}^T(x)) & \perp {\mathcal{R}(P_i(x))}.\label{eqn:properties_of_orthogonal_subspaces_2}
\end{align}
\end{subequations}
Indeed, by Lemma \ref{lem:orthogonal_projectors} we have, for $1\le j\le i$,
\begin{equation*}
    \mathcal{R}(J_{1:j}^T(x)) 
    = \mathcal{R}(Q_{1:j}^T(x)L_{1:j,1:j}^T(x)) 
    = \mathcal{R}(P_{1:j}(x))
\end{equation*}
where the last equality is derived from Item~3 of the lemma and the fact that $L_{ii}$ has full rank. 
From the above equality with $j=i-1$ and $j=i$, we have \eqref{eqn:properties_of_orthogonal_subspaces_1}. 
On the other hand, \eqref{eqn:properties_of_orthogonal_subspaces_2} is derived from the above and Item~1 of Lemma~\ref{lem:orthogonal_projectors}. 
The claim on \eqref{eqn:properties_of_orthogonal_subspaces} says that, each set $\mathcal{R}(P_i(x))$ is the maximum available set of control inputs for ${\mathfrak{T}}_i$ that do not influence the control of ${\mathfrak{T}}_1$, $\dots$,${\mathfrak{T}}_{i-1}$. 
With this kept in mind, we represent the input $u$ in $\mathbb{R}^m$ as a sum of mutually orthogonal vectors  
\begin{align}
    u = u_1 + \cdots + u_{{k}+1},\quad u_i = P_i(x)u \in \mathcal{R}(P_i(x)).\label{eqn:input_decomposition}
\end{align}

We are in place of providing a {\it prioritized} normal form representation of the system \eqref{eqn:dynamical_system_with_multiple_tasks}, which is the domain-extended subsystem \eqref{eqn:subsystem_on_D} along with the internal dynamics \eqref{eqn:normal_form_of_s_i_1} and the orthogonal decomposition \eqref{eqn:input_decomposition} of the input.

\begin{definition}
    \label{def:prioritized_normal_form}
    For a given system \eqref{eqn:dynamical_system_with_multiple_tasks} that is prioritizable at $0$, a \textit{prioritized normal form} of \eqref{eqn:dynamical_system_with_multiple_tasks} is defined as 
\begin{subequations}\label{eqn:prioritized_normal_form}
    \begin{align}
        \dot{\eta}_i &= f_i(x) + G_i(x)u \label{eqn:prioritized_normal_form_1} \\
        \dot{\xi}_{i} &= A_{i}\xi_{i} + B_{i}\bigg( \kappa_{i}(x) + \sum_{j=1}^iL_{ij}(x)Q_j(x)u_j \bigg) \label{eqn:prioritized_normal_form_2} \\
        y_i &= C_i\xi_i \label{eqn:prioritized_normal_form_3}
    \end{align}
    \end{subequations}
    for $1\leq i \leq k$, where $u_i$ is defined in \eqref{eqn:input_decomposition}, \eqref{eqn:prioritized_normal_form_1} holds on ${\mathcal{D}}_i$ with $x = \Phi_i^{-1}(\eta_i,\xi_i)$, and \eqref{eqn:prioritized_normal_form_2}--\eqref{eqn:prioritized_normal_form_3} hold on ${\mathcal{D}}$.
    $\hfill\square$
\end{definition}

From now on, we will pay special attention to a class of prioritizable systems, regarding the prioritized normal form \eqref{eqn:prioritized_normal_form} as the system to be controlled. 
The following theorem says that, for the usual cases with a well-defined vector relative degree, the prioritized normal form \eqref{eqn:prioritized_normal_form} turns out to be the conventional normal form \eqref{eqn:normal_form}.

\begin{theorem}
    \label{thm:prioritized_normal_form_can_be_reduced_to_Byrnes_Isidori_normal_form}
    Assume that $p\le m$, the system \eqref{eqn:dynamical_system_with_single_task} with  $({\mathfrak{T}}_1,\dots,{\mathfrak{T}}_k)$ has a vector relative degree at $0$, and $G(0)$ is nonsingular. 
    Then the associated system \eqref{eqn:dynamical_system_with_multiple_tasks} with $[{\mathfrak{T}}_1,\dots,{\mathfrak{T}}_k]$ is prioritizable at $0$.
    Moreover, for given prioritized normal form \eqref{eqn:prioritized_normal_form} of \eqref{eqn:dynamical_system_with_multiple_tasks}, there exists a normal form \eqref{eqn:normal_form} of \eqref{eqn:dynamical_system_with_single_task} that is a reduced-order model of \eqref{eqn:prioritized_normal_form} on a neighborhood of $0$.
    $\hfill\square$
\end{theorem}
\begin{proof}
    Since nonsingularity of $J(0)$ implies nonsingularity of $J_i(0)$, eash system ${\mathfrak{S}}_i$ has a relative degree at $0$ and the system \eqref{eqn:dynamical_system_with_multiple_tasks} of $[{\mathfrak{T}}_1,\dots,{\mathfrak{T}}_k]$ is prioritizable at $0$.
    Thus, there exists a prioritized normal form \eqref{eqn:prioritized_normal_form} of \eqref{eqn:dynamical_system_with_multiple_tasks} in Definition \ref{def:prioritized_normal_form}.
    Obviously, the prioritized normal form is reduced to the normal form \eqref{eqn:normal_form_of_s_i} of ${\mathfrak{S}}_1$ on $\mathcal{D}_1$.
    By Lemma \ref{lem:diffeomorphism}, there exists a neighborhood $\mathcal{D}_{1:2}\subset{\cal D}_1\cap\Omega_{1:2}$ of $0$ such that
    \begin{align*}
        &\{d\phi_{2,n-r_2+1},\dots,d\phi_{2n}\} 
        \subset \mathrm{span}\{d\phi_{11}(x),\dots,d\phi_{1,n-r_1}(x)\}
    \end{align*}
    for all $x\in \mathcal{D}_{1:2}$.
    Without loss of generality, one can reorder $\phi_{1:2}(x)$ so that $\phi_{2,n-r_2+i}(x) = \phi_{1,n-r_1-r_2+i}(x)$ holds for $1\le i\le r_2$ on $\mathcal{D}_{1:2}$.
    Thus, we have $\eta_1 = \mathrm{col}(\eta_{1:2},\xi_2)$ on $\mathcal{D}_{1:2}$ where $\eta_{1:2} = \phi_{1,1:n-r_1-r_2}(x)$.
    It follows that the prioritized normal form can be reduced to the normal form
    \begin{align*}
        \dot{\eta}_{1:2} &= f_{1:2}[\eta_{1:2},\xi_{1:2}] + G_{1:2}[\eta_{1:2},\xi_{1:2}]u \\
        \dot{\xi}_{1:2} &= A_{1:2}\xi_{1:2} + B_{1:2}\left(\kappa_{1:2}[\eta_{1:2},\xi_{1:2}] + J_{1:2}[\eta_{1:2},\xi_{1:2}]u\right) \\
        y_{1:2} &= C_{1:2}\xi_{1:2}
    \end{align*}
    of ${\mathfrak{S}}_{1:2}$ on $\mathcal{D}_{1:2}$ where $\xi_{1:2} = \mathrm{col}(\xi_1,\xi_2)$.
    By following the same procedure, we see that for each ${\mathfrak{S}}_{1:i}$, there exists a neighborhood $\mathcal{D}_{1:i} \subset {\cal D}_{1:i-1}\cap\Omega_{1:i}$ of $0$ such that the prioritized normal form can be reduced to a normal form of ${\mathfrak{S}}_{1:i}$.
    Obviously, ${\mathfrak{S}}_{1:k}$ is the system \eqref{eqn:dynamical_system_with_single_task}.
\end{proof}

\begin{remark}
It is clearly seen from Theorem~\ref{thm:prioritized_normal_form_can_be_reduced_to_Byrnes_Isidori_normal_form} that the prioritized normal form \eqref{eqn:prioritized_normal_form} is an extension of the conventional normal form \eqref{eqn:normal_form}, where the former is available even at singular points in the state space where \eqref{eqn:dynamical_system_with_single_task} with unprioritized tasks $({\mathfrak{T}}_1,\dots,{\mathfrak{T}}_k)$ does not have a relative degree.
Additionally, the prioritized normal form accepts the case $p > m$: that is, the system with  $({\mathfrak{T}}_1,\dots,{\mathfrak{T}}_k)$ is allowed to be deficient.
A possible usage of this case, e.g., $p_{1:k-1} = m$, is that we may make the last task ${\mathfrak{T}}_k$ have a control strategy that is activated only when the system is in singularity.
$\hfill\square$
\end{remark}

We close this section by remarking important points on the internal dynamics \eqref{eqn:prioritized_normal_form_1}. 
If we do not need to consider the internal behaviors of the prioritized normal form, it is enough to deal with \eqref{eqn:prioritized_normal_form_2} and \eqref{eqn:prioritized_normal_form_3} termed the \textit{partial} prioritized normal form.
However, in general handling the internal dynamics is inevitable and sometimes cumbersome.
Since a single $\Phi_i$ is not defined on the entire domain $\mathcal{D}$, we have multiple internal dynamics equations \eqref{eqn:prioritized_normal_form_1} on each subset $\mathcal{D}_i$, and even such an equation does not exist in $\mathcal{D}\setminus \bigcup_{i=1}^k \mathcal{D}_i$.

One option of tackling this issue would be to select a proper integer $i$ with which $z_{1:i_0}=(\eta_{ 1:i_0},\xi_{1:i_0 }) = \Phi_{1:i_0}(x)$ is well-defined on a  neighborhood $\mathcal{D}_{1:i_0}\subset\Omega_{1:i_0}$ of $0$, on which the prioritized normal form \eqref{eqn:prioritized_normal_form} is reduced to
\begin{subequations}\label{eqn:prioritized_normal_form_on_D1}
\begin{align}
    \dot{\eta}_{1:i_0} &= f_{1:i_0}[z_{1:i_0}] + G_{1:i_0}[z_{1:i_0}]u \label{eqn:prioritized_normal_form_on_D1_1} \\
    \dot{\xi}_i &= A_i\xi_i + B_i\kappa_i[z_{1:i_0}]+ B_i\sum_{j=1}^i(L_{ij}Q_j)[z_{1:i_0}]u_j \label{eqn:prioritized_normal_form_on_D1_2} \\
    y_i &= C_i\xi_i \label{eqn:prioritized_normal_form_on_D1_3}
\end{align}
\end{subequations}
where $1\le i\le k$ and \eqref{eqn:prioritized_normal_form_on_D1_1}--\eqref{eqn:prioritized_normal_form_on_D1_3} hold on $\Phi_{1:i_0}(\mathcal{D}_{1:i_0})$.
For the case when we confine ourselves to a compact subset of $\Phi_{1:i_0}(\mathcal{D}_{1:i_0})$, it is enough to analyze \eqref{eqn:prioritized_normal_form_on_D1} instead of \eqref{eqn:prioritized_normal_form}. 
It should be noted that, due to the existence of the internal dynamics \eqref{eqn:prioritized_normal_form_on_D1_1}, the prioritized normal form \eqref{eqn:prioritized_normal_form_on_D1_1}--\eqref{eqn:prioritized_normal_form_on_D1_3} is still over-determined (when $i_0<k$) in a sense that $\xi_{i_0+1},\dots,\xi_k$ are governed by $(\eta_{1:i_0},\xi_{1:i_0})$.%
\footnote{
    We can clearly see this  by differentiating
    \begin{equation*}
       \label{eqn:xi_i_as_a_function_of_eta_and_xi}
       \xi_i = \Xi_i(\eta_{1:i_1},\xi_{1:i_1}) \coloneqq \phi_{i,n-r_i+1:n}[\eta_{1:i_1},\xi_{1:i_1}]
    \end{equation*}
    for $i_1+1\le i\le k$ on $\mathcal{D}_{1:i_0}$ as
    \begin{equation*}
       \label{eqn:xi_i_depends_on_eta_1_and_xi_1}
       \dot{\xi}_i = \frac{\partial \Xi_i}{\partial \eta_{1:i_1}}\dot{\eta}_{1:i_1} + \frac{\partial \Xi_i}{\partial \xi_{1:i_1}}\dot{\xi}_{1:i_1}.
    \end{equation*}
    This implies that, \eqref{eqn:prioritized_normal_form_on_D1_1} does not exactly match to the internal dynamics left by \eqref{eqn:prioritized_normal_form_on_D1_2}--\eqref{eqn:prioritized_normal_form_on_D1_3} for $1\le i\le k$ but includes it together with a part of the ${\xi}_{i_0+1:k}$-dynamics; the other part of the ${\xi}_{i_0+1:k}$-dynamics is included in the ${\xi}_{1:i_0}$-dynamics.
}
We note in advance that, at the expense of complexity in the proposed framework, this over-determination provides an useful tool for mathematical analysis when we address the priotized output tracking problem, as discussed in Section \ref{sec:prioritized_tracking_control}. 

\section{Prioritized Input-Output Linearization}
\label{sec:prioritized_input_output_linearization}

In this section, we propose the \emph{prioritized} input-output linearization for the system \eqref{eqn:dynamical_system_with_multiple_tasks} with given prioritized tasks $[{\mathfrak{T}}_1,\cdots, {\mathfrak{T}}_k]$.
Under typical assumptions in Theorem~\ref{thm:prioritized_normal_form_can_be_reduced_to_Byrnes_Isidori_normal_form}, the conventional input-output linearization problem can be formulated as that of finding a control input $u$ such that the $\xi$-subsystem in the normal form \eqref{eqn:normal_form} is linearized as
\begin{align*}
\dot{\xi} = A \xi + B v, \quad y = C \xi
\end{align*}
where $v\in\mathbb{R}^p$ is a control command to be determined later after  ${\mathfrak{T}}_i$ are specified. 
This problem is equivalent to solving a set of linear inverse problems on $u$ (with $x$ and $v$ fixed)
\begin{equation}
    \label{eqn:residual}
    e_i(x,v_i,u) \coloneqq v_i - \kappa_i(x) - J_i(x)u = 0,\quad (1\leq i \leq k)
\end{equation}
where $e_i$ is termed the $i$-th {\it residual} for the linearization.
If the input gain matrix $J(x)$ has full rank, then $e=0$ is solved by
\begin{equation}
    u = J^+(x)(v - \kappa(x)) + N(x)u_f
    \label{eqn:u_for_input_output_linearization}
\end{equation}
(with the null-space projector $N(x)$ in \eqref{eqn:null_space_projector}) that serves as a linearizing state-feedback control input with which  the normal form \eqref{eqn:normal_form} is (partially) linearized as
\begin{subequations}    \label{eqn:conventional_linear_form}
\begin{align}
        \dot{\eta} &= {f}_\eta[\eta,\xi] + ({G}_\eta J^+)[\eta,\xi](v - \kappa[\eta,\xi])+ (G_\eta N)[\eta,\xi] u_f \\
        \dot{\xi} &= A\xi + Bv \\
        y &= C\xi
\end{align}
\end{subequations}
in a neighborhood of 0.
For comparison, in this work \eqref{eqn:conventional_linear_form} is called the \textit{conventional (partial) linear form} of \eqref{eqn:dynamical_system_with_single_task}.

We here extend the input-output linearization to a larger class of systems being prioritizable at the origin, with which $J(x)$ may not have full rank at some $x$ and thus the exact feedback linearization as in \eqref{eqn:conventional_linear_form} would not be possible.  
One alternative approach that we propose is to make $\|e_i\|$ as small as possible in the {\it least-squares sense}, which is called the {\it canonical} approach in this work.
Motivated by the  results on the canonical approach, a generalization of the result will be delivered in a sequel.

\subsection{Canonical Approach}
\label{subsec:canonical_approach}

In the canonical approach, we aim to keep the input-output relations of the system \eqref{eqn:dynamical_system_with_multiple_tasks} as linear as possible, with the priority of the tasks taken into account.
Noting that the linear inverse problem $e_i(x,v_i,u)=0$ for the linearization may not be solvable in $\mathcal{D}\setminus {\cal D}_i$,
a remedy is to compute the {\it least-squares solution} $u_i^{\sf c}$ ($i = 1,\cdots,k$) of \eqref{eqn:residual}, which is the solution of the optimization problem
\begin{equation}
    \min_{u_i\in\mathcal{R}(P_i(x))}\left\|e_i(x,v_i,u)\right\|^2 + \lambda_i^2(x)\|u_i\|^2
    \label{eqn:optimization_problem_for_u_i}
\end{equation}
where $\lambda_i(x)\in[0,\infty]$ represents the damping term that serves as a weight for the regularization term $\|u_i\|$. 
It is obvious from the lower-triangular structure of $J(x)$ in \eqref{eqn:prioritized_normal_form} that, once $u_1$, $\dots$, $u_{i-1}$ in \eqref{eqn:optimization_problem_for_u_i} are given as $u_1^{\sf c}$, $\dots$, $u_{i-1}^{\sf c}$, the optimization problem \eqref{eqn:optimization_problem_for_u_i} for $u_i$ admits a \textit{unique} solution
\begin{equation}
    u_i^{\sf c} = \left(Q_i^TL_{ii}^{+(\lambda_i)}\right)(x)\bigg(v_i - \kappa_i(x) - \sum_{j=1}^{i-1}(L_{ij}Q_j)(x)u_j^{\sf c} \bigg)
    \label{eqn:linearizing_control_input_u_i}
\end{equation}
that is well-defined for all $(x,v_i)\in {\mathcal{D}}\times \mathbb{R}^{p_i}$. 
(Here, it is noted that $(Q_i^TL_{ii}^{+(\lambda_i)})(x) \coloneqq 0\in\mathbb{R}^{m\times p_i}$ for $x$ satisfying $\rho_i(x) = 0$.)
Since we are dealing with the linearization problem for \eqref{eqn:dynamical_system_with_multiple_tasks}, the sum
\begin{equation}
    u^{\sf c} = u^{\sf c}_1 + \cdots + u^{\sf c}_{k+1} \in \mathbb{R}^p \label{eqn:canonical_prioritized_linearizer}
\end{equation}
of the partial solutions $u_i^{\sf c}$ in \eqref{eqn:linearizing_control_input_u_i} and the extra input $u_{k+1}^{\sf c} \coloneqq Nu_f = N_{1:k}u_f$ in $\mathcal{R}(P_{1:k})^\perp = \mathcal{R}(N)$ is called a {\it canonical prioritized linearizer} of \eqref{eqn:dynamical_system_with_multiple_tasks}. 
The following lemma shows that the canonical prioritized linearizer $u^{\sf c}=u^{\sf c}(x,v,u_f)$ in \eqref{eqn:canonical_prioritized_linearizer} has both a recursive and a closed form, among which  the former is computationally more efficient  while the latter is easier to handle in the analysis.

\begin{lemma}
    \label{lem:control_input_for_canonical_prioritized_semilinear_form}
    For given $(x,v,u_f)\in \mathcal{D}\times \mathbb{R}^p\times\mathbb{R}^m$, the canonical prioritized linearizer $u^{\sf c}(x,v,u_f)$ in  \eqref{eqn:linearizing_control_input_u_i} and \eqref{eqn:canonical_prioritized_linearizer}  satisfies the recursive relations:
 	\begin{subequations} \label{eqn:control_input_for_canonical_linear_form_recursive}
    \begin{align}
            u^{\sf c} &= u^{\sf c}_{1:k} + Nu_f,\\
            u^{\sf c}_{1:i} &= u^{\sf c}_{1:i-1} + Q_i^TL_{ii}^{+(\lambda_i)}\left(v_i - \kappa_i - J_iu_{1:i-1}^c\right) ~~ (1\leq i \leq k),\label{eqn:control_input_for_canonical_linear_form_recursive_1}
    \end{align}
	\end{subequations}
    and $u_{1:0} = 0$. 
    In addition, $u^{\sf c}$ has the closed form
    \begin{align}
        u^{\sf c} & = J^{\oplus (\lambda)} (v-\kappa) + Nu_f
        \label{eqn:control_input_for_canonical_linear_form_closed}
    \end{align}
    where $L_{\rm D} = {\rm diag}(L_{11},\dots,L_{kk})$, $L_{\rm L} = L - L_{\rm D}$,  and
    \begin{equation}\label{eqn:prioritized_damped_psuedoinverse_of_J}
        J^{\oplus(\lambda)} \coloneqq  Q^T L_{\rm D}^{\oplus(\lambda)}\big(I_p + L_{\rm L} L_{\rm D}^{\oplus(\lambda)}\big)^{-1}.
    \end{equation}
    $\hfill\square$
\end{lemma}

\begin{proof}
    See Appendix \ref{app:proof_of_lemma_about_control_input_for_canonical_prioritized_semilinear_form}.
\end{proof}

It can be seen that, as long as $J$ has full rank and $\lambda =0$, $J^{\oplus(\lambda)}$ in \eqref{eqn:prioritized_damped_psuedoinverse_of_J} becomes the conventional pseudoinverse $J^+$.
Indeed, provided that $\rho(x) = {\rm rank}(J(x)) = \min\{p,m\}$, the corresponding matrices $L_{\rm D}\in\mathbb{R}^{p\times\rho}$ is of full column rank and $I_p + L_{\rm L}L_{\rm D}^+$ is of full row rank, so that $L^+ = (L_{\rm D} + L_{\rm L})^+ = \big((I_p + L_{\rm L}L_{\rm D}^+)L_{\rm D}\big)^+ = L_{\rm D}^+(I_p + L_{\rm L}L_{\rm D}^+)^{-1}$.%
\footnote{
    For matrices $A$ and $B$, $(AB)^+ = B^+A^+$ is true if $A$ is of full column rank and $B$ is of full row rank, if $A^*A = I$, if $BB^* = I$, if $B=A^*$, or if $B=A^+$ where $A^*$ denotes the conjugate transpose of $A$ \cite{Greville1966}.
}
Since $L_{\rm D}^+ = L_{\rm D}^{\oplus(0)}$ holds (regardless of the rank condition), we have
\begin{equation}\label{eqn:J_plus_is_the_same_as_J_oplus_0}
    J^+ = (LQ)^+ = Q^TL^+ = Q^T L_{\rm D}^+\big(I_p + L_{\rm L}L_{\rm D}^+ \big)^{-1} = J^{\oplus(0)}.
\end{equation}
In this context, $J^{\oplus (\lambda)}$ in \eqref{eqn:prioritized_damped_psuedoinverse_of_J} is termed the {\it prioritized damped psuedoinverse} of $J$.

With help of the closed-form representation \eqref{eqn:control_input_for_canonical_linear_form_closed}, the residual $e(x,v,u)$ under the control $u = u^{\sf c}(x,v,u_f)$, say $e^{\sf c}(x,v)$, is computed by
\begin{align*}
e^{\sf c}(x,v) &  = v - \kappa(x) - J(x)(J^{\oplus(\lambda)}(x)(v-\kappa(x)) + N(x)u_f)\\
& =  E^{\sf c}(x)(v-\kappa(x))
\end{align*}
where 
\begin{equation}\label{eqn:canonical_residual_matrix}
    E^{\sf c}(x) = [E_{ij}^{\sf c}(x)]\coloneqq I_p - J(x) J^{\oplus (\lambda)}(x)
\end{equation}
(called {\it (canonical) residual matrix} throughout this paper) is block lower-triangular and is vanished if $p\le m$, $J(x)$ is nonsingular, and $\lambda(x)=0$.
Since $\dot{\xi}_i = A_i x_i + B_i v_i - B_i e_i(x,v_i,u)$ by definition, we finally have a canonical prioritized counterpart of \eqref{eqn:conventional_linear_form}, which is the prioritized normal form \eqref{eqn:prioritized_normal_form} with $u^{\sf c}$ employed as a control input $u$ for linearization.

\begin{definition}
    \label{def:canonical_prioritized_semilinear_form}
    For a given system \eqref{eqn:dynamical_system_with_multiple_tasks} that is prioritizable at $0$, the canonical prioritized (partial) semilinear form of \eqref{eqn:dynamical_system_with_multiple_tasks} is defined as: for  $1\le i\le k$,
    \begin{subequations}\label{eqn:canonical_prioritized_semilinear_form}
    \begin{align}
        \!\!\!\!\!\! \dot{\eta}_i &= f_i(x) + \big(G_i J^{\oplus (\lambda)}\big)(x)\left(v - \kappa(x)\right) + (G_i N)(x) u_f \label{eqn:canonical_prioritized_semilinear_form_1} \!\!\!\! \\
        \!\!\!\!\!\! \dot{\xi}_i &= A_i\xi_i + B_iv_i - B_i \sum_{j=1}^i E_{ij}^{\sf c}(x)\left(v_j - \kappa_j(x)\right) \label{eqn:canonical_prioritized_semilinear_form_2} \\
        \!\!\!\!\!\! y_i &= C_i\xi_i \label{eqn:canonical_prioritized_semilinear_form_3}
    \end{align}
    \end{subequations}
    where $J^{\oplus (\lambda)}$ and $E^{\sf c}$ are given in \eqref{eqn:prioritized_damped_psuedoinverse_of_J} and \eqref{eqn:canonical_residual_matrix}, respectively,  \eqref{eqn:canonical_prioritized_semilinear_form_1} holds on ${\mathcal D}_i$, and  \eqref{eqn:canonical_prioritized_semilinear_form_2}--\eqref{eqn:canonical_prioritized_semilinear_form_3} on ${\mathcal D}$. 
    $\hfill\square$
\end{definition}

We emphasize that 
 \eqref{eqn:canonical_prioritized_semilinear_form} preserves the linear relation between $v_i$ and $y_i$ as close as possible in the order of the priority, by which we call  \eqref{eqn:canonical_prioritized_semilinear_form} a {\it prioritized semilinear} form. 
It is also clearly seen that, the canonical prioritized semilinear form \eqref{eqn:canonical_prioritized_semilinear_form} is a natural extension of the conventional linear form \eqref{eqn:conventional_linear_form} when $\lambda=0$ (if possible), as stated below. 

\begin{theorem}
    \label{thm:canonical_prioritized_semilinear_form_can_be_reduced_to_conventional_linear_form}
    Suppose that the hypothesis of Theorem~\ref{thm:prioritized_normal_form_can_be_reduced_to_Byrnes_Isidori_normal_form} holds. 
    Then for given canonical prioritized semilinear form \eqref{eqn:canonical_prioritized_semilinear_form} of \eqref{eqn:dynamical_system_with_multiple_tasks}, there exists 
    a conventional linear form \eqref{eqn:conventional_linear_form} of \eqref{eqn:dynamical_system_with_single_task} that is a reduced-order model of \eqref{eqn:canonical_prioritized_semilinear_form} with $\lambda=0$ on a neighborhood of $0$.
    $\hfill\square$
\end{theorem}
\begin{proof}
    By Theorem \ref{thm:prioritized_normal_form_can_be_reduced_to_Byrnes_Isidori_normal_form}, the  normal form \eqref{eqn:normal_form} of \eqref{eqn:dynamical_system_with_single_task} is a reduced-order model of the prioritized normal form \eqref{eqn:prioritized_normal_form} of \eqref{eqn:dynamical_system_with_multiple_tasks}.
    The proof is completed since $J^{\oplus(0)}(x) = J^+(x)$ with $J(x)$ having full rank, as in \eqref{eqn:J_plus_is_the_same_as_J_oplus_0}.
\end{proof}

For further discussions, from the recursive form \eqref{eqn:control_input_for_canonical_linear_form_recursive} of $u^{\sf c}$, one may represent the residual $e^{\sf c}_i$ as follows: 
\begin{equation}\label{eqn:canonical_prioritized_residual_2}
e_i^{\sf c}(x,v_{1:i}) = N_i^{\sf c}(x)\bigg(  v_i - \kappa_i(x) - \sum_{j=1}^{i-1}(L_{ij}Q_j) u_j^{\sf c}(x,v_{1:i-1}) \bigg)
\end{equation}
where $P_i^{\sf c} \coloneqq L_{ii} L_{ii}^{+(\lambda_i)}$ and $N_i^{\sf c}\coloneqq I_{p_i} - P_i^{\sf c}$. 
With \eqref{eqn:canonical_prioritized_residual_2}, we have another expression of the $\xi_i$-dynamics as%
\footnote{
    The $\xi_i$-dynamics can be formulated as:
    \begin{align*}
        \dot{\xi}_i &= A_i\xi_i + B_i\left(\kappa_i + \sum_{j=1}^{i-1}L_{ij}Q_ju_j^{\sf c} + L_{ii}Q_iu_i^{\sf c}\right) \\
            &= A_i\xi_i + B_i\left(\kappa_i + \sum_{j=1}^{i-1}L_{ij}Q_ju_j^{\sf c}\right) + B_iL_{ii}Q_iQ_i^TL_{ii}^{+(\lambda_i)}\left(v_i-\kappa_i-\sum_{j=1}^{i-1}L_{ij}Q_ju_j^{\sf c}\right) \\
            &= A_i\xi_i + B_i\left(\kappa_i + \sum_{j=1}^{i-1}L_{ij}Q_ju_j^{\sf c}\right) + B_iL_{ii}L_{ii}^{+(\lambda_i)}v_i - B_iL_{ii}L_{ii}^{+(\lambda_i)}\left(\kappa_i + \sum_{j=1}^{i-1}L_{ij}Q_ju_j^{\sf c}\right) \\
            &= A_i\xi_i + B_iL_{ii}L_{ii}^{+(\lambda_i)}v_i + B_i(I_{p_i}-L_{ii}L_{ii}^{+(\lambda_i)})\left(\kappa_i + \sum_{j=1}^{i-1}L_{ij}Q_ju_j^{\sf c}\right) \\
            &= A_i\xi_i + B_iv_i - B_i(I_{p_i} - L_{ii}L_{ii}^{+(\lambda_i)})\left(v_i - \kappa_i - \sum_{j=1}^{i-1}L_{ij}Q_ju_j^{\sf c}\right).
    \end{align*}
}
\begin{equation}\label{eqn:canonical_prioritized_form_2_2}
    \dot{\xi}_i = A_i \xi_i + B_i v_i - B_ie_i^{\sf c}(x,v_{1:i}).
\end{equation}

\begin{remark}
    While $\lambda_i=0$ would be the best in the sense that $N_i^{\sf c}(x)=0$ when $\rho_i(x) = p_i$ (so that the residual $e_i^{\sf c}(x,v_{1:i})$ becomes $0$),  the control input \eqref{eqn:control_input_for_canonical_linear_form_closed} with $\lambda_i=0$ diverges in the vicinity of singularity.
    Observe that if $L_{ii}Q_i$ is continuous on ${\cal D}$ and $\rho_i(0) = p_i$, then $\Omega_i$ is open and $\|(Q_i^TL_{ii}^+)(x)\|\to\infty$ as $x$  satisfying $x\in \Omega_i$ approaches a point in $\mathrm{bd}(\Omega_i)$:
    in other words, we cannot let $\lambda_i(x) = 0$ for all $x\in \Omega_i$ if we want the control input $u$ bounded on $\mathcal{D}$.
    This defect makes the practical use of the canonical prioritized semilinear form with $\lambda=0$ impossible.

    A remedy to this problem is to make $\lambda_i(x)$ be positive for all $x\in \mathcal{D}$ that are close to singular points.
    It alleviates the divergence problem, at the expense of inaccuracy of the solution \eqref{eqn:linearizing_control_input_u_i} for $e(x,v,u) = 0$.  
    Indeed, one can show that $N_i^{\sf c}(x)$ in \eqref{eqn:canonical_prioritized_residual_2} does not disappear on $\Omega_i$ by which the linearity between $v_i$ and $y_i$ cannot be guaranteed in \eqref{eqn:canonical_prioritized_form_2_2}.
    This motivates us to choose $\lambda_i(x)$ (nonzero but) sufficiently small so that $N_i^{\sf c}(x)\approx 0$ on $\Omega_i$.  
    Nonetheless, the trailing terms  \eqref{eqn:canonical_prioritized_residual_2} in \eqref{eqn:canonical_prioritized_form_2_2} make the analysis of the canonical prioritized semilinear form nontrivial, and we will discuss it later in this article. 
    $\hfill\square$
\end{remark}

\subsection{Characterization of Proper Objective Functions}
\label{subsec:proper_objective_function}

As the first step for generalization, we introduce a general form of objective functions 
\begin{equation}
    \label{eqn:objective_functions}
    \pi_i:{\mathcal{D}}\times\mathbb{R}^{p_i}\times\mathbb{R}^m\to[0,\infty] \quad (1\le i\le k)
\end{equation}
for the $i$-th task ${\mathfrak{T}}_i$, which is expected to be chosen in a way that minimization of $\pi_i$ implies that of the residual $e_i$.  
At this point, $\pi_i(x,v_i,u)$ is not a specific function but a variable on the function space $[0,\infty]^{\mathcal{D}\times\mathbb{R}^{p_i}\times \mathbb{R}^m}$.
Since not every $\pi_i$ is allowable, we characterize the set of $\pi_i$ that are proper for the prioritized input-output linearization, as stated in the following definition.

\begin{definition}
    \label{def:proper_objective_function}
    A vector-valued objective function
    \begin{equation}
        \label{eqn:objective_function_vector}
        \pi = (\pi_1,\dots,\pi_k): \mathcal{D}\times\mathbb{R}^p\times\mathbb{R}^m\to[0,\infty]^k
    \end{equation}
    is said to be \textit{strongly} / \textit{weakly proper} for the prioritized input-output linearization of \eqref{eqn:dynamical_system_with_multiple_tasks} if $\pi$ has both / first of the next two properties:
    \begin{enumerate}
        \item ({\bf Dependence}) For all $1\le i\le k$, $x\in {\mathcal{D}}$, and $v\in\mathbb{R}^p$,
        \begin{equation}
            \label{eqn:proper_objective_function_dependence}
            \pi_i(x,v,u) = \pi_i(x,v_i,P_{1:i}(x)u) \quad (u\in\mathbb{R}^m)
        \end{equation}
        and there exists a unique $u_i^\pi$ satisfying
        \begin{equation}
            \label{eqn:proper_objective_function_uniqueness}
            u^\pi_{i}=\argmin_{u_i\in\mathcal{R}(P_i(x))} \pi_i(x,v_i,u_{1:i-1}^\pi+u_i);
        \end{equation}
        \item ({\bf Representation}): For all $1\le i\le k$ and $x\in {\mathcal{D}}$, the map $v_i\mapsto u_i^\pi$ of $\mathcal{R}((L_{ii}Q_i)(x))$ into $\mathcal{R}(P_i(x))$ is one-to-one.
        $\hfill\square$
    \end{enumerate}
\end{definition}

We note that, in the canonical approach, $\pi_i$ has the form
\begin{equation}        \label{eqn:objective_function_for_canonical_form}
    \pi_i(x,v,u) = \|e_i(x,v,u)\|^2 + \lambda_i^2(x)\|P_i(x)u\|^2
\end{equation}
(that will be discussed in details shortly).
We will see in the rest of this subsection that the items in Definition~\ref{def:proper_objective_function} are closely related with those in {\bf Prioritized Control Problem}.

({\bf Dependence}) 
Note that any $u$ in $\mathbb{R}^m$ is decomposed into 
\begin{align*}
	u = P_{1:i}(x)u + P_{i+1:k+1}(x)u = u_{1:i} + u_{i+1:k+1}.
\end{align*} 
Then the first item of Definition~\ref{def:proper_objective_function} implies that, only $u_{1:i}$ is required in  minimizing $\pi_i$ for ${\mathfrak{T}}_i$, while the remaining part $u_{i+1:k}$ (constructed for the lower-priority tasks ${\mathfrak{T}}_{i+1}$,$\dots$,${\mathfrak{T}}_k$) and $u_{k+1}$ should not affect the result of minimizing $\pi_1$,$\dots$,$\pi_{i}$; this is a formal statement of Item~1 in {\bf Prioritized Control Problem} in some sense. 
On the other hand, the uniqueness for $u_i^\pi$ is preferred in order to avoid a local minima issue, as well as  to remove some trivial choices of $\pi_i$ (e.g., $\pi_1(x,v_1,u) = 0$ for all $u\in\mathcal{R}(P_1(x))$).

({\bf Representation})  
Conceptually, Item 2 of Definition~\ref{def:proper_objective_function} can be interpreted as the condition under which the control effort $v_i$ made in {\it any} direction should not be nullified in the resulting input $u_i^\pi$. 
For a deeper understanding, we remind first that minimization of $\pi_i$ aims to that of the $i$-th residual
\begin{equation*}
    e_i(x,v_i,u)  = v_i  - (L_{ii} Q_i)(x) u_i - \kappa_i(x) - \sum_{j=1}^{i-1} (L_{ij}Q_j)(x)u_j.
\end{equation*}
When solving a minimization problem for $e_i$ (or for $\pi_i$), $u_{1:i-1}$ is already chosen {\it a priori} so that  $-\kappa_i - \sum_{j=1}^{i-1} (L_{ij}Q_j)u_j$ is regarded as being fixed. 
On the other hand, $v_i$ is determined later for achieving a specific ${\mathfrak{T}}_i$, and is allowed to be any vector in $\mathbb{R}^{p_i}$. 
Thus, it would be desirable at this stage that $(L_{ii}Q_i) u_i^\pi$ (with $u_i^\pi \in \mathcal{R}(P_i)$ by construction) should represent all the elements of $\mathcal{R}(L_{ii}Q_i)$, so that any $v_i$ in $\mathcal{R}(L_{ii}Q_i)$ is exactly compensated in the resulting  $e_i(x,v_i,u^\pi)$.
This is enabled with Item~2 of Definition~\ref{def:proper_objective_function}, as 
the mapping $v_i\mapsto u_i^\pi$ of $\mathcal{R}((L_{ii}Q_i)(x))$ into $\mathcal{R}(P_i(x))$ is onto automatically as  $\mathrm{rank}((L_{ii}Q_i)(x)) = \rho_i(x)$ by Lemma \ref{lem:orthogonal_projectors}.

\subsection{General Approach}
\label{subsec:generalization}

We are now ready to generalize the canonical approach in Subsection~\ref{subsec:canonical_approach}, by employing proper objective functions $\pi_i$ in Definition~\ref{def:proper_objective_function}. 
Since we cannot always establish a linear relation between $v_i$ and $y_i$ on $\Omega_i$ with bounded $u$, there is no reason to stick to the optimization problem \eqref{eqn:optimization_problem_for_u_i} formulated in the least-squares manner.
In other words, it seems natural to extend the least-squares problem \eqref{eqn:optimization_problem_for_u_i} to a multi-objective optimization problem with the lexicographical ordering \cite{Miettinen1999}\cite{Ehrgott2005}%
\footnote{ 
    For given vector-valued objective function $\phi={\rm col}(\phi_1,\dots,\phi_k):\mathbb{R}^m\rightarrow [0,\infty]^k$ and a constraint set $\mathcal{U}\subset\mathbb{R}^m$, the problem $\lexmin_{u\in \mathcal{U}} \phi(u)$ 
    is to find an optimal solution $u^*\in \mathcal{U}$ satisfying, for all $1\le i\le k$, $   \phi_i(u^*) = \min\{ \phi_i(u): u\in \mathcal{U} \textrm{ and } \phi_j(u) = \phi_j(u^*) \textrm{ for } 1\le j\le i-1\}$.
}
\begin{equation}
    \label{eqn:multi_objective_optimization_with_lexicographical_ordering}
    \lexmin_{u\in\mathcal{R}(J^T)}\pi(x,v,u).
\end{equation}
The following lemma shows that the properness of $\pi$ ensures the existence and uniqueness of the solution $u^\pi$ for \eqref{eqn:multi_objective_optimization_with_lexicographical_ordering}.

\begin{lemma}
    \label{lem:solution_of_proper_objective_function_equals_solution_of_multi_objective_optimization}
    Suppose that an objective function $\pi_i$ in \eqref{eqn:objective_function_vector} is given.
    Then, for every $x\in {\mathcal{D}}$ and $v\in\mathbb{R}^p$, the sum $u_1^\pi+\cdots+u_k^\pi$ of minimizers $u_i^\pi$ of the optimization problem \eqref{eqn:proper_objective_function_uniqueness} is a unique minimizer of the multi-objective optimization problem \eqref{eqn:multi_objective_optimization_with_lexicographical_ordering} with the lexicographical ordering.
    $\hfill\square$
\end{lemma}

\begin{proof}
    See Appendix \ref{app:proof_of_lemma_about_solution_of_proper_objective_function_equals_solution_of_multi_objective_optimization}.
\end{proof}

Given a proper objective function $\pi$, all the solutions for \eqref{eqn:multi_objective_optimization_with_lexicographical_ordering} can be parameterized with a free variable $u_f \in \mathbb{R}^p$ as
\begin{equation}
    \label{eqn:prioritized_control_input}
    u^\pi(x,v, u_f) = \arglexmin_{u\in \mathcal{R}(J^T) }\pi(x,v,u)+ Nu_f
\end{equation}
which is well-defined for all $x\in {\mathcal{D}}$, $v\in\mathbb{R}^p$, and $u_f\in\mathbb{R}^m$.
With respect to $u^{\sf c}$ in \eqref{eqn:canonical_prioritized_linearizer}, we say \eqref{eqn:prioritized_control_input}  \textit{$\pi$-prioritized linearizer} or simply \textit{prioritized linearizer} if $\pi$ does not need to be specified.
It is obvious by Definition \ref{def:proper_objective_function} that $u_i^\pi = u_i^\pi(x,v_{1:i}) = P_i(x) u^\pi(x,v,u_f)$ satisfies
\begin{equation}\label{eqn:input_decomposition_2}
    u^\pi(x,v, u_f) = u_1^\pi(x,v_1) + \cdots + u_k^\pi(x,v_{1:k}) + u_{k+1}^\pi(x,u_f)
\end{equation}
where $u_{k+1}^\pi = Nu_f$.
Thus, the $i$-th residual $e_i(x,v_i,u)$ with  $u=u^\pi(x,v,u_f)$ becomes
\begin{equation}
    \label{eqn:residual_rewritten}
    e_i^\pi(x,v_{1:i}) = v_i - \kappa_i(x) - \sum_{j=1}^i(L_{ij}Q_j)(x)u_j^\pi(x,v_{1:j}).
\end{equation}
We can hereby get a generalization of the canonical prioritized semilinear form \eqref{eqn:canonical_prioritized_semilinear_form} by applying the $\pi$-prioritized linearizer $u=u^\pi(x,v, u_f)$ to the prioritized normal form \eqref{eqn:prioritized_normal_form}.

\begin{definition}
    \label{def:prioritized_semilinear_form}
    For given system \eqref{eqn:dynamical_system_with_multiple_tasks} that is prioritizable at $0$ and given a strongly / weakly proper objective function $\pi$ in \eqref{eqn:objective_function_vector}, 
    the strongly / weakly prioritized (partial) semilinear form of \eqref{eqn:dynamical_system_with_multiple_tasks} is defined as: for $1\le i\le k$,
    \begin{subequations}\label{eqn:prioritized_semilinear_form}
    \begin{align}
        \dot{\eta}_i & = f_i(x) + G_i u^\pi (x,v, u_f) \label{eqn:prioritized_semilinear_form_1}\\
        \dot{\xi}_i & = A_i \xi_i + B_i v_i - B_i e_i^{\pi} (x,v_{1:i})\label{eqn:prioritized_semilinear_form_2}\\
        y_i & = C_i \xi_i\label{eqn:prioritized_semilinear_form_3}
    \end{align}
    \end{subequations}
    where $u^\pi$ and $e_i^\pi$ are defined in \eqref{eqn:prioritized_control_input} and \eqref{eqn:residual_rewritten}, respectively,   \eqref{eqn:prioritized_semilinear_form_1} holds on $\mathcal{D}_i$, and \eqref{eqn:prioritized_semilinear_form_2}--\eqref{eqn:prioritized_semilinear_form_3} on $\mathcal{D}$.
    $\hfill\square$
\end{definition}

We note that, with a particular choice of $\pi$, the prioritized semilinear form \eqref{eqn:prioritized_semilinear_form} is reduced to the canonical one \eqref{eqn:canonical_prioritized_semilinear_form}, in the following manner. 

\begin{theorem}
    \label{thm:generalization_contains_canonical_semilinear_from}
    If an objective function \eqref{eqn:objective_function_vector} is defined as $\pi=(\pi_1,\dots,\pi_k)$ with \eqref{eqn:objective_function_for_canonical_form}
    and $\lambda_i(x)\in[0,\infty]$, then $\pi$ is 
    \begin{itemize}
        \item strongly proper if $\lambda_i(x)\in[0,\infty)$ for all $(i,x)$;
        \item weakly proper if $\lambda_i(x)=\infty$ for some $(i,x)$ satisfying $\rho_i(x)>0$.
    \end{itemize}
    Moreover, the $\pi$-prioritized linearizer $u^\pi(x,v, u_f)$ in \eqref{eqn:prioritized_control_input} equals to $u^{\sf c}(x,v, u_f)$ in  \eqref{eqn:control_input_for_canonical_linear_form_closed}.
    $\hfill\square$
\end{theorem}

\begin{proof}
    By Lemma \ref{lem:orthogonal_projectors}, $\pi$ satisfies \eqref{eqn:proper_objective_function_dependence}.
    Then, the optimization problem \eqref{eqn:proper_objective_function_uniqueness} becomes identical to \eqref{eqn:optimization_problem_for_u_i} and there is a unique solution $u_i^*$ of \eqref{eqn:proper_objective_function_uniqueness} for $1\le i\le k$.
    Thus, $\pi$ is at least weakly proper and Lemma \ref{lem:solution_of_proper_objective_function_equals_solution_of_multi_objective_optimization} gives $u_i(x,v_{1:i}) = u_i^*$.
    It follows that $\pi$-prioritized linearizer  is identical to \eqref{eqn:control_input_for_canonical_linear_form_recursive} and \eqref{eqn:control_input_for_canonical_linear_form_closed} by Lemma \ref{lem:control_input_for_canonical_prioritized_semilinear_form}.
    If $\rho_i(x) = 0$, then $\mathcal{R}((L_{ii}Q_i)(x)) = \{0\}$ and $\mathcal{R}(P_i(x)) = \{0\}$.
    Thus, the mapping $v_i\mapsto u_i^*$ of $\mathcal{R}((L_{ii}Q_i)(x))$ into $\mathcal{R}(P_i(x))$ is one-to-one for all $\lambda_i(x)\in[0,\infty]$.
    Assume $\rho_i(x)>0$ and let $v_i,\bar{v}_i\in\mathcal{R}((L_{ii}Q_i)(x))$ and $v_i\neq \bar{v}_i$.
    Let $u_i^*$ and $\bar{u}_i^*$ be the unique solutions of \eqref{eqn:optimization_problem_for_u_i} given $v_i$ and $\bar{v}_i$, respectively.
    Since $u_j^*$ for $1\le j<i$ do not depend on $v_i$ and $\bar{v}_i$, we have $u_i^* - \bar{u}_i^* = \left(Q_i^TL_{ii}^{+(\lambda_i)}\right)(x)(v_i - \bar{v}_i)$
    by \eqref{eqn:linearizing_control_input_u_i}.
    If $\lambda_i(x) = \infty$, then $L_{ii}^{+(\lambda_i)}(x) = 0$ and $u_i=\bar{u}_i$.
    Thus, the mapping $v_i\mapsto u_i^*$ is not one-to-one if $\lambda_i(x) = \infty$.
    Assume $\lambda_i(x)>0$ and let 
    \begin{equation*}
        L_{ii}(x) = \underbrace{\begin{bmatrix} U_1 & U_2 \end{bmatrix}}_{U\in\mathbb{R}^{p_i\times p_i}}\underbrace{\begin{bmatrix} \Sigma_1 \\ 0 \end{bmatrix}}_{\Sigma\in\mathbb{R}^{p_i\times \rho_i(x)}}V^T = U_1\Sigma_1V^T
    \end{equation*}
    be the singular value decomposition where $U_1\in\mathbb{R}^{p_i\times\rho_i(x)}$ and $\Sigma_1,V\in\mathbb{R}^{\rho_i(x)\times\rho_i(x)}$.
    Then,
    \begin{equation*}
        (Q_i^TL_{ii}^{+(\lambda_i)})(x) = Q_i^T(x)\underbrace{V\Sigma_1(\Sigma_1^2 + \lambda_i^2(x)I_{\rho_i(x)})^{-1}}_{M\in\mathbb{R}^{\rho_i(x)\times\rho_i(x)}}U_1^T.
    \end{equation*}
    Since $\mathrm{rank}(Q_i(x)) = \mathrm{rank}(M) = \mathrm{rank}(U_1) = \rho_1(x)$,
    \begin{equation*}
        \mathcal{R}((L_{ii}Q_i)(x)) = \mathcal{R}(U_1) \perp \mathcal{R}(U_2) = \mathcal{N}((Q_i^TL_{ii}^{+(\lambda_i)})(x)).
    \end{equation*}
    Thus, $v_i - \bar{v}_i\not\in\mathcal{N}((Q_i^TL_{ii}^{+(\lambda_i)})(x))$ and $u_i^*\neq \bar{u}_i^*$.
    Therefore, the mapping $v_i\mapsto u_i^*$ is one-to-one if $\lambda_i(x)>0$.
\end{proof}

For controller design and analysis to come, it may be desired to characterize a class of the prioritized linearizers $u^\pi$ under consideration.  
In the remainder, we assume that $u^\pi$ has the following form with a block lower-triangular matrix $\Gamma(x)$:
\begin{equation}
    \label{eqn:prioritized_control_input_with_gamma}
    u^{\pi}(x,v, u_f) = Q^T L_{\rm D}^T \underbrace{\begin{bmatrix} \Gamma_{11} & \cdots & 0 \\ \vdots & \ddots & \vdots \\ \Gamma_{k1} & \cdots & \Gamma_{kk} \end{bmatrix}}_{\Gamma(x)=[\Gamma_{ij}(x)]}(v - \kappa) + Nu_f.
\end{equation}
It should be emphasized that this requirement is satisfied for a large class of $\pi$. 
For instance, if $\pi_i$ is selected as in \eqref{eqn:objective_function_for_canonical_form} with $\lambda_i(x) \in [0, \infty]$, then the resulting prioritized linearizer \eqref{eqn:control_input_for_canonical_linear_form_closed} is the same as \eqref{eqn:prioritized_control_input_with_gamma} with block lower triangular matrix
\begin{equation*}
    \Gamma = {\rm diag}(L_{11}^{\star(\lambda_1)},\dots,L_{kk}^{\star(\lambda_k)})(I_p+L_{\rm L}L_{\rm D}^{\oplus(\lambda)})^{-1}
\end{equation*}
where $L_{ii}^{\star(\lambda_i)} = (L_{ii}L_{ii}^T+\lambda_i^2I_{p_i})^+$ if $\lambda_i(x)<\infty$ and $L_{ii}^{\star(\infty)} = 0 \in\mathbb{R}^{p_i\times p_i}$.
Notice that, with $u^\pi$ in \eqref{eqn:prioritized_control_input_with_gamma}, the residual $e^{\pi}$ in \eqref{eqn:residual} for \eqref{eqn:prioritized_control_input_with_gamma} can be rewritten as 
\begin{equation}
    e(x,v,u^\pi) = \underbrace{\begin{bmatrix} E_{11}(x) & \cdots & 0 \\ \vdots & \ddots & \vdots \\ E_{k1}(x) & \cdots & E_{kk}(x) \end{bmatrix}}_{E(x) = [E_{ij}(x)] = (I_p - LL_{\rm D}^T\Gamma)(x)}(v-\kappa(x))\label{eqn:residual_rewritten2}
\end{equation}
where the block lower-triangular matrix 
\begin{equation*}
E(x) = I_p - L(x) L_{\rm D}^T(x) \Gamma(x)
\end{equation*}
is termed the \textit{residual matrix}.
It is obvious that, the prioritized semilinear form \eqref{eqn:prioritized_semilinear_form} with $u^\pi$ in \eqref{eqn:prioritized_control_input_with_gamma} becomes
\begin{subequations}\label{eqn:prioritized_semilinear_form_rewritten}
\begin{align}
    \dot{\eta}_i &= f_i(x) + (G_i Q^TL_{\rm D}^T\Gamma)(v-\kappa(x)) + (G_i N)u_f \label{eqn:prioritized_semilinear_form_rewritten_1}\\
    \dot{\xi}_i &= A_i\xi_i + B_iv_i - B_i\sum_{j=1}^iE_{ij}(x)(v_j-\kappa_j(x)) \label{eqn:prioritized_semilinear_form_rewritten_2}\\
    y_i &= C_i\xi_i \label{eqn:prioritized_semilinear_form_rewritten_3}
\end{align}
\end{subequations}
which has a general structure of the canonical semilinear form \eqref{eqn:canonical_prioritized_semilinear_form} in a sense that $J^{\oplus(\lambda)}$ and $E^{\sf c}$ in \eqref{eqn:canonical_prioritized_semilinear_form} are replaced with $Q^T L_{\rm D}^T \Gamma$ and $E$, respectively. 

The following assumption is made for technical reasons.

\begin{assumption}\label{asm:properties_of_prioritized_linearizer}
For given $(x,v,u_f)\in \mathcal{D} \times \mathbb{R}^p \times \mathbb{R}^m$, the $\pi$-prioritized linearizer $u^\pi(x,v,u_f)$ in \eqref{eqn:prioritized_control_input} has the form \eqref{eqn:prioritized_control_input_with_gamma} and satisfies the following statements for all $1\leq i \leq k$:
\begin{itemize}
    \item[1)] ${\rm rank}(L_{ii}L_{ii}^T \Gamma_{ii}(x)) = \rho_i(x)$;
    \item[2)] $    (L_{ii}L_{ii}^T\Gamma_{ii})(x) + (L_{ii}L_{ii}^T\Gamma_{ii})^T(x) \ge 0$.
    $\hfill\square$
\end{itemize}
\end{assumption}

It is readily seen that, the first item of the assumption  is satisfied with any strongly proper $\pi$, while the second item holds for the canonical form of $\pi$ in \eqref{eqn:objective_function_for_canonical_form}. 
On the other hand, by the structural property of $\Gamma$, we have 
\begin{equation}\label{eqn:xi_subsystem_Lure_form}
    \dot{\xi}_i = A_i \xi_i + B_i L_{ii}L_{ii}^T \Gamma_{ii} v_i + B_iE_{ii}\kappa_i - B_i\sum_{j=1}^{i-1}E_{ij} (v_j - \kappa_j)
\end{equation}
where we use the fact that the $(i,j)$-th block of $E$ is given as $E_{ij} = 0$ for $i<j$, $E_{ij} = I_{p_i} - L_{ii}L_{ii}^T\Gamma_{ii}$ for $i=j$, and $E_{ij} = -\sum_{a=j}^iL_{ia}L_{aa}^T\Gamma_{aj}$ for $i>j$.%
\footnote{
    The residual matrix can be formulated as:
    \begin{align*}
        E &= I_p - LL_D^T\Gamma \\
            &= I_p - \begin{bmatrix} L_{11} & 0 & \cdots & 0 \\ L_{21} & L_{22} & \cdots & 0 \\ \vdots & \vdots & \ddots & \vdots \\ L_{k1} & L_{k2} & \cdots & L_{kk} \end{bmatrix}\begin{bmatrix} L_{11}^T & 0 & \cdots & 0 \\ 0 & L_{22}^T & \cdots & 0 \\ \vdots & \vdots & \ddots & \vdots \\ 0 & 0 & \cdots & L_{kk}^T \end{bmatrix}\begin{bmatrix} \Gamma_{11} & 0 & \cdots & 0 \\ \Gamma_{21} & \Gamma_{22} & \cdots & 0 \\ \vdots & \vdots & \ddots & \vdots \\ \Gamma_{k1} & \Gamma_{k2} & \cdots & \Gamma_{kk} \end{bmatrix} \\
            &= I_p - \begin{bmatrix} L_{11} & 0 & \cdots & 0 \\ L_{21} & L_{22} & \cdots & 0 \\ \vdots & \vdots & \ddots & \vdots \\ L_{k1} & L_{k2} & \cdots & L_{kk} \end{bmatrix}\begin{bmatrix} L_{11}^T\Gamma_{11} & 0 & \cdots & 0 \\ L_{22}^T\Gamma_{21} & L_{22}^T\Gamma_{22} & \cdots & 0 \\ \vdots & \vdots & \ddots & \vdots \\ L_{kk}^T\Gamma_{k1} & L_{kk}^T\Gamma_{k2} & \cdots & L_{kk}^T\Gamma_{kk} \end{bmatrix} \\
            &= I_p - \begin{bmatrix} L_{11}L_{11}^T\Gamma_{11} & 0 & \cdots & 0 \\ \sum_{a=1}^2L_{2a}L_{aa}^T\Gamma_{a1} & L_{22}L_{22}^T\Gamma_{22} & \cdots & 0 \\ \vdots & \vdots & \ddots & \vdots \\ \sum_{a=1}^kL_{ka}L_{aa}^T\Gamma_{a1} & \sum_{a=2}^kL_{ka}L_{aa}^T\Gamma_{a2} & \cdots & L_{kk}L_{kk}^T\Gamma_{kk} \end{bmatrix} \\
            &= \begin{bmatrix} I_{p_1} - L_{11}L_{11}^T\Gamma_{11} & 0 & \cdots & 0 \\ -\sum_{a=1}^2L_{2a}L_{aa}^T\Gamma_{a1} & I_{p_2} - L_{22}L_{22}^T\Gamma_{22} & \cdots & 0 \\ \vdots & \vdots & \ddots & \vdots \\ -\sum_{a=1}^kL_{ka}L_{aa}^T\Gamma_{a1} & -\sum_{a=2}^kL_{ka}L_{aa}^T\Gamma_{a2} & \cdots & I_{p_k}-L_{kk}L_{kk}^T\Gamma_{kk} \end{bmatrix}.
    \end{align*}
}
Then the two requirements in Assumption~\ref{asm:properties_of_prioritized_linearizer} can be interpreted as a condition for the nonlinear term $(L_{ii} L_{ii}^T \Gamma_{ii})(x)$ in \eqref{eqn:xi_subsystem_Lure_form}, which will be utilized in the analysis to come.

\begin{remark}
The discussions made so far (including Definition \ref{def:proper_objective_function}, Lemma \ref{lem:solution_of_proper_objective_function_equals_solution_of_multi_objective_optimization}, and Definition \ref{def:prioritized_semilinear_form}) introduce a systematic way of defining various prioritized semilinear forms \eqref{eqn:prioritized_semilinear_form}.
Note that $\pi$ determines how close the subsystem \eqref{eqn:prioritized_semilinear_form_2}--\eqref{eqn:prioritized_semilinear_form_3} is to the linear system $(A_i,B_i,C_i)$ in the sense of \eqref{eqn:proper_objective_function_uniqueness}, where $B_i e_i^{\pi}$ is regarded as the error for linearization.   
Thus, we have a possibility to obtain a prioritized semilinear form that gives a better performance for tasks by constructing a suitable $\pi$, which will be left for future works.
$\hfill\square$
\end{remark}

\section{Prioritized Output Tracking Control}
\label{sec:prioritized_tracking_control}

In this section, we address the output tracking problem of the system \eqref{eqn:dynamical_system_with_multiple_tasks} with the prioritized tasks $[{\mathfrak{T}}_1,\dots, {\mathfrak{T}}_k]$, on top of the results on the prioritized input-output linearization in Section~\ref{sec:prioritized_input_output_linearization}. 
In other words, we specify the $i$-th task ${\mathfrak{T}}_i$ as the tracking task for the $i$-th task variable $y_i$, with a reference signal $y_{i}^\star:[0,\infty)\to\mathbb{R}^{p_i}$ that is assumed to be sufficiently smooth. 

We begin by remarking that  under strong assumptions as in Theorem~\ref{thm:prioritized_normal_form_can_be_reduced_to_Byrnes_Isidori_normal_form}, the system \eqref{eqn:dynamical_system_with_single_task} associated with $({\mathfrak{T}}_1,\dots,{\mathfrak{T}}_k)$ has the conventional linear form \eqref{eqn:conventional_linear_form}, from which the output tracking problem is solved by applying a linear control law 
\begin{equation}
    \label{eqn:external_reference_input}
    v_i = -K_i (\xi_i - \xi_i^\star) + \kappa_i^\star,\quad i = 1,\cdots,k
\end{equation}
where $\xi_i^\star(t)$ and $\kappa_i^\star(t)$ are given by 
\begin{align*}
    \xi_i^\star &= \mathrm{col}(y_{i1}^{\star(0)},\dots,y_{i1}^{\star(r_{i1}-1)},\dots,y_{ip_i}^{\star(0)},\dots,y_{ip_i}^{\star(r_{ip_i}-1)}) \\
    \kappa_{i}^\star &= \mathrm{col}(y_{i1}^{\star(r_{i1})},\dots,y_{ip_i}^{\star(r_{ip_i})}),
\end{align*}
and are assumed to be bounded on $[0,\infty)$.
Then, with the error variable 
$$\tilde{\xi}_i = \xi_i - \xi_{i}^\star,\quad i=1,\dots,k,$$ 
and $u_f=0$, the closed-loop error dynamics is computed as
\begin{subequations}     	\label{eqn:closed_system_conventional_linear_form}
\begin{align} 
    \dot{\eta} &= ({f}_\eta - {G}_\eta J^+\kappa)[\eta,\tilde{\xi}+\xi^\star] + {G}_\eta J^+ [\eta,\tilde{\xi}+\xi^\star](-K\tilde\xi + \kappa^\star) \\
    \dot{\tilde\xi} &= (A-BK)\tilde\xi
\end{align}
\end{subequations}
in which all the tracking error variables $\tilde{\xi}_i(t)$ converge to $0$ as long as $A-BK$ is Hurwitz where $K = \mathrm{diag}(K_1,\dots,K_k)$.%
\footnote{
    \label{foot:convergnece_and_boundedness_of_single_task}
    Concerning the boundedness of $\eta(t)$, it is well known that if the origin of $\dot{\eta} = f_\eta^\circ[\eta,0]$ is asymptotically stable and $A-BK$ is Hurwitz, then for every sufficiently small $\eta(0)$, $\tilde{\xi}(0)$, $\xi^\star(t)$, and $\kappa^\star(t)$, there exists a unique classical solution $(\eta,\tilde{\xi}):[0,\infty)\to\Phi({\cal D})$ of \eqref{eqn:closed_system_conventional_linear_form} with the initial value $(\eta(0),\tilde{\xi}(0))$ such that $\eta(t)$ is bounded on $[0,\infty)$ and $\tilde{\xi}(t)\to0$ as $t\to\infty$.
    See Appendix \ref{app:proof_of_footnote_about_convergnece_and_boundedness_of_single_task} for the proof.
}

As discussed so far, if $J(x)$ becomes singular at some point on ${\cal D}$, to obtain \eqref{eqn:closed_system_conventional_linear_form} we should restrict ourselves to a neighborhood ${\cal D}_{1:k}\subset\Omega_{1:k}$ of 0 (that is possibly too small or empty); otherwise, it may be impossible to achieve all the tracking tasks simultaneously. 
As an alternative, we in this work propose the $\pi$-prioritized linearizer \eqref{eqn:prioritized_control_input_with_gamma} along with the linear control law \eqref{eqn:external_reference_input} as the controller that resolves the prioritized output tracking problem.  
It will be seen that, the proposed controller achieves higher-priority output tracking tasks (so that some of $\xi_i(t)$ approach the reference $\xi_i^\star(t)$) even in the presence of singularity of $J(x)$. 

The following assumption allows to enlarge the region of our interest.  
\begin{assumption}\label{asm:region_of_interest}
    There exist an integer $1\le i_0 \le k$ and a neighborhood ${\cal D}_{1:i_0}\subset\Omega_{1:i_0}$ of 0 such that a deffiomorphism $\Phi_{1:i_0}(x)$ is well defined on ${\cal D}_{1:i_0}$ and \eqref{eqn:prioritized_normal_form_on_D1} holds on $\Phi_{1:i_0}({\cal D}_{1:i_0})$.
    Let ${\cal C}_{1:i_0}\subset\Phi_{1:i_0}({\cal D}_{1:i_0})$ be a compact neighborhood of 0.
    $\hfill\square$
\end{assumption}

For simplicity, in what follows we frequently use $\bullet_0$ instead of $\bullet_{1:i_0}$ or $\bullet_{1:i_0,1:i_0}$, as discussed in Subsection \ref{subsec:notations}.
We also write $\bullet(x)$ or $\bullet[z_0] = \bullet[\eta_0,\xi_0] \coloneqq \bullet(\Phi_0^{-1}(\eta_0,\xi_0))$ interchangeably for a function $\bullet$ defined on ${\cal D}_0$.
If $\bullet[z_0]$ is bounded or Lipschitz on $\mathcal{C}_0$, ${\sf M}_\bullet>0$ or ${\sf L}_\bullet>0$ will be used to denote an upper bound of $\|\bullet[z_0]\|$ or a Lipschitz constant of $\bullet[z_0]$ on $\mathcal{C}_0$, respectively.
Additionally, ${\sf M}_{\xi_0^\star}>0$ and ${\sf M}_{\kappa_0^\star}>0$ will be used to denote upper bounds of $\xi_0^\star$ and $\kappa_0^\star$ on $[0,\infty)$, respectively:
\begin{align}\label{eqn:upper_bound_of_reference_signals}
    \|\xi_{0}^\star(t)\| \le {\sf M}_{\xi_0^\star},~~~~\|\kappa_{0}^\star(t)\| \le {\sf M}_{\kappa_0^\star} \qquad(t\in [0,\infty)).
\end{align}

It is emphasized that as long as $i_0 < k$, the set ${\cal D}_0$ can be viewed as an extension of the domain ${\cal D}_{1:k}$ for \eqref{eqn:closed_system_conventional_linear_form}, since ${\cal D}_0$ includes singular points on the lower-priority tasks ${\mathfrak{T}}_{i_0+1},\dots,{\mathfrak{T}}_k$.
It is easily seen that, as in  \eqref{eqn:prioritized_normal_form_on_D1} on the enlarged region ${\cal D}_0$, the error dynamics of the overall system with the proposed controller \eqref{eqn:prioritized_control_input_with_gamma} and \eqref{eqn:external_reference_input} on ${\cal D}_0$ is derived as
\begin{subequations}     \label{eqn:prioritized_closed_loop_system_on_omega_0}
\begin{align}
    \dot{\eta}_0 &= f_0^\circ + G_0^\circ(-K_0\tilde{\xi}_0 + \kappa_0^\star) + G_0u_{i_0+1:k+1}^\pi \label{eqn:prioritized_closed_loop_system_on_omega_0_eta} \\
    \dot{{\tilde{\xi}}}_i &= (A_i - B_iK_i){\tilde{\xi}}_i + B_i\sum_{j=1}^iE_{ij}(K_j{\tilde{\xi}}_j-\kappa_{j}^\star + \kappa_j) \label{eqn:prioritized_closed_loop_system_on_omega_0_xi}
\end{align}
\end{subequations}
for all $1\le i\le k$ where $G_0^\circ = G_0Q_0^TL_{D0}^T\Gamma_0$, $f_0^\circ = f_0 - G_0^\circ\kappa_0$, and $u^\pi = u_0^\pi + u_{i_0+1:k+1}^\pi = u_{1:i_0}^\pi + u_{i_0+1:k}^\pi + Nu_f$ with
\begin{equation*}
    u_{a:b}^\pi \coloneqq \sum_{i=a}^bQ_i^TL_{ii}^T\sum_{j=1}^i\Gamma_{ij}(-K_j{\tilde{\xi}}_j + \kappa_j^\star - \kappa_j), ~~ 1\leq a,b\leq k. 
\end{equation*}

\subsection{Lyapunov Analysis of Tracking Error Dynamics With Higher Priority Under Imperfect Linearization}
\label{subsec:lyapunov_analysis_of_tracking_error_dynamics_with_higher_priority_under_imperfect_linearization}

A remarkable point on \eqref{eqn:prioritized_closed_loop_system_on_omega_0} is that the $\tilde{\xi}$-dynamics in \eqref{eqn:prioritized_closed_loop_system_on_omega_0_xi}  is not a linear system due to the imperfect linearization (and also due to the singularity of $J(x)$). 
This implies that, simply letting $A-BK$ being Hurwitz as in \eqref{eqn:closed_system_conventional_linear_form} may be not enough for stabilizing the error dynamics. 
 
We here propose a new way of stabilizing $\tilde{\xi}_i(t)$ with higher priority (i.e., $i=1,\dots,i_0$) in the presence of perturbations due to the imperfect linearization. 
This subsection aims to provide the underlying principle and conditions for design parameters (i.e., $\Gamma$ and $K$), while the detailed proof can be found in Subsection~\ref{subsec:main_result}. 
The key idea is to represent the $\tilde{\xi}$-dynamics in \eqref{eqn:prioritized_closed_loop_system_on_omega_0_xi} as an interconnection of $k$ isolated systems. 
To proceed, we assume the following first. 
\begin{assumption}\label{asm:Gamma}
    The matrix $\Gamma$ in \eqref{eqn:prioritized_control_input_with_gamma} satisfies the following: 
    \begin{itemize}
        \item[(a)] For each $1\le i\le i_0$, there exists $\varsigma_i>0$ satisfying
        \begin{equation}
            \frac{1}{2}\left(L_{ii}L_{ii}^T\Gamma_{ii} + (L_{ii}L_{ii}^T\Gamma_{ii})^T\right)[\eta_0,\xi_0] \ge \varsigma_iI_{p_i}
            \label{eqn:diagonal_block_positive_definite}
        \end{equation}
        on ${\mathcal{C}}_0$;
        \item[(b)] $\Gamma_{ij}(x)$, $1\le i,j\le i_0$, is sufficiently smooth on ${\mathcal{D}}_0$;
        \item[(c)] $\|\Gamma(x)\|$ is bounded on every compact subset of ${\mathcal{D}}_0$.
        $\hfill\square$
    \end{itemize}
\end{assumption}

\begin{remark}
Assumption~\ref{asm:Gamma} holds for a large class of the $\pi$-prioritized linearizer $u^\pi$ in \eqref{eqn:prioritized_control_input_with_gamma}.
For instance, Assumption \ref{asm:properties_of_prioritized_linearizer} implies \eqref{eqn:diagonal_block_positive_definite} and the canonical prioritized control input $u^{\sf c}$ in  \eqref{eqn:control_input_for_canonical_linear_form_closed} with sufficiently smooth damping functions $\lambda_1(x),\dots,\lambda_{i_0}(x)$ and positive damping functions $\lambda_{i_0+1}(x),\dots,\lambda_k(x)$ on $\mathcal{C}_0$ satisfies the latter two items.
$\hfill\square$
\end{remark}

With $\varsigma_i$, we represent each $\tilde{\xi}_i$-dynamics as a linear system%
\footnote{
    The $\tilde{\xi}_i$-dynamics can be reformulated as:
    \begin{align*}
        \dot{{\tilde{\xi}}}_i &= (A_i - B_iK_i){\tilde{\xi}}_i + B_i\sum_{j=1}^iE_{ij}(K_j{\tilde{\xi}}_j-\kappa_{j}^\star + \kappa_j) \\
            &= (A_i - B_iK_i){\tilde{\xi}}_i + B_i(I_{p_i}-L_{ii}L_{ii}^T\Gamma_{ii})(K_i{\tilde{\xi}}_i-\kappa_i^\star + \kappa_i) + B_i\sum_{j=1}^{i-1}E_{ij}(K_j{\tilde{\xi}}_j-\kappa_{j}^\star + \kappa_j) \\
            &= (A_i - \varsigma_iB_iK_i)\tilde{\xi}_i + B_i\underbrace{(-1)(L_{ii}L_{ii}^T\Gamma_{ii} - \varsigma_iI_{p_i})K_i\tilde{\xi}_i}_{u_{\tilde{\xi}_i}} + B_i\underbrace{\bigg(E_{ii}(\kappa_i-\kappa_i^\star) + \sum_{j=1}^{i-1}E_{ij}(K_j{\tilde{\xi}}_j-\kappa_{j}^\star + \kappa_j)\bigg)}_{w_{\tilde{\xi}_i}}.
    \end{align*}
}
\begin{subequations}\label{eq:error_dynamics_linear_system}
\begin{align}
    \dot{\tilde{\xi}}_i & = \big(A_i- \varsigma_i B_i K_i \big) \tilde{\xi}_i + B_i u_{\tilde{\xi}_i} + B_i w_{\tilde{\xi}_i},\\
    y_{\tilde{\xi}_i} & = K_i \tilde{\xi}_i, 
\end{align}
\end{subequations}
with the nonlinear feedback term
\begin{align}
    u_{\tilde{\xi}_i} = - \big(  L_{ii} L_{ii}^T \Gamma_{ii}  - \varsigma_i I_{p_i}\big) y_{\tilde{\xi}_i}\label{eq:error_dynamics_nonlinearity}
\end{align}
and the interconnection term 
\begin{equation}
    w_{\tilde{\xi}_i} = E_{ii}(\kappa_i-\kappa_i^\star) + \sum_{j=1}^{i-1} E_{ij} (K_j \tilde{\xi}_j - \kappa_j^\star + \kappa_j).
\end{equation}
Note that the transfer function of \eqref{eq:error_dynamics_linear_system} from $u_{\tilde{\xi}_i}$ to $y_{\tilde{\xi}_i}$ is given by 
\begin{equation}
    \label{eqn:transfer_function_H_i}
    H_i(s) = K_i(sI_{r_i} - A_i + \varsigma_iB_iK_i)^{-1}B_i
\end{equation}
which satisfies the following \cite{An2023}.

\begin{lemma}
    \label{lem:strictly_positive_real_transfer_function}
    For any $\varsigma_i>0$, there exists $K_i\in\mathbb{R}^{p_i\times r_i}$ such that $(A_i-\varsigma_iB_iK_i,B_i)$ is controllable, $(A_i-\varsigma_iB_iK_i,K_i)$ is observable, and $H_i(s)$ in \eqref{eqn:transfer_function_H_i} is strictly positive real.
    $\hfill\square$
\end{lemma}

From now on, assume the following:

\begin{assumption}\label{asm:feedback_gain_matrix}
    For each $K_i$, $(A_i-\varsigma_iB_iK_i,B_i)$ is controllable, $(A_i-\varsigma_iB_iK_i,K_i)$ is observable, and $H_i(s)$ is strictly positive real.
    $\hfill\square$
\end{assumption}

Then the Kalman-Yakubovich-Popov lemma \cite[Lemma 6.3]{Khalil2015} says that, for each $1\le i\le i_0$, there exist a constant $\vartheta_i>0$ and  matrices $X_i=X_i^T>0$ and $R_i$ satisfying
\begin{align}    
    X_i (A_i - \varsigma_i B_i K_i) + (A_i - \varsigma_i B_i K_i)^TX_i &= -R_i^T R_i - 2\vartheta_iX_i \notag \\
    X_iB_i &= K_i^T. \label{eqn:Kalman_Yakubovichi_Popov_lemma}
\end{align}
With $X_i$ above, we employ a Lyapunov function candidate for each $\tilde{\xi}_i$ as follows:
\begin{equation}\label{eqn:Lyapunov_function_i}
    V_{\tilde{\xi}_i}({\tilde{\xi}}_i) \coloneqq \sqrt{{\tilde{\xi}}_i^TX_i{\tilde{\xi}}_i}.
\end{equation}
Then the time derivative $\dot{V}_{\tilde{\xi}_i}$ along with \eqref{eq:error_dynamics_linear_system} is computed as%
\footnote{
    If $\tilde{\xi}_i$ is differentiable at $t$ and $\tilde{\xi}_i(t)\neq0$, then $V_{\tilde{\xi}}$ is also differentiable at $t$ and we can compute an upper bound of $\dot{V}_{\tilde{\xi}_i}$ as
    \begin{align*}
        \dot{V}_{\tilde{\xi}_i} &= \frac{d(\tilde{\xi}_i^TX_i\tilde{\xi}_i)^{1/2}}{dt} = \frac{d(\tilde{\xi}_i^TX_i\tilde{\xi}_i)^{1/2}}{d(\tilde{\xi}_i^TX_i\tilde{\xi}_i)}\frac{d(\tilde{\xi}_i^TX_i\tilde{\xi}_i)}{dt} = \frac{1}{2\sqrt{\tilde{\xi}_i^TX_i\tilde{\xi}_i}}\left(\tilde{\xi}_i^TX_i\dot{\tilde{\xi}}_i + \dot{\tilde{\xi}}_i^TX_i\tilde{\xi}_i\right) \\
        &= \frac{1}{2\sqrt{\tilde{\xi}_i^TX_i\tilde{\xi}_i}}\bigg(\tilde{\xi}_i^TX_i\big((A_i - \varsigma_iB_iK_i)\tilde{\xi}_i + B_iu_{\tilde{\xi}_i} + B_iw_{\tilde{\xi}_i}\big) + \big((A_i - \varsigma_iB_iK_i)\tilde{\xi}_i + B_iu_{\tilde{\xi}_i} + B_iw_{\tilde{\xi}_i}\big)^TX_i\tilde{\xi}_i\bigg) \\
        &= \frac{1}{2\sqrt{\tilde{\xi}_i^TX_i\tilde{\xi}_i}}\bigg(\tilde{\xi}_i^T\underbrace{\big(X_i(A_i-\varsigma_iB_iK_i) + (A_i-\varsigma_iB_iK_i)^TX_i\big)}_{-R_i^TR_i-2\vartheta_iX_i}\tilde{\xi}_i + \tilde{\xi}_i^T\underbrace{(X_iB_i)}_{K_i^T}(u_{\tilde{\xi}_i} + w_{\tilde{\xi}_i}) + (u_{\tilde{\xi}_i} + w_{\tilde{\xi}_i})^T\underbrace{(X_iB_i)^T}_{K_i}\tilde{\xi}_i \bigg) \\
        &= \frac{1}{2\sqrt{\tilde{\xi}_i^TX_i\tilde{\xi}_i}}\bigg(\tilde{\xi}_i^T(-R_i^TR_i - 2\vartheta_iX_i)\tilde{\xi}_i + \tilde{\xi}_i^TK_i^Tu_{\tilde{\xi}_i} + u_{\tilde{\xi}_i}^TK_i\tilde{\xi}_i + \tilde{\xi}_i^TX_iB_iw_{\tilde{\xi}_i} + w_{\tilde{\xi}_i}^TB_i^TX_i^T\tilde{\xi}_i\bigg) \\
        &= \frac{1}{2\sqrt{\tilde{\xi}_i^TX_i\tilde{\xi}_i}}\bigg(\tilde{\xi}_i^T(-\underbrace{R_i^TR_i}_{\ge0} - 2\vartheta_iX_i)\tilde{\xi}_i - \tilde{\xi}_i^TK_i^T\underbrace{\big(L_{ii}L_{ii}^T\Gamma_{ii} + (L_{ii}L_{ii}^T\Gamma_{ii})^T - 2\varsigma_iI_{p_i}\big)}_{\ge0}K_i\tilde{\xi}_i + \tilde{\xi}_i^TX_iB_iw_{\tilde{\xi}_i} + w_{\tilde{\xi}_i}^TB_i^TX_i^T\tilde{\xi}_i\bigg) \\
        &\le -\vartheta_i\sqrt{\tilde{\xi}_i^TX_i\tilde{\xi}_i} + \frac{\|X_i\tilde{\xi}_i\|}{\sqrt{\tilde{\xi}_i^TX_i\tilde{\xi}_i}}\underbrace{\|B_i\|}_{=1}\|w_{\tilde{\xi}_i}\| \\
        &\le -\vartheta_i\sigma_{\min}^{1/2}(X_i)\|\tilde{\xi}_i\| + \frac{\sigma_{\max}(X_i)}{\sigma_{\min}^{1/2}(X_i)}\|w_{\tilde{\xi}_i}\|.
    \end{align*}
    Indeed, this upper bound holds almost everywhere if $\tilde{\xi}_i(t)$ is absolutely continuous.
    To see this, let $\tilde{\xi}_i(t)$ be absolutely continuous on some interval and observe that $V_{\tilde{\xi}_i}(\tilde{\xi}_i)$ is Lipschitz such that $V_{\tilde{\xi}_i}(\tilde{\xi}_i(t))$ is also absolutely continuous on the interval.
    It follows that $V_{\tilde{\xi}_i}(\tilde{\xi}_i(t))$ is differentiable almost everywhere.
    Let both $\tilde{\xi}(t)$ and $V_{\tilde{\xi}_i}(\tilde{\xi}_i(t))$ be differentiable at $t_0$ (it happens almost everywhere).
    If $\tilde{\xi}_i(t_0)\neq0$, then the upper bound holds at $t_0$ as discussed.
    Assume $\tilde{\xi}_i(t_0) = 0$.
    Since $V_{\tilde{\xi}_i}(\tilde{\xi}_i(t))\ge0$ for all $t$, it must be $\dot{V}_{\tilde{\xi}_i} = 0$ at $t_0$; otherwise, $V_{\tilde{\xi}_i}(\tilde{\xi}_i(t))<0$ for some $t$ in the vicinity of $t_0$.
    Since the upper bound of $\dot{V}_{\tilde{\xi}_i}$ is nonnegative when $\tilde{\xi}_i(t_0) = 0$, the upper bound holds at $t_0$.
}
\begin{equation}
    \dot{V}_{\tilde{\xi}_i} \le - \vartheta_i {\sigma_{\min}^{1/2}(X_i)} \|\tilde{\xi}_i\| + \frac{\sigma_{\rm max}(X_i)}{{\sigma_{\rm min}^{1/2}(X_i)}}  \|w_{\tilde{\xi}_i}\|
    \label{eqn:upper_bound_of_V_i_dot}
\end{equation} 
(as long as $\tilde{\xi}_i$ is differentiable at $t$ and $\tilde{\xi}_i(t)\neq0$) where we use  \eqref{eqn:diagonal_block_positive_definite} and \eqref{eqn:Kalman_Yakubovichi_Popov_lemma}.
 
In the analysis to come, we deal with the interconnection term $w_{\tilde{\xi}_i}$ in \eqref{eqn:upper_bound_of_V_i_dot} by selecting the weights ${\sf w}_i$, $i=1,\dots,i_0$, used in the sum of the Lyapunov functions
\begin{align}\label{eqn:Lyapunov_xi}
    V_{\tilde{\xi}_0}(\tilde{\xi}_0)\coloneqq \sum_{i=1}^{i_0}\frac{{\sf w}_i}{\vartheta_i {\sigma_{\min}^{1/2} (X_i)}} V_{\tilde{\xi}_i}(\tilde{\xi}_i).
\end{align}
To present a condition for ${\sf w}_i$, note first that $w_{\tilde{\xi}_i}$ in \eqref{eqn:upper_bound_of_V_i_dot} is bounded by
\begin{subequations}\label{eqn:bound_w_tilde_xi}
\begin{equation*}
    \|w_{\tilde{\xi}_i}\| \leq \sum_{j=1}^{i-1} \|E_{ij}\|\| K_j\| \|\tilde{\xi}_j\| + \sum_{j=1}^i \|E_{ij}\| \big( \|\kappa_j^\star\| + \|\kappa_j\| \big)
\end{equation*}
\end{subequations}
where $E_{ij}[z_0]$ and $\kappa_j[z_0]$ satisfy
\begin{align}
    \|E_{ij}[z_0]\| & \le {\sf M}_{E_{ij}},\label{eqn:bounds_E_kappa} \\
    \|\kappa_i[z_0]\| & \leq  {\sf L}_{\kappa_i} \| z_0\| \leq {\sf L}_{\kappa_i} \bigg( \|\eta_0\| + \sum_{a=1}^{i_0}\| \tilde{\xi}_a \| + \|\xi_0^\star\| \bigg) \notag
\end{align}
for any $1\leq i,j\leq i_0$ and $z_0 \in \mathcal{C}_0$.
Then one obtains, for each $1\leq i \leq i_0$, that%
\footnote{
    The upper bound is computed as
    \begin{align*}
        \frac{{\sf w}_i}{\vartheta_i {\sigma_{\min}^{1/2} (X_i)}} \dot{V}_{\tilde{\xi}_i}
            &\le \frac{{\sf w}_i}{\vartheta_i {\sigma_{\min}^{1/2} (X_i)}}\bigg( -\vartheta_i\sigma_{\min}^{1/2}(X_i)\|\tilde{\xi}_i\| + \frac{\sigma_{\max}(X_i)}{\sigma_{\min}^{1/2}(X_i)}\|w_{\tilde{\xi}_i}\| \bigg) \\
            &= -{\sf w}_i\|\tilde{\xi}_i\| + \frac{{\sf w}_i\tilde{\sigma}_i}{\vartheta_i}\bigg( \sum_{j=1}^{i-1}\|E_{ij}\|\|K_j\|\|\tilde{\xi}_j\| + \sum_{j=1}^i\|E_{ij}\|(\|\kappa_j^\star\|+\|\kappa_j\|) \bigg) \\
            &\le -{\sf w}_i\|\tilde{\xi}_i\| + \frac{{\sf w}_i\tilde{\sigma}_i}{\vartheta_i}\bigg( \sum_{j=1}^{i-1}{\sf M}_{E_{ij}}\|K_j\|\|\tilde{\xi}_j\| + \sum_{j=1}^i{\sf M}_{E_{ij}}\bigg({\sf M}_{\kappa_j^\star}+{\sf L}_{\kappa_j}\bigg(\|\eta_0\| + \sum_{a=1}^{i_0}\|\tilde{\xi}_a\| + {\sf M}_{\xi_0^\star}\bigg)\bigg) \bigg) \\
            &= -{\sf w}_i\|\tilde{\xi}_i\| - {\sf w}_i\sum_{j=1}^{i-1}\underbrace{\bigg(-\frac{\tilde{\sigma}_i}{\vartheta_i}{\sf M}_{E_{ij}}\|K_j\|\bigg)}_{{\sf y}_{ji}}\|\tilde{\xi}_j\| + {\sf w}_i\underbrace{\bigg(\frac{\tilde{\sigma}_i}{\vartheta_i}\sum_{j=1}^i{\sf M}_{E_{ij}}{\sf L}_{\kappa_j}\bigg)}_{{\sf z}_{ai}}\sum_{a=1}^{i_0}\|\tilde{\xi}_a\| \\
            &\quad + \frac{{\sf w}_i\tilde{\sigma}_i}{\vartheta_i}\sum_{j=1}^i{\sf M}_{E_{ij}}\big({\sf L}_{\kappa_j}\|\eta_0\| + {\sf L}_{\kappa_j}{\sf M_{\xi_0^\star}} + {\sf M}_{\kappa_j^\star}\big) \\
            &= -\bigg( {\sf w}_i \|{\tilde{\xi}}_i\| + {\sf w}_i \sum_{j=1}^{i-1}{\sf y}_{ji}\|{\tilde{\xi}}_j\| - {\sf w}_i\sum_{j=1}^{i_0}{\sf z}_{ji}\|{\tilde{\xi}}_j\| \bigg) + \frac{{\sf w}_i\tilde{\sigma}_i}{\vartheta_i} \sum_{j=1}^{i} {\sf M}_{E_{ij}} \big( {\sf L}_{\kappa_j} \|\eta_0\| + {\sf L}_{\kappa_j} {\sf M}_{\xi_0^\star} +  {\sf M}_{\kappa_j^\star} \big).
    \end{align*}
}
\begin{align}
    & \frac{{\sf w}_i}{\vartheta_i {\sigma_{\min}^{1/2} (X_i)}} \dot{V}_{\tilde{\xi}_i}\nonumber\\
    & \le  -\bigg( {\sf w}_i \|{\tilde{\xi}}_i\| + {\sf w}_i \sum_{j=1}^{i-1}{\sf y}_{ji}\|{\tilde{\xi}}_j\| - {\sf w}_i \sum_{j=1}^{i_0}{\sf z}_{ji}\|{\tilde{\xi}}_j\| \bigg) 	\label{eqn:upper_bound_of_V_i_dot_with_weight} \\
    &\quad + \frac{{\sf w}_i\tilde{\sigma}_i}{\vartheta_i} \sum_{j=1}^{i} {\sf M}_{E_{ij}} \big( {\sf L}_{\kappa_j} \|\eta_0\| + {\sf L}_{\kappa_j} {\sf M}_{\xi_0^\star} +  {\sf M}_{\kappa_j^\star} \big)\notag 
\end{align}
where $\tilde{\sigma}_i\coloneqq {\sigma_{\max}(X_i)}/{\sigma_{\min}(X_i)}$ and for $1\leq i,j\leq i_0$
\begin{align*}
    {\sf y}_{ij} \coloneqq \begin{dcases*} 0, & $i>j$ \\ 1, & $i=j$ \\ - \frac{\tilde{\sigma}_j}{\vartheta_j} {\sf M}_{E_{ji}}\|K_i\|, & $i<j$ \end{dcases*},~
    {\sf z}_{ij} \coloneqq \frac{\tilde{\sigma}_j}{\vartheta_j} \sum_{a=1}^j {\sf M}_{E_{ja}}{\sf L}_{\kappa_a}.
\end{align*}
Note here that the sum of the brace terms in \eqref{eqn:upper_bound_of_V_i_dot_with_weight} can be rewritten as%
\footnote{
    The sum can be formulated as
    \begin{align*}
        &\sum_{i=1}^{i_0} \bigg( {\sf w}_i \|\tilde{\xi}_i\| + {\sf w}_i \sum_{j=1}^{i-1} {\sf y}_{ji} \|\tilde{\xi}_j\| - {\sf w}_i\sum_{j=1}^{i_0}{\sf z}_{ji}\|{\tilde{\xi}}_j\| \bigg)
        = \sum_{i=1}^{i_0}\bigg({\sf w}_i\sum_{j=1}^{i_0}({\sf y}_{ji}-{\sf z}_{ji})\|\tilde{\xi}_j\|\bigg)
        = \sum_{i=1}^{i_0} \bigg(\sum_{j=1}^{i_0} ({\sf y}_{ij} - {\sf z}_{ij}){\sf w}_j\bigg) \|\tilde{\xi}_i\| \\
        &= \begin{bmatrix} \displaystyle\sum_{j=1}^{i_0}({\sf y}_{1j}-{\sf z}_{1j}){\sf w}_j & \cdots & \displaystyle\sum_{j=1}^{i_0}({\sf y}_{i_0j}-{\sf z}_{i_0j}){\sf w}_j \end{bmatrix}\begin{bmatrix} \|\tilde{\xi}_1\| \\ \vdots \\ \|\tilde{\xi}_{i_0}\| \end{bmatrix}
        = \big(({\sf Y} - {\sf Z}){\sf w} \big)^T \begin{bmatrix} \|\tilde{\xi}_1\| \\ \vdots \\ \|\tilde{\xi}_{i_0}\| \end{bmatrix}.
    \end{align*}
}
\begin{align*}
    & \sum_{i=1}^{i_0} \bigg( {\sf w}_i \|\tilde{\xi}_i\| + {\sf w}_i \sum_{j=1}^{i-1} {\sf y}_{ji} \|\tilde{\xi}_j\| - {\sf w}_i\sum_{j=1}^{i_0}{\sf z}_{ji}\|{\tilde{\xi}}_j\| \bigg)\\
    &  = \sum_{i=1}^{i_0} \bigg(\sum_{j=1}^{i_0} ({\sf y}_{ij} - {\sf z}_{ij}){\sf w}_j\bigg) \|  \tilde{\xi}_i\| = \big(({\sf Y} - {\sf Z}){\sf w} \big)^T \begin{bmatrix}
        \|\tilde{\xi}_1\|\\
        \vdots\\
        \|\tilde{\xi}_{i_0}\|
    \end{bmatrix}
\end{align*}
where ${\sf Y}\coloneqq [{\sf y}_{ij}] \in \mathbb{R}^{i_0 \times i_0}$ and ${\sf Z}\coloneqq [{\sf z}_{ij}]\in \mathbb{R}^{i_0 \times i_0}$. 
We finally derive the time derivative of $V_{\tilde{\xi}_0}$ as follows:
\begin{align}\label{eqn:upper_bound_of_V_xi}
    \dot{V}_{\tilde{\xi}_0} & \leq - \big(({\sf Y} - {\sf Z}){\sf w} \big)^T \begin{bmatrix}
        \|\tilde{\xi}_1\|\\
        \vdots\\
        \|\tilde{\xi}_{i_0}\|
    \end{bmatrix}\\
    &~~~  + \sum_{i=1}^{i_0}\frac{{\sf w}_i\tilde{\sigma}_i}{\vartheta_i} \sum_{j=1}^{i} {\sf M}_{E_{ij}} \big( {\sf L}_{\kappa_j} \|\eta_0\| + {\sf L}_{\kappa_j} {\sf M}_{\xi_0^\star} + {\sf M}_{\kappa_j^\star} \big).\notag
\end{align} 

The following assumption provides a sufficient condition for existence of ${\sf w}\in(0,\infty)^{i_0}$ satisfying $({\sf Y}-{\sf Z}){\sf w}\in(0,\infty)^{i_0}$.
\begin{assumption}\label{asm:M_matrix_condition_with_spectral_radius}
    $\mathrm{sr}({\sf Y}^{-1}{\sf Z}) < 1$.
    $\hfill\square$
\end{assumption}

Indeed, all the off-diagonal entries of ${\sf Y}-{\sf Z}$ is non-positive, and thus ${\sf Y}-{\sf Z}$ is a Z-matrix by definition.
Moreover, since ${\sf Y}^{-1}$ and ${\sf Z}$ are non-negative, ${\sf Y}-{\sf Z}$ is an M-matrix if and only if Assumption~\ref{asm:M_matrix_condition_with_spectral_radius} holds \cite{Plemmons1977}, under which one can always find ${\sf w}\in (0,\infty)^{i_0}$ satisfying $({\sf Y}-{\sf Z}){\sf w}\in(0,\infty)^{i_0}$. 

\begin{remark}
It should be also pointed out that the matrices $\sf Y$ and $\sf Z$ depend on the bounds ${\sf M}_{E_{ij}}$ for $\|E_{ij}\|$ and $\|K\|$. 
Assumption~\ref{asm:M_matrix_condition_with_spectral_radius} holds if the size of the residual matrix $\|E_0\|$ (cannot be zero identically but) is sufficiently small on ${ \mathcal{C}}_0$, in which ${\sf Y}^{-1}$ and ${\sf Z}$ are approximated as $I_{i_0}$ and $0$, respectively.
$\hfill\square$
\end{remark}

\subsection{Handling Issues Caused by Nonsmooth Orthogonalization}
\label{subsec:handling_issues_caused_by_nonsmooth_orthogonalization}

The main difficulty  in analyzing the overall system is that the differential equation in \eqref{eqn:prioritized_closed_loop_system_on_omega_0} may not be smooth. 
Conceptually, this problem stems from the fact that the orthogonalization $J_iP_j = L_{ij}Q_j$ is not smooth in general \cite{An2019a}\cite{An2022} and the nonsmoothness of $L_{ij}Q_j$ propagates to the optimal solution $u_i^\pi$ of the problem \eqref{eqn:proper_objective_function_uniqueness} by the representation property of Definition \ref{def:proper_objective_function}.
(We remark that a similar conclusion was reached in the authors' previous work on the prioritized inverse kinematic problem \cite{An2019a}.) 
Since the proof of this nonsmoothness property demands a lot, we do not go in detail for this work. 
Nonetheless, for a rigorous analysis it is still necessary to investigate the existence of a solution for this discontinuous dynamical system \eqref{eqn:prioritized_closed_loop_system_on_omega_0}. 
In the rest of this subsection, as an alternative we define a Krasovskii solution of \eqref{eqn:prioritized_closed_loop_system_on_omega_0} with the differential inclusion approach. 
Before going further, it is noted that the domain of interest ${\cal C}_0$ is defined as a set of $(\eta_0,\xi_0)$, not $(\eta_0, \tilde{\xi}_0)$. 
For this reason, it would be natural to define a solution of  the error dynamics \eqref{eqn:prioritized_closed_loop_system_on_omega_0} from that of 
\begin{align}
    \dot{\eta}_0 &= f_0^\circ + G_0^\circ (-K_0(\xi_0-\xi_{0}^\star) + \kappa_{0}^\star) + G_0 u_{i_0+1:k+1}^\pi \notag \\
    \dot{\xi}_i &= A_i\xi_i + B_i(-K_i(\xi_i-\xi_{i}^\star) + \kappa_{i}^\star) \notag \\
    &\quad + B_i\sum_{j=1}^iE_{ij}(K_j(\xi_j-\xi_{j}^\star) - \kappa_{j}^\star + \kappa_j). \label{eqn:prioritized_semilinear_form_with_tracking_control_input}
\end{align}
Throughout this paper, we will say that for given $(\eta_0(t), \xi_0(t))$ that satisfies \eqref{eqn:prioritized_semilinear_form_with_tracking_control_input}, the associated $(\eta_0(t), \xi_0(t) - \xi_0^\star(t))$ is a solution of the error dynamics \eqref{eqn:prioritized_closed_loop_system_on_omega_0}. 

\begin{remark}
    For brevity, we do not explicitly define the space of solutions here and do not distinguish the solutions unless otherwise noted.
    In addition, the $\xi_i$-dynamics for $i_0 + 1\leq i \leq k$ are ignored in defining a solution of \eqref{eqn:prioritized_closed_loop_system_on_omega_0}, as the differential equations \eqref{eqn:prioritized_semilinear_form_with_tracking_control_input} are overdetermined when $i_0<k$. 
    $\hfill\square$
\end{remark}

The following assumption is made for the free variable $u_f$.

\begin{assumption}\label{asm:free_variable_in_compact_set}
    There exists convex compact $\mathcal{U}_f\subset\mathbb{R}^m$ satisfying $u_f(t,x)\in\mathcal{U}_f$ for all $(t,x)\in[0,\infty)\times\Phi_0^{-1}(\mathcal{C}_0)$.
    $\hfill\square$
\end{assumption}

It should be clarified here that, only $u_{i_0+1:k+1}^\pi$ in the $\eta_0$-dynamics is the source of the nonsmoothness. 
Indeed, $\Omega_0$ is open and $L_{ij}Q_j$  ($1\le i,j\le i_0$) is smooth on $\Omega_0$ \cite{An2019a}.
It means that, among two components consisting of $u^\pi$ in \eqref{eqn:input_decomposition_2}, $u_{1:i_0}^\pi$ is smooth on $\Omega_0$ and discontinuity takes place in only $u_{i_0+1:k+1}^\pi$.
With this kept in mind, we introduce a differential inclusion
\begin{subequations}\label{eqn:differential_inclusion_of_prioritized_closed_system_on_omega_0}
\begin{align}
    \dot{\eta}_{0} &\in \mathcal{K} \big[f_{0}^\circ + G_0^\circ(-K_0(\xi_0 - \xi_{0}^\star)+\kappa_{0}^\star) \label{eqn:differential_inclusion_of_prioritized_closed_system_on_omega_0_eta} \\
    & \quad\quad\quad\quad\quad\quad\quad\quad  + G_0 u_{i_0+1:k}^\pi + G_0N u_f\big] \notag\\
    \dot{\xi}_0 &= A_0 \xi_0 + B_0(-K_0(\xi_0-\xi_0^\star) + \kappa_0^\star) \label{eqn:differential_inclusion_of_prioritized_closed_system_on_omega_0_xi} \\
    &\quad + B_0 E_0 (K_0(\xi_0-\xi_0^\star) - \kappa_0^\star + \kappa_0). \notag
\end{align}
\end{subequations}
which is the overall dynamics \eqref{eqn:prioritized_semilinear_form_with_tracking_control_input} on ${\cal D}_0$ with the $\eta_0$-dynamics rewritten in a differential inclusion form, where
\begin{equation}
    \label{eqn:Krasovskii_regularization}
    \mathcal{K}[\star](t,x) = \bigcap_{\delta>0}\overline{\mathrm{co}}(\star(t,x + \delta\mathrm{cl}(\mathcal{B}_n)))
\end{equation}
is the Krasovskii regularization of a function $\star(t,x)$ \cite{Smirnov2002}\cite{Cortes2008}.
In the remainder of this work, a trajectory $(\eta_0(t),\tilde{\xi}_0(t))$ is said to be  a Krasovskii solution of \eqref{eqn:prioritized_closed_loop_system_on_omega_0}, if $z_0(t)=(\eta_0(t),\xi_0(t))$ is the solution of the differential inclusion \eqref{eqn:differential_inclusion_of_prioritized_closed_system_on_omega_0}.%
\footnote{
    A solution of a differential inclusion $\dot{x}\in{\cal F}(t,x)$ is a function $x(t)$ defined on an interval $l$ such that $x(t)$ is absolutely continuous on $l$ and satisfies $\dot{x}(t)\in{\cal F}(t,x(t))$ almost everywhere.
}

From now on, we will observe that a local solution must exist if the state $z_0(t)$ is initiated in the interior of $\mathcal{C}_0$. 
\begin{lemma}\label{lem:solution_existence_of_differential_inclusion}
    Suppose that all assumptions hold.
    Then for each $z_0^\circ \in {\rm int}(\mathcal{C}_0)$, there exists a Krasovskii solution $z_0:[0,T)\rightarrow {\rm int}(\mathcal{C}_0)$ of \eqref{eqn:prioritized_semilinear_form_with_tracking_control_input} that is initiated at $z_0(0)=z_0^\circ$ and is defined over $[0,T)$ for some constant $T \in (0,\infty]$. 
    Moreover, $T$ satisfies either $T=\infty$ or $z_0(T)\in\mathrm{bd}({\mathcal{C}}_0)$.
    $\hfill\square$
\end{lemma}

\begin{proof}
    See Appendix \ref{lem:solution_existence_of_differential_inclusion}.
\end{proof}

The following lemma introduces another form of differential inclusion for the internal dynamics, which makes the analysis to come easier. 

\begin{lemma}
    \label{lem:superset_of_differential_inclusion}
    Suppose that all assumptions hold. 
    Then there exists an upper semicontinuous%
\footnote{
    A set-valued map ${\cal F}:{\cal X}\to2^{\cal Y}$ is called upper semicontinuous at $x\in{\cal X}$ if for every open set ${\cal B}\subset{\cal Y}$ containing ${\cal F}(x)$ there exists a neighborhood ${\cal A}\subset{\cal X}$ of $x$ such that ${\cal F}({\cal A})\coloneqq \bigcup_{z\in{\cal A}}{\cal F}(z)\subset{\cal B}$.
    A set-valued map ${\cal F}:{\cal X}\to2^{\cal Y}$ is said to be upper semicontinuous if ${\cal F}$ is upper semincontinuous for each $x\in{\cal X}$.
}
    set-valued map $\mathcal{U}: \Phi_0^{-1}(\mathcal{C}_0) \rightarrow 2^{\mathbb{R}^m}$ such that the following holds:
    \begin{itemize}
    \item[(a)] for each $x \in \Phi_0^{-1}(\mathcal{C}_0)$, $\mathcal{U}(x)$ is a compact and convex neighborhood of $0$;
    \item[(b)] there exists a positive constant ${\sf L}_{\mathcal{U}}$ such that 
    \begin{equation}
        \mathcal{U}[\eta_0,\xi_0]\subset \mathcal{U}[\eta_0,0] + {\sf L}_{\mathcal{U}}\|\xi_0\|\mathrm{cl}(\mathcal{B}_m)
    \end{equation}
    for all $(\eta_0,\xi_0)\in{\mathcal{C}}_0$;
    \item[(c)] the addition of the set-valued map ${\mathcal{F}}_0^\circ \coloneqq f_0^\circ + G_0 N_0 \mathcal{U}$ and $G_0^\circ(-K_0\tilde{\xi}_0+\kappa_0^\star)$ is the superset of the right-hand side of \eqref{eqn:differential_inclusion_of_prioritized_closed_system_on_omega_0_eta} such that every Krasovskii solution of \eqref{eqn:prioritized_closed_loop_system_on_omega_0} is a solution of the differential inclusion
    \begin{subequations}\label{eqn:differential_inclusion_with_supper_set}
    \begin{align}
        \dot{\eta}_0 &\in \mathcal{F}_0^\circ + G_0^\circ(-K_0 \tilde{\xi}_0 + \kappa_0^\star) \label{eqn:differential_inclusion_with_supper_set_1} \\
        \dot{{\tilde{\xi}}}_i &= (A_i - B_iK_i){\tilde{\xi}}_i + B_i\sum_{j=1}^iE_{ij}(K_j{\tilde{\xi}}_j - \kappa_{j}^\star + \kappa_j) \label{eqn:differential_inclusion_with_supper_set_2}
    \end{align}
    \end{subequations}
    where \eqref{eqn:differential_inclusion_with_supper_set_1} and \eqref{eqn:differential_inclusion_with_supper_set_2} for $1\le i\le i_0$ hold on ${\mathcal{C}}_0$.
    $\hfill\square$
    \end{itemize} 
\end{lemma}

\begin{proof}
    See Appendix \ref{app:proof_of_lemma_about_superset_of_differential_inclusion}.
\end{proof}

We restrict ourselves to a certain class of minimum phase systems in a sense of the differential inclusion, as stated below.

\begin{assumption}\label{asm:zero_dynamics_lyapunov}
    There exist a positive constant $\bar{\sf r}_{\eta_0}$ and a Lipschitz function $V_{\eta_0}: \bar{\sf r}_{\eta_0}\mathcal{B}_{n-r_0}\to\mathbb{R}$ satisfying
    \begin{subequations}     \label{eqn:zero_dynamics_lyapunov}
    \begin{align}
        \bar{\alpha}_1(\|\eta_0\|) \le V_{\eta_0}(\eta_0) & \le \bar{\alpha}_2(\|\eta_0\|) \\
        D^+V_{\eta_0}(\eta_0,\tau_0) & \le -\bar{\alpha}_3(\|\eta_0\|)
    \end{align}
    \end{subequations}
    for all $\|\eta_0\|<\bar{\sf r}_{\eta_0}$ and $\tau_0\in \mathcal{F}_0^\circ [\eta_0,0]$ where $\bar{\alpha}_i$, $1\le i\le 3$,  are class $\mathcal{K}$ functions defined on $[0,\bar{\sf r}_{\eta_0}]$ and $D^+V_{\eta_0}(\eta_0,\tau_0)$ is the upper Dini derivative defined by 
    \begin{equation*}
        D^+V_{\eta_0}(\eta_0,\tau_0) = \limsup_{h\downarrow 0,\,\tau\to \tau_0}\frac{V_{\eta_0}(\eta_0 + h \tau) - V_{\eta_0}(\eta_0)}{h}.
    \end{equation*}
    $\hfill\square$
\end{assumption}
We remark that \eqref{eqn:zero_dynamics_lyapunov} implies that the equilibrium point $0\in\mathbb{R}^{n-r_0}$ of the differential inclusion 
\begin{equation}
    \dot{\eta}_0\in \mathcal{F}_0^\circ[\eta_0,0]
    \label{eqn:zero_dynamics_tracking}
\end{equation}
is asymptotically stable \cite[Theorem 8.2]{Smirnov2002}.
As similar to the case for differential equations, we call \eqref{eqn:zero_dynamics_tracking} the \textit{zero dynamics} of the differential inclusion \eqref{eqn:differential_inclusion_with_supper_set_1}--\eqref{eqn:differential_inclusion_with_supper_set_2}.

\begin{remark}
    Summarizing so far, we dealt with the discontinuity in \eqref{eqn:prioritized_closed_loop_system_on_omega_0} by employing its Krasovskii solution derived from the differential inclusion. 
    A consequence of this approach is that the Krasovskii solution of the differential inclusion does not satisfy \eqref{eqn:differential_inclusion_with_supper_set_2} almost everywhere for lower-priority tasks $\mathfrak{T}_{i_0+1},\dots,\mathfrak{T}_{k}$.
    Nonetheless, we will observe in the upcoming subsection that, at least up to $i_0$ higher-priority tasks can be achieved in a sense that the error dynamics converge near $0$ as time goes on.
    Therefore, we understand that the differential inclusion is a proper tool to handle discontinuity in {\bf Prioritized Control Problem}.

    On the other hand, the issue related to lower-priority tasks could be overcome with a Carath\`{e}odory solution of \eqref{eqn:prioritized_closed_loop_system_on_omega_0}, since it satisfies \eqref{eqn:differential_inclusion_with_supper_set_2} almost everywhere for all tasks.
    Nonetheless, it is not simple to derive the existence theorem of the Carath\`{e}odory solution, especially for a nonsmooth prioritized system because discontinuity occurs on both time and state (see \cite{Filippov1988} for details).
    Thus, we do not construct an existence theorem of a Carath\`{e}odory solution in this paper (our previous work in a simple case can be found in \cite{An2019a}).
    Instead, we will 
    focus on the result of the assumption that the Krasovskii solution satisfies \eqref{eqn:differential_inclusion_with_supper_set_2} almost everywhere for some of the lower-priority tasks.
    $\hfill\square$
\end{remark}

\subsection{Main Result}\label{subsec:main_result}

We present two theoretical results on the prioritized control framework.
The following theorem states that under certain conditions the overall system is ultimately bounded.

\begin{theorem}
    \label{thm:prioritized_output_tracking}
    Consider the prioritized closed-loop system \eqref{eqn:prioritized_closed_loop_system_on_omega_0} and suppose that all assumptions hold. 
    Then, for every sufficiently small $\xi_0^\star$, $\kappa_0^\star$, and $E_0$, there exist constants $0<\varepsilon<\delta$ ($\varepsilon$ can be chosen arbitrarily small by reducing upper bounds of $\xi_0^\star$, $\kappa_0^\star$, and $E_0$), class $\mathcal{K}$ functions $\alpha_1\le\alpha_2$, and a class $\mathcal{KL}$ function $\beta$ such that for every $\tilde{z}_0^\circ = (\eta_0^\circ, \tilde{\xi}_0^\circ)\in\mathbb{R}^n$ satisfying $\|\tilde{z}_0^\circ\|\le\alpha_2^{-1}(\alpha_1(\delta))$
    \begin{itemize}
        \item[(a)] a Krasovskii solution $z_0 = (\eta_0,\xi_0):[0,\infty)\to{\mathcal{C}}_0$ of \eqref{eqn:prioritized_closed_loop_system_on_omega_0} initiated at $z_0(0) = \tilde{z}_0^\circ + (0,\xi_{0}^\star(0))$ exists; and
        \item[(b)] there exists $T_0>0$ (dependent on $\tilde{z}_0^\circ$ and $\varepsilon$) such that every Krasovskii solution of \eqref{eqn:prioritized_closed_loop_system_on_omega_0} satisfies
        \begin{equation}
            \|\tilde{z}_0(t)\| \le \begin{dcases*} \beta(\|\tilde{z}_0^\circ\|,t), & $t\in[0,T_0]$ \\ \alpha_1^{-1}(\alpha_2(\varepsilon)), & $t\in[T_0,\infty)$. \end{dcases*}
            \label{eqn:ultimate_bound_of_eta_and_zeta_on_omega_0}
        \end{equation}
        $\hfill\square$
    \end{itemize}
\end{theorem}

\begin{proof}
    We will find a condition of ${\sf M}_{\xi_0^\star}$, ${\sf M}_{\kappa_0^\star}$, and ${\sf M}_{E_{ij}}$ for $1\le i,j\le i_0$ on which the statement of the theorem holds for all $\xi_0^\star$, $\kappa_0^\star$, and $E_0$ satisfying \eqref{eqn:upper_bound_of_reference_signals} and \eqref{eqn:bounds_E_kappa}.
    Let $0<{\sf r}_{\eta_0}\le \bar{\sf r}_{\eta_0}$ and ${\sf M}_{\xi_0^\star} < {\sf r}_{\tilde{\xi}_0} + {\sf M}_{\xi_0^\star} = {\sf r}_{\xi_0}$ be such that
    \begin{equation*}
        \tilde{\mathcal{E}}_0 \coloneqq {\sf r}_{\eta_0}\mathcal{B}_{n-r_0}\times {\sf r}_{\tilde{\xi}_0}\mathcal{B}_{r_0} \subset \mathcal{E}_0 \coloneqq {\sf r}_{\eta_0}\mathcal{B}_{n-r_0}\times {\sf r}_{\xi_0}\mathcal{B}_{r_0} \subset \mathcal{C}_0.
    \end{equation*}
    Note that $\mathrm{cl}(\mathcal{E}_0)$ is a compact subset of ${\mathcal{C}}_0$ and $\tilde{z}_0 = (\eta_0,\tilde{\xi}_0)\in\tilde{\mathcal{E}}_0$ implies $z_0=(\eta_0,\xi_0)\in\mathcal{E}_0$.
    Thus Lemma~\ref{lem:solution_existence_of_differential_inclusion} says that for every $\tilde{z}_0^\circ\in\tilde{\mathcal{E}}_0$ there exists a local Krasovskii solution $z_0:[0,T)\to\mathcal{E}_0$ of \eqref{eqn:prioritized_closed_loop_system_on_omega_0} initiated at $z_0(0)=\tilde{z}_0^\circ+(0,\xi_0^\star(0))$, with $T>0$, satisfying $\tilde{z}_0(t)\in\tilde{\mathcal{E}}_0$ for all $t\in[0,T)$ and either $T=\infty$ or $\tilde{z}_0(T)\in\mathrm{bd}(\tilde{\mathcal{E}}_0)$.
    We propose a Lyapunov function candidate $V:\tilde{\mathcal{E}}_0\to\mathbb{R}$ for the entire system 
    \eqref{eqn:differential_inclusion_with_supper_set} as 
    \begin{equation}\label{eqn:Lyapunov_entire}
        V(\eta_0,\tilde{\xi}_0) = V_{\eta_0}(\eta_0) + 	V_{\tilde{\xi}_0}(\tilde{\xi}_0)
    \end{equation}
    where $V_{\eta_0}$ is defined in Assumption~\ref{asm:zero_dynamics_lyapunov}, and $V_{\tilde{\xi}_0}$ is given in \eqref{eqn:Lyapunov_xi}. 
    Let $\delta_0 = \min\{{\sf r}_{\eta_0},\,{\sf r}_{\tilde{\xi}_0}\}$.
    Note that, since $V$ in \eqref{eqn:Lyapunov_entire} is continuous and positive definite, there exist class $\mathcal{K}$ functions $\alpha_1$ and $\alpha_2$ defined on $[0,\delta_0]$  such that
    \begin{equation*}
        \alpha_1(\|(\eta_0,{\tilde{\xi}}_0)\|) \le V(\eta_0,{\tilde{\xi}}_0) \le \alpha_2(\|(\eta_0,{\tilde{\xi}}_0)\|)
    \end{equation*}
    for all $(\eta_0,{\tilde{\xi}}_0)\in \delta_0 \mathcal{B}_n\subset\tilde{\mathcal{E}}_0$ \cite[Lemma 4.3]{Khalil2015}.

    We now compute the time derivative $\dot{V}_{\eta_0}$ of $V_{\eta_0}$ in \eqref{eqn:Lyapunov_entire}.
    It is noted first that $V_{\eta_0}(\eta_0)$ is Lipschitz on $\delta_0 \mathcal{B}_{n-r_0}$ with a Lipschitz constant ${\sf L}_{V_{\eta_0}}>0$ and $\eta_0(t)$ is absolutely continuous on $[0,T)$.  
    Thus, the mapping $t\mapsto V_{\eta_0}(\eta_0(t))$ is absolutely continuous and  differentiable almost everywhere on $[0,T)$.
    It is then readily seen that $\dot{V}_{\eta_0}(\eta_0(t))\le D^+V_{\eta_0}(\eta_0(t),\dot{\eta}_0(t))$ almost everywhere.%
\footnote{
    Since $\eta_0(t+h) = \eta_0(t) + \int_t^{t+h}\dot{\eta}_0(s)ds$ for all $t$ \cite[Theorem 7.20]{Rudin1987} and $\lim_{h\to0}\frac{1}{h}\int_t^{t+h}\dot{\eta}_0(s)ds = \dot{\eta}_0(t)$ for almost all $t$ \cite[Theorem 7.11]{Rudin1987}, the upper bound of $\dot{V}_{\eta_0}$ can be computed as
    \begin{align*}
        \dot{V}_{\eta_0}(\eta_0(t)) &= \lim_{h\to0}\frac{V_{\eta_0}(\eta_0(t+h)) - V_{\eta_0}(\eta_0(t))}{h}
            = \lim_{h\to0}\frac{1}{h}\bigg[V_{\eta_0}\bigg(\eta_0(t) + h\frac{1}{h}\int_t^{t+h}\dot{\eta}_0(s)ds\bigg) - V_{\eta_0}(\eta_0(t))\bigg] \\
            &\le \limsup_{h\downarrow0,\tau\to\dot{\eta}_0(t)}\frac{V_{\eta_0}(\eta_0(t) + h\tau) - V_{\eta_0}(\eta_0(t))}{h}
            = D^+V_{\eta_0}(\eta_0(t),\dot{\eta}_0(t))
    \end{align*}
    almost everywhere.
}
    On the other hand, by Lemma \ref{lem:superset_of_differential_inclusion}, there exists a constant ${\sf L}_{\mathcal{F}_0^\circ}>0$ satisfying%
\footnote{
    Since ${\cal U}[z_0]$ is upper semicontinuous on a compact set ${\cal C}_0$ and each ${\cal U}[z_0]\subset {\cal B}_m$ is compact, there exists a positive constant ${\sf M}_{\cal U}$ such that ${\cal U}[z_0] \subset {\sf M}_{\cal U}{\cal B}_m$ for all $z_0\in{\cal C}_0$.
    To see this, observe that for each $z_0\in{\cal C}_0$ there exists a compact set ${\cal M}_{z_0}\subset\mathbb{R}^m$ and a neighborhood ${\cal N}_{z_0}\subset\mathbb{R}^n$ of $z_0$ satisfying ${\cal U}[{\cal N}_{z_0}]\subset {\cal M}_{z_0}$.
    Since ${\cal M}_{z_0}$ is compact, for each $z_0\in{\cal C}_0$ there exists a positive constant ${\sf M}_{z_0}$ satisfying ${\cal M}_{z_0}\subset {\sf M}_{z_0}{\cal B}_m$.
    Since $\{{\cal N}_{z_0} : z_0\in{\cal C}_0\}$ is an open cover of the compact set ${\cal C}_0$, there exists a finite subcover $\{{\cal N}_1,\dots,{\cal N}_N\}$ along with the positive constants ${\sf M}_1,\dots,{\sf M}_N$ satisfying ${\cal C}_0\subset\bigcup_{i=1}^N{\cal N}_i$ and ${\cal U}[{\cal N}_i]\subset {\sf M}_i{\cal B}_m$ for all $1\le i\le N$.
    By letting ${\sf M}_{\cal U} = \max\{{\sf M}_1,\dots,{\sf M}_N\}$, we have ${\cal U}[{\cal C}_0]\subset{\cal U}[\bigcup_{i=1}^N{\cal N}_i]\subset{\sf M}_{\cal U}{\cal B}_m$.
    Then, the supperset of ${\cal F}_0^\circ[\eta_0,\xi_0]$ can be found as
    \begin{align*}
        {\cal F}_0^\circ[z_0]
            &\subset f_0^\circ[\eta_0,0] + {\sf L}_{f_0^\circ}\|\xi_0\|{\rm cl}({\cal B}_{n-r_0}) + (G_0N_0)[z_0]\big({\cal U}[\eta_0,0] + L_{\cal U}\|\xi_0\|{\rm cl}({\rm B}_m)\big) \\
            &\subset {\cal F}_0^\circ[\eta_0,0] + \big((G_0N_0)[z_0] - (G_0N_0)[\eta_0,0]\big){\cal U}[\eta_0,0] + ({\sf L}_{f_0^\circ} + {\sf M}_{G_0N_0}{\sf L}_{\cal U})\|\xi_0\|{\rm cl}({\cal B}_{n-r_0}) \\
            &\subset {\cal F}_0^\circ[\eta_0,0] + {\sf M}_{\cal U}\|(G_0N_0)[z_0] - (G_0N_0)[\eta_0,0]\|{\rm cl}({\cal B}_{n-r_0}) + ({\sf L}_{f_0^\circ} + {\sf M}_{G_0N_0}{\sf L}_{\cal U})\|\xi_0\|{\rm cl}({\cal B}_{n-r_0}) \\
            &\subset {\cal F}_0^\circ[\eta_0,0] + ({\sf L}_{f_0^\circ} + {\sf M}_{G_0N_0}{\sf L}_{\cal U} + {\sf M}_{\cal U}{\sf L}_{G_0N_0})\|\xi_0\|{\rm cl}({\cal B}_{n-r_0}).
    \end{align*}
}
    \begin{equation*}
        \mathcal{F}_0^\circ[\eta_0,\xi_0] \subset \mathcal{F}_0^\circ[\eta_0,0] + {\sf L}_{\mathcal{F}_0^\circ}\|\xi_0\|\mathrm{cl}(\mathcal{B}_{n-r_0})
    \end{equation*}
    for all $(\eta_0,\xi_0)\in{\mathcal{C}}_0$.
    It means that for every $t\in[0,T)$ at which  \eqref{eqn:differential_inclusion_of_prioritized_closed_system_on_omega_0} holds and for 
    $$\tau_3 \coloneqq G_0^\circ[z_0(t)](-K_0{\tilde{\xi}}_0(t) + \kappa_0^\star(t)),$$ the time derivative $\dot\eta_0(t)$ of $\eta_0(t)$ in \eqref{eqn:differential_inclusion_with_supper_set_1} satisfies
    \begin{equation*}
        \dot{\eta}_0(t) = \tau_1 + \tau_2 + \tau_3
    \end{equation*}
    where $\tau_1$ and $\tau_2$ are such that
    \begin{equation*}
        \tau_1\in \mathcal{F}_0^\circ[\eta_0(t),0],\quad \tau_2\in {\sf L}_{{\mathcal{F}}_0^\circ}\|\xi_0(t)\|\mathrm{cl}(\mathcal{B}_{n-r_0}).
    \end{equation*} 
    Since $\tau_1' \to \tau_1$ if and only if  $\tau_1' + \tau_2 + \tau_3 \to\dot{\eta}_0(t)$, we have%
\footnote{
    The upper bound of $\dot{V}_{\eta_0}$ can be found as
    \begin{align*}
        \dot{V}_{\eta_0} &\le D^+V_{\eta_0}(\eta_0(t),\dot{\eta}_0(t)) 
            = \limsup_{h\downarrow0,\tau\to\dot{\eta_0}(t)}\frac{V_{\eta_0}(\eta_0(t) + h\tau) - V_{\eta_0}(\eta_0(t))}{h}
            = \limsup_{h\downarrow0,\tau\to\tau_1+\tau_2+\tau_3}\frac{V_{\eta_0}(\eta_0(t) + h\tau) - V_{\eta_0}(\eta_0(t))}{h} \\
            &= \limsup_{h\downarrow0,\tau-\tau_2-\tau_3\to\tau_1}\frac{V_{\eta_0}(\eta_0(t) + h\tau) - V_{\eta_0}(\eta_0(t))}{h}
            = \limsup_{h\downarrow0,\tau_1'\to\tau_1}\frac{V_{\eta_0}(\eta_0(t) + h(\tau_1'+\tau_2+\tau_3)) - V_{\eta_0}(\eta_0(t))}{h} \\
            &\le \limsup_{h\downarrow0,\tau_1'\to\tau_1}\frac{V_{\eta_0}(\eta_0(t) + h(\tau_1'+\tau_2+\tau_3)) - V_{\eta_0}(\eta_0(t)+h\tau_1')}{h} + \limsup_{h\downarrow0,\tau_1'\to\tau_1}\frac{V_{\eta_0}(\eta_0(t) + h\tau_1') - V_{\eta_0}(\eta_0(t))}{h} \\
            &\le \limsup_{h\downarrow0,\tau_1'\to\tau_1}\frac{{\sf L}_{V_{\eta_0}}\|h(\tau_2+\tau_3)\|}{h} + D^+V_{\eta_0}(\eta_0(t),\tau_1)
            = D^+V_{\eta_0}(\eta_0(t),\tau_1) + {\sf L}_{V_{\eta_0}}(\|\tau_2\| + \|\tau_3\|) \\
            &\le -\bar{\alpha}_3(\|\eta_0(t)\|) + {\sf L}_{V_{\eta_0}}\big({\sf L}_{{\cal F}_0^\circ}(\|\tilde{\xi}_0\| + {\sf M}_{\xi_0^\star}) + {\sf M}_{G_0^\circ}(\|K_0\|\|\tilde{\xi}_0\| + {\sf M}_{\kappa_0^\star})\big) \\
            &= -\bar{\alpha}_3(\|\eta_0(t)\|) + \underbrace{{\sf L}_{V_{\eta_0}}({\sf L}_{{\cal F}_0^\circ} + {\sf M}_{G_0^\circ}\|K_0\|)}_{{\sf m}_1}\|\tilde{\xi}_0\| + \underbrace{{\sf L}_{V_{\eta_0}}{\sf L}_{{\cal F}_0^\circ}}_{{\sf m}_2}{\sf M}_{\xi_0^\star} + \underbrace{{\sf L}_{V_{\eta_0}}{\sf M}_{G_0^\circ}}_{{\sf m}_3}{\sf M}_{\kappa_0^\star}
    \end{align*}
}
    \begin{align}
        \dot{V}_{\eta_0}
        & \le D^+V_{\eta_0}(\eta_0(t),\dot{\eta}_0(t)) \nonumber\\
        &= \limsup_{h\downarrow0,\,\tau_1'\to\tau_1}\frac{V_{\eta_0}(\eta_0(t)+h(\tau_1'+\tau_2 + \tau_3)) - V_{\eta_0}(\eta_0(t))}{h} \nonumber\\
        &\le D^+V_{\eta_0}(\eta_0(t),\tau_1) + {\sf L}_{V_{\eta_0}}(\|\tau_2\| + \|\tau_3\|) \nonumber\\
        &\le -\bar{\alpha}_3(\|\eta_0(t)\|) + {\sf m}_1  \|{\tilde{\xi}}_0(t)\| + {\sf m}_2{\sf M}_{\xi_0^\star} + {\sf m}_3{\sf M}_{\kappa_0^\star} \label{eqn:upper_bound_of_derivative_of_V_eta_0}
    \end{align}
    where 
    \begin{align*}
        {\sf m}_1 = {\sf L}_{V_{\eta_0}}({\sf L}_{{\mathcal{F}}_0^\circ} + {\sf M}_{G_0^\circ}\|K_0\|),~{\sf m}_2 = {\sf L}_{V_{\eta_0}}{\sf L}_{{\cal F}_0^\circ},~ {\sf m}_3 = {\sf L}_{V_{\eta_0}}{\sf M}_{{G}_0^\circ}.
    \end{align*}
    
    We are ready to find an upper bound of the time derivative of the Lyapunov function candidate $V$.
    By Assumption \ref{asm:M_matrix_condition_with_spectral_radius}, there exists ${\sf w}^\star\in(0,\infty)^{i_0}$ satisfying ${\sf v}^\star \coloneqq ({\sf Y}-{\sf Z}){\sf w}^\star\in(0,\infty)^{i_0}$.
    Let ${\sf v}_0^\star = \min\{{\sf v}_i^\star\}$, ${\sf w}_0>0$, ${\sf w} = {\sf w}_0{\sf w}^\star$, ${\sf v} = ({\sf Y}-{\sf Z}){\sf w}$, and ${\sf v}_0 = \min\{{\sf v}_i\}$.
    Obviously, ${\sf w}, ({\sf Y}-{\sf Z}){\sf w}\in(0,\infty)^{i_0}$ and ${\sf v}_0 = {\sf w}_0{\sf v}_0^\star$ for all ${\sf w}_0>0$.
    Combining \eqref{eqn:upper_bound_of_derivative_of_V_eta_0} with \eqref{eqn:upper_bound_of_V_xi}, we have
    \begin{align*}
        \dot{V} &\le -\bar{\alpha}_3(\|\eta_0\|) + \bigg({\sf w_0}\sum_{i=1}^{i_0}\frac{{\sf w}_i^\star\tilde{\sigma}_i}{\vartheta_i}\sum_{j=1}^i{\sf M}_{E_{ij}}{\sf L}_{\kappa_j}\bigg)\|\eta_0\| \\
            &\quad~ - ({\sf w}_0{\sf v}_0^\star - {\sf m}_1)\|\tilde{\xi}_0\| \\
            &\quad~ + \bigg({\sf w}_0\sum_{i=1}^{i_0}\frac{{\sf w}_i^\star\tilde{\sigma}_i}{\vartheta_i}\sum_{j=1}^i{\sf M}_{E_{ij}}{\sf L}_{\kappa_j} + {\sf m}_2\bigg){\sf M}_{\xi_0^\star} \\
            &\quad~ + \bigg({\sf w}_0\sum_{i=1}^{i_0}\frac{{\sf w}_i^\star\tilde{\sigma}_i}{\vartheta_i}\sum_{j=1}^i{\sf M}_{E_{ij}} + {\sf m}_3\bigg){\sf M}_{\kappa_0^\star}
    \end{align*}
    for all ${\sf w}_0>0$ and almost all $t\in[0,T)$.
    Finally, for sufficiently large ${\sf w}_0>0$ and sufficiently small ${\sf M}_{E_{ij}}$,  ${\sf M}_{\kappa_0^\star}$, and ${\sf M}_{\xi_0^\star}$, one can always find ${\sf c}_1, {\sf c}_2>0$, $0 < \delta <\delta_0$, and $0<\varepsilon<\alpha_2^{-1}(\alpha_1(\delta))\le \delta$ satisfying
    \begin{equation*}
        \dot{V} \le -{\sf c}_1\bar{\alpha}_3(\|\eta_0\|) - {\sf c}_2\|\tilde{\xi}_0\| \quad (\varepsilon\le\|(\eta_0,\tilde{\xi}_0)\|\le\delta)
    \end{equation*}
    almost everywhere.
    A similar computation with \cite[Theorem 4.18]{Khalil2015}%
    \footnote{
        As a matter of fact, \cite[Theorem 4.18]{Khalil2015} requires $V(\eta_0,{\tilde{\xi}}_0)$ to be continuously differentiable, which may be violated at some $t$.  
        Yet this condition can be easily adapted for this proof by using the fundamental theorem of calculus \cite[Theorem 7.20]{Rudin1987}.
    }
    gives a class $\mathcal{KL}$ function $\beta$, so that for every $\|(\eta_0(0),{\tilde{\xi}}_0(0))\|<\alpha_2^{-1}(\alpha_1(\delta))$, there is $T_0>0$ (dependent on $(\eta_0(0),{\tilde{\xi}}_0(0))$ and $\varepsilon$) such that every Krasovskii solution of \eqref{eqn:prioritized_closed_loop_system_on_omega_0} satisfies \eqref{eqn:ultimate_bound_of_eta_and_zeta_on_omega_0}.%
\footnote{
    \label{fnt:find_ultimate_bound}
    See Appendix \ref{app:proof_of_footnote_about_find_ultimate_bound} for the proof.
}
    Obviously, \eqref{eqn:ultimate_bound_of_eta_and_zeta_on_omega_0} implies $T=\infty$.
\end{proof}

A direct consequence of the theorem is that, for all $t\in [T_0,\infty)$, $\kappa_i[z_0]$ are uniformly bounded on $\mathcal{C}_0$ as
\begin{align*}
    \|\kappa_i[z_0]\| & \leq  {\sf L}_{\kappa_i} \| z_0\| \leq {\sf L}_{\kappa_i} \alpha_1^{-1}(\alpha_2(\varepsilon)).
\end{align*}
Thus, the analysis of the $i$-th error dynamics \eqref{eqn:differential_inclusion_with_supper_set_2} can be decoupled from unpredictable behavior of the lower-priority tasks ${\frak T}_{i+1},\dots,{\frak T}_k$.
This leads to a different result on finding the upper bound of the tracking errors $\tilde{\xi}_i(t)$.

\begin{corollary}
    \label{cor:ultimate_bound_of_tracking_errors}
    Assume that all the hypothesis of Theorem~\ref{thm:prioritized_output_tracking} hold, $\|\tilde{z}_0^\circ\|\le\alpha_2^{-1}(\alpha_1(\delta))$, and there exist $i_0\le i_1 \le k$ and $T_1\in[T_0,\infty)$ such that for each $i_0 +1 \leq i\le i_1$,
    \begin{enumerate}
        \item a Krasovskii solution $z_0(t)$ of \eqref{eqn:prioritized_closed_loop_system_on_omega_0} with $z_0(0) = \tilde{z}_0^0 + (0,\xi_{0}^\star(0))$ satisfies  \eqref{eqn:prioritized_closed_loop_system_on_omega_0_xi} almost everywhere on $[T_1,\infty)$;
        \item there exist positive constants $\varsigma_i$ and $\vartheta_i$ and matrices $K_i$, $X_i=X_i^T>0$, and $R_i$ satisfying \eqref{eqn:diagonal_block_positive_definite} almost everywhere on $[T_1,\infty)$ along the trajectory $z_0(t)$ and \eqref{eqn:Kalman_Yakubovichi_Popov_lemma}.
    \end{enumerate}
    Then, there exist ${\sf a}_{ij},{\sf b}_{ij},{\sf c}_i>0$ for $1\le j\le i\le i_1$ such that
    \begin{align}\label{eqn:upper_bound_of_zeta_i}
        \|{\tilde{\xi}}_i(t)\| & \le \sum_{j=1}^i {\sf a}_{ij}\|{\tilde{\xi}}_j(T_1)\|e^{-{\sf b}_{ij}(t-T_1)} + {\sf c}_i
    \end{align}
    for all $1\le i\le i_1$ and $t\in[T_1,\infty)$.
    Additionally, if there exists $1\le i\le i_1$ such that $E_{ij}[z_0(t)] = 0$ for all $1\le j\le i$ and all $t\in[T_1,\infty)$, then \eqref{eqn:upper_bound_of_zeta_i} holds with ${\sf c}_i = 0$.
    $\hfill\square$
\end{corollary}

\begin{proof}
    Define ${\sf M}_{E_{ij}}^\circ \coloneqq \sup\{\|E_{ij}[z_0(t)]\| : t\in[T_1,\infty)\}$.
    With Item~2) of the corollary, we take $V_{\tilde{\xi}_i}({\tilde{\xi}}_i)$ for the lower-priority tasks $i_0 +1\leq  i\le i_1$ as in \eqref{eqn:Lyapunov_function_i}.
    From \eqref{eqn:ultimate_bound_of_eta_and_zeta_on_omega_0}, Gr\"onwall's inequality \cite[Lemma 2]{An2019a}, and the fact that $\|\tilde{\xi}_i\|\le V_{\tilde{\xi}_i}(\tilde{\xi}_i)/\sigma_{\min}^{1/2}(X_i)$, we can find an upper bound of $\|\tilde{\xi}_i(t)\|$ as
    \begin{subequations}\label{eqn:upper_bound_of_norm_of_tracking_error}
        \begin{align}
            \|\tilde{\xi}_i(t)\| &\le \sqrt{\tilde{\sigma}_i}\|\tilde{\xi}_i(T_1)\|e^{-\vartheta_i(t-T_1)} \label{eqn:upper_bound_of_norm_of_tracking_error_1} \\
                &\quad + \tilde{\sigma}_i\sum_{j=1}^{i-1}{\sf M}_{E_{ij}}^\circ\|K_j\|\int_{T_1}^t\|\tilde{\xi}_j(s)\|e^{-\vartheta_i(t-s)}ds \label{eqn:upper_bound_of_norm_of_tracking_error_2} \\
                &\quad + \frac{\tilde{\sigma}_i}{\vartheta_i}\sum_{j=1}^i{\sf M}_{E_{ij}}^\circ\big({\sf L}_{\kappa_j}(\alpha_1^{-1}(\alpha_2(\varepsilon)) + {\sf M}_{\xi_0^\star}) + {\sf M}_{\kappa_j^\star}\big) \label{eqn:upper_bound_of_norm_of_tracking_error_3}
        \end{align}
    \end{subequations}
    for all $1\le i\le i_1$ and all $t\in[T_1,\infty)$.
    The proof of the remaining can be done by showing \eqref{eqn:upper_bound_of_zeta_i} via mathematical induction: its derivation is straightforward and thus is omitted here.%
\footnote{
    See Appendix \ref{app:proof_of_the_omitted_part_of_ultimate_bound_of_tracking_errors} for the omitted part of the proof.
}
    Instead, we note that  \eqref{eqn:upper_bound_of_zeta_i} holds for $i=1$ with ${\sf a}_{11} = \sqrt{\tilde{\sigma}_1}$, ${\sf b}_{11} = \vartheta_1$, and ${\sf c}_1$ equals \eqref{eqn:upper_bound_of_norm_of_tracking_error_3}.
    Finally, if $E_{ij}[z_0(t)] = 0$ for all $1\le j\le i$ and $t\in[T_1,\infty)$, then ${\sf c}_i = 0$ because \eqref{eqn:upper_bound_of_norm_of_tracking_error_2} and \eqref{eqn:upper_bound_of_norm_of_tracking_error_3} vanish.
\end{proof}

\begin{remark}
    Although the set-valued map $\mathcal{U}$ used in the definition of the zero dynamics \eqref{eqn:zero_dynamics_tracking} always exists, it is not guaranteed that there exists a Lipschitz function $V_{\eta_0}(\eta_0)$ satisfying \eqref{eqn:zero_dynamics_lyapunov}.
    As a matter of fact, existence of such $V_{\eta_0}$ is directly related to $\mathcal{U}$.
    To see this, let $u_{i_0+1:k+1}[t,z_0] = 0$ for all $(t,z_0)\in[0,\infty)\times{\mathcal{C}}_0$, which is obtained by selecting $\Gamma_{ij}[z_0] = 0$ and $u_f[t,z_0] = 0$ for all $i_0<i\le k$, $1\le j\le i$, and $z_0\in{\mathcal{C}}_0$.
    It is then possible to choose $\mathcal{U}$ as $\mathcal{U}[z_0]=\alpha\mathrm{cl}(\mathcal{B}_m)$ for all $z_0\in{\mathcal{C}}_0$ with an arbitrarily small constant $\alpha>0$.
    It follows that Assumption~\ref{asm:zero_dynamics_lyapunov} can be reduced to the asymptotic stability of the equilibrium point 0 of the differential equation $\dot{\eta}_0 = f_0^\circ[\eta_0,0]$.
    Notice that, as we increase the size of $\mathcal{U}[z_0]$, the margin of stability of the zero dynamics $\dot{\eta}_0 = {\mathcal{F}}_0^\circ[\eta_0,0]$ decreases, which implies that there could be a maximum size of $\mathcal{U}[z_0]$ that allows \eqref{eqn:zero_dynamics_lyapunov}.

    From this perspective, Theorem \ref{thm:prioritized_output_tracking} and Corollary \ref{cor:ultimate_bound_of_tracking_errors} suggest a design guideline for the proposed control framework to be utilized in practical applications.
    First, we design $u_{1:i_0}$ and find a set-valued map $\mathcal{U}$ that allows \eqref{eqn:zero_dynamics_lyapunov}.
    Then, one can design $u_{i_0+1:k+1}$ such that ${\mathcal{U}}_{i_0+1:k+1}[t,z_0]\subset \mathcal{U}[z_0]$ on $[0,\infty)\times{\mathcal{C}}_0$.
    It is important to note that not every reference output signal $y_{i}^\star$ for $i_0<i\le k$ is trackable by the output $y_i$ because of singularity.
    Thus, one of the best scenarios is that there are $i_0< i_1\le k$ and $T_1\ge T_0$ such that the assumptions on Corollary \ref{cor:ultimate_bound_of_tracking_errors} hold.
    To match this condition, we need to try various $\Gamma_{ij}$ and $y_{i}^\star$ for $i_0<i\le k$ under the condition ${\mathcal{U}}_{i_0+1:k+1}[t,z_0]\subset {\mathcal{U}}[z_0]$.
    This clearly shows that the prioritized output tracking provides a framework in which we can try the output tracking of the lower priority tasks ${\mathfrak
    {T}}_{i_0+1},\dots,{\mathfrak{T}}_k$ that can be singular, while guaranteeing the output tracking of the higher priority tasks ${\mathfrak{T}}_1,\dots,{\mathfrak{T}}_{i_0}$.
    $\hfill\square$
\end{remark}

\section{Conclusion}
\label{sec:conclusion}

In this article, we proposed an extension of the input-output linearization based on the task priority.
The main contribution of this work is two folds.
Firstly, it allows us to apply well-developed theories of the input-output linearization to the input affine multivariate system that is not input-output linearizable.
Secondly, it provides a framework to construct a complex control system that has both critical and noncritical control issues such that the critical part can be handled by higher priority tasks and the noncritical part by lower priority tasks.
As a result, we can apply the input-output linearization to more complicated practical problems.

As the first attempt to extend the input-output linearization to the nonlinearizable systems based on the task priority, there are some unsolved problems in this article:
1) the task priority-based approach can be extended further by using the nonlinear structure algorithm, 
2) nonsmoothness caused by priority has not been understood sufficiently, 
3) existence of a trajectory that satisfies nonsmooth error dynamics has not been studied well, 
4) an efficient algorithm to find $\Gamma_{ij}$ and $y_i^\star$ for $i_0<i\le k$ under the condition $\mathcal{U}_{i_0+1:k+1}[t,z_0]\subset \mathcal{U}[z_0]$ needs to be developed, 
5) the prioritized output tracking problem might be solved with the prioritized control input which is not in the form of \eqref{eqn:prioritized_control_input_with_gamma} and \eqref{eqn:external_reference_input}, and 
6) the developed theorems need to be applied for practical problems.
We leave those as future works.

\appendix

\section{Proof of Lemma \ref{lem:diffeomorphism}}
\label{app:proof_of_lemma_about_diffeomorphism}

We provide a summary of \cite[Lemma 5.1.1, Proposition 5.1.2, Remark 5.1.3]{Isidori1995} for those who are not used to the input-output feedback linearization of multivariate nonlinear systems.
We first prove that the row vectors
\begin{align*}
    &dh_1(x_0),dL_fh_1(x_0),\dots,dL_f^{{r}_1-1}h_1(x_0),\\
    &dh_2(x_0),dL_fh_2(x_0),\dots,dL_f^{{r}_2-1}h_2(x_0),\\
    &\ddots \\
    &dh_{p}(x_0),dL_fh_{p}(x_0),\dots,dL_f^{{r}_{p}-1}h_{p}(x_0)
\end{align*}
are linearly independent.
Then, it will also prove ${r}\le n$.
Assume, without loss of generality, ${r}_1\ge {r}_i$ for $2\le i\le {p}$.
Consider the matrices
\begin{align*}
    X_i &= \begin{bmatrix} dh_i \\ \vdots \\ dL_f^{{r}_i-1}h_i \end{bmatrix}(x_0) \in \mathbb{R}^{{r}_i\times n} &
    X &= \begin{bmatrix} X_1 \\ \vdots \\ X_{p} \end{bmatrix} = \begin{bmatrix} dh_1 \\ \vdots \\ dL_f^{{r}_1-1}h_1 \\ \vdots \\ dh_{p} \\ \vdots \\ dL_f^{{r}_{p}-1}h_{p} \end{bmatrix}(x_0) \in \mathbb{R}^{{r}\times n}
\end{align*}
and
\begin{align*}
    Y = \begin{bmatrix} g_1 & \cdots & g_{{m}} & \cdots & ad_f^{{r}_1-1}g_1 & \cdots & ad_f^{{r}_1-1}g_{{m}} \end{bmatrix}(x_0) \in \mathbb{R}^{n\times {m}{r}_1}.
\end{align*}
We need the next lemma:

\begin{lemma}
    \label{lem:recursive_relation_between_lie_derivative_and_adjoint}
    \cite[Lemma 4.1.2]{Isidori1995}
    Let $\phi$ be a real-valued function and $f,g$ vector fields, all defined in an open set $U$ of $\mathbb{R}^n$. Then, for any choice of $s,j,r\ge0$,
    \begin{equation*}
        L_{ad_f^{j+r}g}L_f^s\phi(x) = \sum_{i=0}^r(-1)^i\binom{ r }{i}L_f^{r-i}L_{ad_f^jg}L_f^{s+i}\phi(x).
    \end{equation*}
    As a consequence, the two sets of conditions
    \begin{enumerate}
        \item $L_g\phi(x) = L_gL_f\phi(x) = \cdots = L_gL_f^j\phi(x) = 0,\,\forall x\in U$
        \item $L_g\phi(x) = L_{ad_fg}\phi(x) = \cdots = L_{ad_f^j}\phi(x) = 0,\,\forall x\in U$
    \end{enumerate}
    are equivalent.
    $\hfill\square$
\end{lemma}

From Lemma \ref{lem:recursive_relation_between_lie_derivative_and_adjoint} and the relative degree, we have
\begin{equation*}
    L_{ad_f^ag_j}L_f^bh_i(x) = \begin{dcases*} 0, & $0\le a+b<{r}_i-1$ \\ (-1)^aL_{g_j}L_f^{{r}_i-1}h_i(x), & $a+b={r}_i-1$ \end{dcases*}
\end{equation*}
for all $i,j,a,b$.
Then, we can formulate%
\footnote{
    \begin{align*}
        X_iY &= \begin{bmatrix} dh_i(x) \\ \vdots \\ dL_f^{{r}_1-1}h_i(x) \end{bmatrix}\begin{bmatrix} g_1(x) & \cdots & g_{{m}}(x) & \cdots & ad_f^{{r}_1-1}g_1(x) & \cdots & ad_f^{{r}_1-1}g_{{m}}(x) \end{bmatrix} \\
            &= \begin{bmatrix} 
                dh_i(x)g_1(x) & \cdots & dh_i(x)g_{{m}}(x) & \cdots & dh_i(x)ad_f^{{r}_1-1}g_1(x) & \cdots & dh_i(x)ad_f^{{r}_1-1}g_{{m}}(x) \\
                \vdots & \ddots & \vdots & \ddots & \vdots &\ddots &\vdots \\
                dL_f^{{r}_1-1}h_i(x)g_1(x) & \cdots & dL_f^{{r}_1-1}h_i(x)g_{{m}}(x) & \cdots & dL_f^{{r}_1-1}h_i(x)ad_f^{{r}_1-1}g_1(x) & \cdots & dL_f^{{r}_1-1}h_i(x)ad_f^{{r}_1-1}g_{{m}}(x)
            \end{bmatrix} \\
            &= \begin{bmatrix} 
                L_{g_1}h_i(x) & \cdots & L_{g_{{m}}}h_i(x) & \cdots & L_{ad_f^{{r}_1-1}g_1}h_i(x) & \cdots & L_{ad_f^{{r}_1-1}g_{{m}}}h_i(x) \\
                \vdots & \ddots & \vdots & \ddots & \vdots &\ddots &\vdots \\
                L_{g_1}L_f^{{r}_1-1}h_i(x) & \cdots & L_{g_{{m}}}L_f^{{r}_1-1}h_i(x) & \cdots & L_{ad_f^{{r}_1-1}g_1}L_f^{{r}_1-1}h_i(x) & \cdots & L_{ad_f^{{r}_1-1}g_{{m}}}L_f^{{r}_1-1}h_i(x)
            \end{bmatrix} \\
            &= \begin{bmatrix} 
                0 & \cdots & 0 & \cdots & (-1)^{{r}_i-1}L_{g_1}L_f^{{r}_i-1}h_i(x) & \cdots & (-1)^{{r}_i-1}L_{g_{{m}}}L_f^{{r}_i-1}h_i(x) & * & \cdots & * \\
                \vdots & \ddots & \vdots & \ddots & \vdots &\ddots &\vdots & \vdots & \ddots & \vdots \\
                L_{g_1}L_f^{{r}_i-1}h_i(x) & \cdots & L_{g_{{m}}}L_f^{{r}_i-1}h_i(x) & \cdots & L_{ad_f^{{r}_i-1}g_1}L_f^{{r}_i-1}h_i(x) & \cdots & L_{ad_f^{{r}_i-1}g_{{m}}}L_f^{{r}_i-1}h_i(x) & * & \cdots & *
            \end{bmatrix} \\
            &= \begin{bmatrix} 
                0 & 0 & \cdots & 0 & (-1)^{{r}_i-1}L_GL_f^{{r}_i-1}h_i(x) & * & \cdots & * \\
                0 & 0 & \cdots & (-1)^{{r}_i-2}L_GL_f^{{r}_i-1}h_i(x) & * & * & \cdots & * \\
                \vdots & \vdots & \ddots & \vdots & \vdots & \vdots & \ddots & \vdots \\
                0 & (-1)L_GL_f^{{r}_i-1}h_i(x) & \cdots & * & * & * & \cdots & * \\
                L_GL_f^{{r}_i-1}h_i(x) & * & \cdots & * & * & * & \cdots & *
            \end{bmatrix}.
    \end{align*}
}
\begin{equation*}
    X_iY = \begin{bmatrix} 
        0 & 0 & \cdots & 0 & (-1)^{{r}_i-1}j_i(x) & * \\
        \vdots & \vdots & \ddots & \vdots & \vdots & \vdots \\
        0 & -j_i(x) & \cdots & * & * & * \\
        j_i(x) & * & \cdots & * & * & *
    \end{bmatrix}
\end{equation*}
where $j_i$ is the $i$-th row of $J(x)$.
Thus, it is easy to see that the matrix $XY$, after possibly a reordering of the rows, exhibits a block triangular structure in which the diagonal blocks consist of rows of the matrix \eqref{eqn:input_gain_matrix}.
This shows the linear independence of the rows of the matrix $XY$, i.e., that of the rows of the matrix $X$.

Since
\begin{equation*}
    J(x) = \begin{bmatrix} dL_f^{{r}_1-1}h_1(x) \\ \vdots \\ dL_f^{{r}_{p}-1}h_{p}(x) \end{bmatrix}G(x) \in \mathbb{R}^{{p}\times {m}}
\end{equation*}
is nonsingular at $x_0$ by the definition of the relative degree, there exists $g\in\{g_1,\dots,g_{{m}}\}$ such that the vectors $g(x_0)$ and
\begin{equation*}
    \begin{bmatrix} dL_f^{{r}_1-1}h_1(x_0) \\ \vdots \\ dL_f^{{r}_{p}-1}h_{p}(x_0) \end{bmatrix}g(x_0)
\end{equation*}
are nonzero, and, thus, the distribution $\Delta = \mathrm{span}\{g\}$ is nonsingular around $x_0$.
Being $1$-dimensional, this distribution is also involutive.
Therefore, by Frobenius' Theorem, we deduce the existence of $n-1$ real-valued functions, $\zeta_1(x),\dots,\zeta_{n-1}(x)$, defined in a neighborhood $\mathcal{D}$ of $x_0$, such that $d\zeta_1(x),\dots,d\zeta_{n-1}(x)$ are linearly independent for all $x\in \mathcal{D}$ and
\begin{align}
    \begin{split}
        \mathrm{span}\{d\zeta_1,\dots,d\zeta_{n-1}\} = \Delta^\perp 
        \iff d\zeta_i(x)g(x) = 0,\,\forall x\in \mathcal{D},\,\forall i\in\overline{1,n-1}.
    \end{split}
    \label{eqn:span_equal_G_perp}
\end{align}

\begin{lemma}
    \label{lem:orthogonal_complement_of_sum_of_subspaces}
    Let $A_1$ and $A_2$ be subspaces of $\mathbb{R}^n$.
    Then, $(A_1+A_2)^\perp = A_1^\perp\cap A_2^\perp$.
    $\hfill\square$
\end{lemma}
\begin{proof}
    Let $v\in A_1^\perp\cap A_2^\perp$.
    Then, $v\cdot a_1 = v\cdot a_2 = 0$ for all $a_1\in A_1$ and all $a_2\in A_2$.
    It follows that $v\cdot(a_1+a_2) = v\cdot a_1 + v\cdot a_2 = 0$ for all $a_1\in A_1$ and all $a_2\in A_2$.
    Thus, $v\in (A_1+A_2)^\perp$ and $A_1^\perp\cap A_2^\perp \subset (A_1+A_2)^\perp$.
    Let $v\not\in A_1^\perp\cap A_2^\perp$.
    Without loss of generality, $v\not\in A_1^\perp$.
    Then, there is $a_1\in A_1$ satisfying $v\cdot a_1 = v\cdot(a_1+0) \neq 0$.
    Since $a_1 + 0 \in A_1 + A_2$, we have $v\not\in (A_1+A_2)^\perp$.
    Thus, $(A_1+A_2)^\perp \subset A_1^\perp\cap A_2^\perp$.
    Therefore, $(A_1+A_2)^\perp = A_1^\perp\cap A_2^\perp$.
\end{proof}

Let $S = \{dh_1,\dots,dL_f^{{r}_1-1}h_1,\dots,dh_{p},\dots,dL_f^{{r}_{p}-1}h_{p}\}$ and $S(x) = \{dh_1(x),\dots,dL_f^{{r}_{p}-1}h_{p}(x)\}$.
It is easy to show that
\begin{equation}
    \mathrm{dim}(\Delta^\perp + \mathrm{span}S) = n
    \label{eqn:dim_G_perp_plus_span_equal_n}
\end{equation}
at $x_0$.
For, suppose this is false.
Then,
\begin{align*}
    &\mathrm{dim}(\Delta^\perp(x_0) + \mathrm{span}S(x_0)) < n \\
    \iff & \Delta^\perp(x_0) + \mathrm{span}S(x_0) \neq \mathbb{R}^n \\
    \iff & (\Delta^\perp(x_0)+\mathrm{span}S(x_0))^\perp \neq \{0\} \\
    &\downarrow\quad \text{Lemma \ref{lem:orthogonal_complement_of_sum_of_subspaces}} \\
    \iff & \Delta(x_0) \cap \mathrm{span}S(x_0)^\perp \neq \{0\}
\end{align*}
i.e., the vector $g(x_0)$ annihilates all the covectors in
\begin{equation*}
    \mathrm{span}\left\{\begin{array}{l}
        dh_1(x_0),\dots,dL_f^{{r}_1-1}h_1(x_0),\\
        \dots,\\
        dh_{p}(x_0),\dots,dL_f^{{r}_{p}-1}h_{p}(x_0)
    \end{array}\right\}.
\end{equation*}
But this is a contradiction, because $g$ satisfies
\begin{equation*}
    \begin{bmatrix} dL_f^{{r}_1-1}h_1(x_0) \\ \vdots \\ dL_f^{{r}_{p}-1}h_{p}(x_0) \end{bmatrix}g(x_0) \neq 0.
\end{equation*}

From \eqref{eqn:span_equal_G_perp}, \eqref{eqn:dim_G_perp_plus_span_equal_n}, and from the fact that $\mathrm{span}S$ has dimension ${r}$, by the first part of the proof, we see that
\begin{equation*}
    \left(
    \begin{split}
    &\Delta^\perp = \mathrm{span}\{d\zeta_1,\dots,d\zeta_{n-1}\} \\
    &S = \{dh_1,\dots,dL_f^{{r}_1-1}h_1,\dots,dh_{p},\dots,dL_f^{{r}_{p}-1}h_{p}\} \\
    &\mathrm{dim}(\Delta^\perp) = n-1 \\
    &\mathrm{dim}(\mathrm{span}S) = r = {p}{r} \\
    &\mathrm{dim}(\Delta^\perp + \mathrm{span}S) = n
    \end{split}
    \right)
\end{equation*}
implies
\begin{equation*}
    \mathrm{span}
    \left\{
    \begin{split}
    &dh_1,\dots,dL_f^{{r}_1-1}h_1,\\
    &\dots,\\
    &dh_{p},\dots,dL_f^{{r}_{p}-1}h_{p},\\
    &d\zeta_1,\dots,d\zeta_{n-1}
    \end{split}
    \right\}
    = \mathbb{R}^n.
\end{equation*}
Thus, we conclude that in the set $\{\zeta_1,\dots,\zeta_{n-1}\}$ it is possible to find $n-{r}$ functions, without loss of generality $\zeta_1,\dots,\zeta_{n-{r}}$, with the property that the $n$ differentials
\begin{align*}
    &d\zeta_1,\dots,d\zeta_{n-{r}}, \\
    &dh_1,\dots,dL_f^{{r}_1-1}h_1,\\
    &\dots,\\
    &dh_{p},\dots,dL_f^{{r}_{p}-1}h_{p}
\end{align*}
are linearly independent at $x_0$.
Since any other set of functions of the form $\zeta_i'(x) = \zeta_i(x) + c_i$, where $c_i$ is a constant, satisfies the same conditions:
\begin{align*}
    &\zeta_i'\in C^\bullet \iff \zeta_i\in C^\bullet \\
    &d\zeta_i'(x) = d\zeta_i(x) \\
    &\mathrm{dim}\{d\zeta_1'(x),\dots,d\zeta_{n-1}'(x)\} = \mathrm{dim}\{d\zeta_1(x),\dots,d\zeta_{n-1}(x)\} \\
    &L_g\zeta_i'(x) = d\zeta_i'(x)g(x) = d\zeta_i(x)g(x) = L_g\zeta_i(x) = 0,
\end{align*}
the value of these functions at the point $x_0$ can be chosen arbitrarily.

Define $\Phi:\mathbb{R}^n\to\mathbb{R}^n$ as
\begin{equation*}
    \Phi(x) = \mathrm{col}\left(\begin{array}{l}
        \zeta_1(x),\dots,\zeta_{n-{r}}(x),\\
        h_1(x),\dots,L_f^{{r}_1-1}h_1(x),\\
        \dots,\\
        h_{p}(x),\dots,L_f^{{r}_{p}-1}h_{p}(x)
    \end{array}\right)
\end{equation*}
where $\zeta_1,\dots,\zeta_{n-{r}}$ are ommited when ${r} = n$.
Since $\Phi$ is continuously differentiable and $\Phi(x_0)$ is nonsingular, $\Phi$ is a diffeomorphism on a neighborhood of $x_0$.

Let ${p} = {m}$, ${r}<n$, and the distribution $\Delta = \{g_1,\dots,g_{{m}}\}$ be involutive near $x_0$.
Since $J(x)$ is nonsingular by the definition of the relative degree, we have
\begin{equation*}
    {p}\le \mathrm{rank}(G(x_0)) \le {m} \le n.
\end{equation*}
Thus, $G(x_0)$ is nonsingular and so is $\Delta$.
Since $\Delta$ is nonsingular and involutive, by Frobenius' Theorem, there exist $n-{m}$ real-valued functions $\zeta_1(x),\dots,\zeta_{n-{m}}(x)$, defined in a neighborhood $\mathcal{D}$ of $x_0$, such that $d\zeta_1(x),\dots,d\zeta_{n-{m}}(x)$ are linearly independent for all $x\in \mathcal{D}$ and
\begin{align*}
    &\mathrm{span}\{d\zeta_1,\dots,d\zeta_{n-{m}}\} = \Delta^\perp 
    \iff L_{g_j}\zeta_i(x) = 0,\,\forall x\in \mathcal{D},\,\forall i\in\overline{1,n-{m}},\,\forall j\in\overline{1,{m}}.
\end{align*}
Consider now the codistribution
\begin{equation*}
    \Omega = \mathrm{span}\{dL_f^jh_i : 1\le i\le {p},\,0\le j\le {r}_i-1\}
\end{equation*}
which has dimension ${r}$.

Note that
\begin{equation}
    \begin{split}
        &\Delta(x_0)\cap\Omega^\perp(x_0) = \{0\} \\
        \iff &(\Delta(x_0)\cap\Omega^\perp(x_0))^\perp = \mathbb{R}^n \\
        \iff &\Delta^\perp(x_0) + \Omega(x_0) = \mathbb{R}^n \\
        \iff &\mathrm{dim}(\Delta^\perp(x_0)+\Omega(x_0)) = n.
    \end{split}
    \label{eqn:delta_perp_plus_omega_equals_r_n}
\end{equation}
For, if this were not true, there would exist a nonzero vector in $\Delta(x_0)$, i.e., a vector of the form
\begin{equation*}
    g = \sum_{i=1}^{p}c_ig_i(x_0) = G(x_0)\begin{bmatrix} c_1 \\ \vdots \\ c_{p} \end{bmatrix}
\end{equation*}
that would annihilate all vectors of $\Omega(x_0)$, but this is a contradiction, because
\begin{equation*}
    \begin{bmatrix} dL_f^{{r}_1-1}h_1(x_0) \\ \vdots \\ dL_f^{{r}_{p}-1}h_{p}(x_0) \end{bmatrix}g = J(x)\begin{bmatrix} c_1 \\ \vdots \\ c_{p} \end{bmatrix} = 0
\end{equation*}
implies $c_1=c_2=\cdots=c_{p}=0$, by the nonsingularity of $J(x_0)$.
Since 
\begin{align*}
    \mathrm{span}\{d\zeta_1,\dots,d\zeta_{n-{m}}\} &= \Delta^\perp \\
    \mathrm{dim}(\Delta^\perp(x_0) + \Omega(x_0)) &= n \\
    \mathrm{dim}(\Omega(x_0)) &= {r},
\end{align*}
we conclude that in the set $\{\zeta_1,\dots,\zeta_{n-{m}}\}$ it is possible to find $n-{r}$ functions, without loss of generality $\zeta_1,\dots,\zeta_{n-{r}}$, with the property that the $n$ differentials 
\begin{equation*}
    d\zeta_1,\dots,d\zeta_{n-{r}},dh_1,\dots,dL_f^{{r}_1-1}h_1,\dots,dh_{p},\dots,dL_f^{{r}_{p}-1}h_{p}
\end{equation*}
are linearly independent at $x_0$.
By construction, the functions $\zeta_1,\dots,\zeta_{n-{r}}$ are such that
\begin{equation*}
    d\zeta_i(x)g_j(x) = L_{g_j}\zeta_i(x) = 0
\end{equation*}
for all $x$ near $x_0$, all $1\le i\le n-{r}$, and all $1\le j\le {m}$.

\section{Proof of Lemma \ref{lem:condition_for_GN_equal_zero}}
\label{app:proof_of_lemma_about_condition_for_GN_equal_zero}

By the definition of the relative degree and $N(x)$, we have $(d\phi_i^jGN)(x) = (L_GL_f^{j-1}h_iN)(x) = 0$ for all $1\le i\le {p}$ and $1\le j\le {r}_i-1$ and 
\begin{equation*}
    \big(\mathrm{col}(d\phi_1^{{r}_1},\dots,d\phi_{p}^{{r}_{p}})GN\big)(x) = (JN)(x) = 0.
\end{equation*}
Thus, $(d\phi_{n-r+1:n}GN)(x) = 0$.
If $L_G\phi_{1:n-r}(x) = 0$, then
\begin{equation*}
    (d\Phi GN)(x)
    = \begin{bmatrix} L_G\phi_{1:n-r}N \\ d\phi_{n-r+1:n}GN \end{bmatrix}(x) = 0
\end{equation*}
and $(GN)(x) = 0$ becuase $d\Phi(x)$ is invertible.
If $(GN)(x) = 0$, then $(G_\eta N)(x) = (d\phi_{1:n-r}GN)(x) = 0$.
If $(GN)(x) \neq 0$, then $(d\Phi GN)(x) \neq 0$ and $(d\phi_{1:n-r}GN)(x) \neq 0$ because $d\Phi(x)$ is invertible and $(d\phi_{n-r+1:n}GN)(x) = 0$.

Assume that $G(x)$ is nonsingular on $\mathcal{D}$.
If ${p}={m}$, then $N(x) = 0$ and $(GN)(x) = 0$.
If ${p}<{m}$, then $(G^+GN)(x) = N(x) \neq 0$ and $(GN)(x) \neq 0$.
Thus, $(GN)(x)=0$ if and only if $p = m$.
Let $r = n$.
Then, the manifold \eqref{eqn:manifold_of_zero_dynamics} becomes $\mathcal{Z} = \{0\}$.
Thus, every solution of \eqref{eqn:zero_dynamics_original_coordinate} starting from 0 should stay at 0.
Therefore, $f^\circ(0) + (GN)(0)u_f(t) = 0$ for arbtrary $u_f(t)\in\mathbb{R}^m$, and it implies $(GN)(0) = 0$.

\section{Proof of Lemma \ref{lem:orthogonal_projectors}}
\label{app:proof_of_lemma_about_orthogonal_projectors}

1)
Denote $M_i = J_iN_{1:i-1}$ and let $1\le j<i\le {k}$.
Since
\begin{equation*}
    \mathcal{R}(J_j^T) \subset \mathcal{R}(J_{1:i-1}^T)\perp \mathcal{N}(J_{1:i-1})\subset\mathcal{N}(J_{1:j-1}),
\end{equation*}
we have $M_iM_j^T = J_iN_{1:i-1}J_j^T = 0$ and
\begin{equation*}
    P_iP_j = M_i^+M_iM_j^T(M_jM_j^T)^+M_j = 0.
\end{equation*}
Then, $P_iP_j=0$ for $1\le i\neq j\le {k}+1$ because orthogonal projectors are symmetric and idempotent.

2--6)
By \cite{Baerlocher2004}, $N_{1:i} = N_{1:i-1} - P_i$.
Then, $N_{1:i} = I_{{m}} - P_{1:i}$ and $J_i - J_iP_{1:i} = J_iN_{1:i} = 0$.
Next, $M_i = J_i - J_iP_{1:i-1} = J_iP_{1:i} - J_iP_{1:i-1} = J_iP_i$ and $\mathrm{rank}(P_i) = \mathrm{rank}(M_i^+M_i) = \mathrm{rank}(M_i) = \mathrm{rank}(J_iP_i) \le \mathrm{rank}(J_i) \le {p}_i$.
Finally, ${\rho}_{1:i} = \mathrm{rank}(P_{1:i}) = {m} - \mathrm{rank}(N_{1:i}) = \mathrm{rank}(J_{1:i}) \le {p}_{1:i}$ by the first property and the rank-nullity theorem.

\begin{center}
\begin{minipage}{0.52\textwidth}
\begin{algorithm}[H]
    \caption{Find $\begin{bmatrix} J(x) & 0 \end{bmatrix} = \bar{L}\bar{Q}$.}
    \begin{algorithmic}[1]
        \State $V = \mathrm{col}(v_1,\dots,v_{p}) \gets \begin{bmatrix} J(x) & 0 \end{bmatrix}\in\mathbb{R}^{{p}\times\max\{{p},{m}\}}$
        \State $\bar{Q} = \mathrm{col}(\bar{q}_1,\dots,\bar{q}_{p}) \gets 0 \in \mathbb{R}^{{p}\times\max\{{p},{m}\}}$
        \State $\bar{L} = [l_{ij}] \gets 0\in\mathbb{R}^{{p}\times {p}}$
        \For{$j = 1$ \textbf{to} ${p}$}
        \State $l_{jj} = \|v_j\|$
        \If{$l_{jj} > 0$}
        \State $\bar{q}_j = v_j / l_{jj}$
        \If{$j<{p}$}
        \For{$i = j + 1$ \textbf{to} ${p}$}
        \State $l_{ij} = \langle v_i,\bar{q}_j\rangle$
        \State $v_i = v_i - l_{ij}\bar{q}_j$
        \EndFor
        \EndIf
        \EndIf
        \EndFor
    \end{algorithmic}
    \label{alg:find_qr_decomposition}
\end{algorithm}
\end{minipage}
\end{center}

7) Run Algorithm \ref{alg:find_qr_decomposition}.
Then, we get $\begin{bmatrix} J(x) & 0 \end{bmatrix} = \bar{L}\bar{Q}$ where $\bar{L}$ is lower triangular; $l_{ii}\ge0$ for $1\le i\le {p}$; $l_{ij} = 0$ for $1\le i\le {p}$ if $l_{jj} = 0$; $\bar{q}_i = 0$ if and only if $l_{ii} = 0$; and nonzero rows of $\bar{Q}$ are orthonormal.
Observe that if ${p}>{m}$, then, by the 7-th line of Algorithm \ref{alg:find_qr_decomposition}, we have $\bar{Q} = \begin{bmatrix} \hat{Q} & 0 \end{bmatrix}$ where $\hat{Q}\in\mathbb{R}^{{p}\times {m}}$.
Also, the total number of nonzero diagonal entries of $\bar{L}$ is equal to ${\rho} = \mathrm{rank}(J(x)) \le \min\{{p},{m}\}$.
Let $a_1,\dots,a_{\rho}$ be such that $l_{a_ia_i} \neq 0$ for $1\le i\le {\rho}$.

Let $l_i$ and $q_i$ be the $i$-th column of $\bar{L}$ and the $i$-th row of $\hat{Q}$, respectively.
Define
\begin{align*}
    \tilde{L} &= \mathrm{row}\left(l_{a_1},\dots,l_{a_{\rho}},0_{m\times({m}-{\rho})}\right)\in\mathbb{R}^{{p}\times {m}} \\
    \check{Q} &= \mathrm{col}\left(q_{a_1},\dots,q_{a_{\rho}},0_{({m}-{\rho})\times {m}}\right) \in \mathbb{R}^{{m}\times {m}}.
\end{align*}
Then, we have
\begin{equation*}
    \begin{bmatrix} J(x) & 0 \end{bmatrix} 
    = \bar{L}\bar{Q}
    = \sum_{i=1}^{\rho} l_{a_i}\begin{bmatrix} q_{a_i} & 0 \end{bmatrix}
    = \begin{bmatrix} \tilde{L}\check{Q} & 0 \end{bmatrix}
\end{equation*}
and $J_i(x) = \sum_{j=1}^i\tilde{L}_{ij}\check{Q}_i$ for $1\le i\le k$ where $\tilde{L}_{ij}\in\mathbb{R}^{{p}_i\times {\rho}_i}$ is the $(i,j)$-th block of $\tilde{L}$; $\check{Q}_i\in\mathbb{R}^{{\rho}_i\times {m}}$ is the $i$-th row block of $\check{Q}$; $\mathrm{rank}(\tilde{L}_{ii}) = {\rho}_i$; and $\check{Q}_{{k}+1} = 0$.

If ${\rho}<{m}$, then let $\{p_1,\dots,p_{{m}-{\rho}}\} \subset\mathbb{R}^{{m}}$ be an orthonormal basis of $\mathcal{N}(J(x))$.
Construct $\tilde{Q}\in\mathbb{R}^{{m}\times {m}}$ by replacing zero rows of $\check{Q}$ with $p_1,\dots,p_{{m}-{\rho}}$.
Since $\mathcal{R}(\check{Q}^T) = \mathcal{R}(J^T(x)) \perp \mathcal{N}(J(x))$, $\tilde{Q}$ is orthogonal.
By letting $L_e(x) = \tilde{L}$ and $Q_e(x) = \tilde{Q}$, we get the full LQ decomposition \eqref{eqn:orthogonalization_of_J}.

Since $L_{ii}$ has full rank and $Q_i$ has orthonormal rows, we have $(L_{ii}Q_i)^+(L_{ii}Q_i) = Q_i^TL_{ii}^+L_{ii}Q_i^T = Q_i^TQ_i$ for $1\le i\le {k}$.
By the decomposition \eqref{eqn:orthogonalization_of_J} and the definition \eqref{eqn:definitions_when_r_j_equals_0}, $J_i(x) = \sum_{j=1}^iL_{ij}Q_j$ for $1\le i\le k$.
Thus, $J_1P_1 = J_1 = L_{11}Q_1$ and $P_1 = (L_{11}Q_1)^+(L_{11}Q_1) = Q_1^TQ_1$.
Assume that $J_jP_j = L_{jj}Q_j$ and $P_j = Q_j^TQ_j$ for $1\le j\le i-1$.
Then, 
\begin{equation*}
    J_iP_i = J_i - (L_{i1}Q_1 + \cdots + L_{i,i-1}Q_{i-1}) = L_{ii}Q_i
\end{equation*}
and $P_i = (L_{ii}Q_i)^+(L_{ii}Q_i) = Q_i^TQ_i$.
Since $P_{{k}+1} = I_{{m}} - P_{1:{k}} = I_{{m}} - \sum_{i=1}^{k}Q_i^TQ_i = Q_{{k}+1}^TQ_{{k}+1}$, we have $J_iP_j = (L_{i1}Q_1 + \cdots + L_{ii}Q_i)Q_j^TQ_j = L_{ij}Q_j$ for all $1\le i\le {k}$ and $1\le j\le {k}+1$.

\section{Proof of Lemma \ref{lem:control_input_for_canonical_prioritized_semilinear_form}}
\label{app:proof_of_lemma_about_control_input_for_canonical_prioritized_semilinear_form}

Since $\sum_{j=1}^{i-1}L_{ij}Q_ju_j^{\sf c} = J_iu_{1:i-1}^{\sf c}$, we have $u^{\sf c}(x,v,u_f)$ in the recursive form \eqref{eqn:control_input_for_canonical_linear_form_recursive} from \eqref{eqn:linearizing_control_input_u_i} and \eqref{eqn:canonical_prioritized_linearizer}.
Since $L_{ij}Q_ju_j^{\sf c} = L_{ij}Q_ju_{1:k}^{\sf c}$, we can formulate
\begin{align*}
    u_{1:k}^{\sf c} &= u_1^{\sf c} + \cdots + u_k^{\sf c} 
        = \sum_{i=1}^kQ_i^TL_{ii}^{+(\lambda_i)}\left(v_i - \kappa_i - \sum_{j=1}^{i-1}L_{ij}Q_ju_j^{\sf c}\right) \\
        &= Q^T\begin{bmatrix} L_{11}^{+(\lambda_1)}(v_1-\kappa_1) \\ L_{22}^{+(\lambda_2)}(v_2 - \kappa_2 - L_{21}Q_1u_1^{\sf c}) \\ \vdots \\ L_{kk}^{+(\lambda_k)}\left(v_k - \kappa_k - \sum_{i=1}^{k-1}L_{ki}Q_iu_i^{\sf c}\right) \end{bmatrix} 
        = Q^TL_D^{+(\lambda)}\begin{bmatrix} v_1-\kappa_1 \\ v_2-\kappa_2-L_{21}Q_1u_1^{\sf c} \\ \vdots \\ v_k - \kappa_k - \sum_{i=1}^{k-1}L_{ki}Q_iu_i^{\sf c} \end{bmatrix} \\
        &= Q^TL_D^{+(\lambda)}\begin{bmatrix} v_1-\kappa_1 \\ v_2-\kappa_2-L_{21}Q_1u_{1:k}^{\sf c} \\ \vdots \\ v_k - \kappa_k - \sum_{i=1}^{k-1}L_{ki}Q_iu_{1:k}^{\sf c} \end{bmatrix} 
        = Q^TL_D^{+(\lambda)}\left(v - \kappa - L_LQu_{1:k}^{\sf c}\right)
\end{align*}
and
\begin{equation*}
    Qu_{1:k}^{\sf c} = (I_\rho + L_D^{+(\lambda)} L_L)^{-1}L_D^{+(\lambda)} (v - \kappa).
\end{equation*}
Since $L_D^{+(\lambda)} L_L$ and $L_LL_D^{+(\lambda)}$ are strictly block lower triangular, we have
\begin{equation*}
    \begin{split}
        (I_\rho + L_D^{+(\lambda)} L_L)^{-1}L_D^{+(\lambda)} 
            &= \left[ I_\rho - L_D^{+(\lambda)} L_L + \cdots + (-L_D^{+(\lambda)} L_L)^{k-1}\right]L_D^{+(\lambda)} \\
            &= L_D^{+(\lambda)}\left[I_p - L_LL_D^{+(\lambda)} + \cdots + (-L_LL_D^{+(\lambda)})^{k-1}\right] \\
            &= L_D^{+(\lambda)}(I_p + L_LL_D^{+(\lambda)})^{-1}
    \end{split}
    \label{eqn:change_order_of_L_D_circledast}
\end{equation*}
and we find $u^{\sf c}(x,v,u_f)$ in the closed form \eqref{eqn:control_input_for_canonical_linear_form_closed}.

\section{Proof of Lemma \ref{lem:solution_of_proper_objective_function_equals_solution_of_multi_objective_optimization}}
\label{app:proof_of_lemma_about_solution_of_proper_objective_function_equals_solution_of_multi_objective_optimization}

Fix $x\in{\cal D}$ and $v\in\mathbb{R}^{p}$ and let $\phi_i(u) = \pi_i(x,v_i,u)$ for $1\le i\le k$.
By the dependence property of Definition \ref{def:proper_objective_function}, there exist unique $u_1^*,\dots,u_k^*\in\mathbb{R}^m$ satisfying
\begin{equation*}
    \{u_i^*\} = \argmin_{u\in\mathcal{R}(P_i(\mathbf{x}))}\phi_i(u_{1:i-1}^* + u) \quad (1\le i\le k)
\end{equation*}
where $u_{i:j}^* = u_i^* + \cdots + u_j^*$.
Then, by \eqref{eqn:proper_objective_function_dependence}, we have
\begin{subequations}\label{eqn:multi_objective_optimization}
\begin{align}
    \phi_1(u_{1:k}^*) &= \min\{\phi_1(u): u\in\mathbb{R}^m\} \label{eqn:multi_objective_optimization_phi_1}\\
    \phi_2(u_{1:k}^*) &= \min\{\phi_2(u): u\in u_1^* + \mathcal{R}(I_m-P_1(x)) \} \nonumber\\
    &= \min\{\phi_2(u): u\in\mathbb{R}^m \text{ and } \phi_1(u) = \phi_1(u_1^*)\} \label{eqn:multi_objective_optimization_phi_2}\\
    &\vdotswithin{=} \nonumber\\
    \phi_k(u_{1:k}^*) &= \min\{\phi_k(u) : u\in u_{1:k-1}^* + \mathcal{R}(I_m - P_{1:k-1}(x))\} \nonumber\\
    &= \min\{\phi_k(u) : u\in\mathbb{R}^m \text{ and } \phi_i(u) = \phi_i(u_{1:k}^*) \text{ for } 1\le i\le k-1 \}. \label{eqn:multi_objective_optimization_phi_k}
\end{align}
\end{subequations}
Thus, $u_{1:k}^*$ is a unique solution of \eqref{eqn:multi_objective_optimization_with_lexicographical_ordering}.

\section{Proof of Footnote \ref{foot:convergnece_and_boundedness_of_single_task}}
\label{app:proof_of_footnote_about_convergnece_and_boundedness_of_single_task}

While the assertion made in Footnote \ref{foot:convergnece_and_boundedness_of_single_task} is widely known and can be found in references such as \cite[\S 13.4.2]{Khalil2015} and \cite[\S 8]{Isidori1995}, we provide a proof here to offer both the background of this study and an introduction to the main result of this paper, Theorem \ref{thm:prioritized_output_tracking} and Corollary \ref{cor:ultimate_bound_of_tracking_errors}.
Moreover, the proof for the statement in Footnote \ref{foot:convergnece_and_boundedness_of_single_task} shares a similar logical flow with the proof of Theorem \ref{thm:prioritized_output_tracking}, so readers can gain insight into the proof of Theorem \ref{thm:prioritized_output_tracking} from this proof before we introduce various analysis tools in order to handle the imperfect inversion and the discontinuous orthogonalization for the prioritized output tracking control.

Since we assumed that $\dot{\eta} = f_\eta^\circ[\eta,0]$ is asymptotically stable, there exist a class ${\cal KL}$ function $\beta_\eta$ and a positive constant $\hat{\sf r}_\eta$ such that for every $\|\eta(0)\|<\hat{\sf r}_\eta$ there exists asolution $\eta:[0,\infty)\to\mathbb{R}^{n-r}$ of $\dot{\eta} = f_\eta^\circ[\eta,0]$ with the initial value $\eta(0)$ satisfying $\|\eta(t)\|\le \beta_\eta(\|\eta(0)\|,t)$ for all $t\ge0$ \cite[Lemma 4.5]{Khalil2015}.
Let $0<\bar{\sf r}_\eta<\hat{\sf r}_\eta$ be such that $\beta_\eta(\bar{\sf r}_\eta,0)<\hat{\sf r}_\eta$.
By the converse Lyapunov theorem \cite[Theorem 4.16]{Khalil2015}, there is a continuously differentiable function $V_\eta:\bar{\sf r}_\eta{\cal B}_{n-r}\to \mathbb{R}$ satisfying
\begin{gather*}
    \bar{\alpha}_1(\|\eta\|)\le V_\eta(\eta) \le \bar{\alpha}_2(\|\eta\|) \\
    \frac{\partial V_\eta}{\partial \eta}f_\eta[\eta,0] \le -\bar{\alpha}_3(\|\eta\|) \\
    \left\|\frac{\partial V_\eta}{\partial \eta}\right\| \le \bar{\alpha}_4(\|\eta\|)
\end{gather*}
for all $\eta\in \bar{\sf r}_\eta{\cal B}_{n-r}$ where $\bar{\alpha}_i$ for $1\le i\le 4$ are class ${\cal K}$ functions defined on $[0,\bar{\sf r}_\eta]$.
Since $A-BK$ is assumed to be Hurwitz, there exists a solution $X = X^T>0$ of the Lyapunov equation $$X(A-BK) + (A-BK)^TX = -2I_r.$$
Define $V_{\tilde{\xi}}(\tilde{\xi}) = \sqrt{\tilde{\xi}^TX\tilde{\xi}}$ for $\tilde{\xi}\in\mathbb{R}^r$.

Since we assumed that $\xi^\star(t)$ and $\kappa^\star(t)$ are bounded, there exist positive constants ${\sf M}_{\xi^\star}$ and ${\sf M}_{\kappa^\star}$ satisfying $\|\xi^\star(t)\|<{\sf M}_{\xi^\star}$ and $\|\kappa^\star(t)\|<{\sf M}_{\kappa^\star}$ for all $t\in[0,\infty)$.
Since the domain ${\cal D}$ of the diffeomorphism $\Phi$ is open and both $\Phi$ and $\Phi^{-1}$ are continuous, $\Phi({\cal D})\subset\mathbb{R}^n$ is open.
Let ${\cal C}\subset\Phi({\cal D})$ be a compact neighborhood of 0, $0<{\sf r}_\eta \le \bar{\sf r}_\eta$, and ${\sf M}_{\xi^\star}<{\sf r}_{\tilde{\xi}} + {\sf M}_{\xi^\star} = {\sf r}_{\xi}$ be such that
\begin{equation*}
    \tilde{\cal E} \coloneqq {\sf r}_\eta{\cal B}_{n-r}\times {\sf r}_{\tilde{\xi}}{\cal B}_r \subset {\cal E} \coloneqq {\sf r}_\eta{\cal B}_{n-r}\times{\sf r}_{\xi}{\cal B}_r\subset{\cal C}.
\end{equation*}
Note that ${\rm cl}({\cal E})$ is a compact subset of ${\cal C}$ and $\tilde{z} = (\eta,\tilde{\xi})\in\tilde{\cal E}$ implies $z = (\eta,\xi)\in{\cal E}$.
We propose a Lyapunov function candidate $V:\tilde{\cal E}\to\mathbb{R}$ for the entire system \eqref{eqn:closed_system_conventional_linear_form} as
\begin{equation*}
    V(\eta,\tilde{\xi}) = V_\eta(\eta) + {\sf w}V_{\tilde{\xi}}(\tilde{\xi})
\end{equation*}
with a weight ${\sf w}>0$.
Let $\delta_0 = \min\{{\sf r}_\eta,{\sf r}_{\tilde{\xi}}\}$.
Since $V$ is continuous and positive definite, there exist class ${\cal K}$ functions $\alpha_1$ and $\alpha_2$ defined on $[0,\delta_0]$ satisfying
\begin{equation*}
    \alpha_1(\|(\eta,\tilde{\xi})\|) \le V(\eta,\tilde{\xi}) \le \alpha_2(\|(\eta,\tilde{\xi})\|)
\end{equation*}
for all $(\eta,\tilde{\xi})\in\delta_0{\cal B}_n\subset\tilde{\cal E}$ \cite[Lemma 4.3]{Khalil2015}.

Assume that $f_{\eta}[z]$, $G_\eta[z]$, $J[z]$, and $\kappa[z]$ are continuously differentiable on $\Phi({\cal D})$.
Since $J[z]$ has full rank on $\Phi({\cal D})$, $J^+[z]$ keeps smoothness of $J[z]$ on $\Phi({\cal D})$.
Thus, $f_\eta^\circ = f_\eta - G_\eta J^+\kappa$ and $G_\eta J^+$ are continuously differentiable on $\Phi({\cal D})$.
Since ${\rm cl}({\cal E})$ is a compact subset of $\Phi({\cal D})$, there exist positive constants ${\sf M}_{GJ^+}$ and ${\sf L}_{f_\eta^\circ}$ satisfying
\begin{align*}
    \|(G_\eta J^+)[z]\| &\le {\sf M}_{G_\eta J^+} \\
    \|f_\eta^\circ[z_1] - f_\eta^\circ[z_2]\| &\le {\sf L}_{f_\eta^\circ}\|z_1-z_2\|
\end{align*}
for all $z,z_1,z_2\in{\cal E}$.
Also, $${\sf M}_{\bar{\alpha}_4} = \max_{\|\eta\|<{\sf r}_\eta}\bar{\alpha}_4(\|\eta\|)<\infty.$$
Then, we can compute the time derivative of $V$ along the trajectories of the system \eqref{eqn:closed_system_conventional_linear_form} as
\begin{align*}
    \dot{V} &= \frac{\partial V_\eta}{\partial \eta}f_\eta^\circ[\eta,0] + \frac{\partial V_\eta}{\partial\eta}(f_\eta^\circ[\eta,\xi]-f_\eta^\circ[\eta,0]) + \frac{\partial V_\eta}{\partial \eta}G_\eta J^+(-K\tilde{\xi}+\kappa^\star) \\
        &\quad + {\sf w}\frac{\tilde{\xi}^T(X(A-BK) + (A-BK)^TX)\tilde{\xi}}{2\sqrt{\tilde{\xi}^TX\tilde{\xi}}} \\
        &\le -\bar{\alpha}_3(\|\eta\|) + \bar{\alpha}_4(\|\eta\|){\sf L}_{f_\eta^\circ}(\|\tilde{\xi}\|+{\sf M}_{\xi^\star}) + \bar{\alpha}_4(\|\eta\|){\sf M}_{G_\eta J^+}(\|K\|\|\tilde{\xi}\|+{\sf M}_{\kappa^\star}) - {\sf w}\frac{\|\tilde{\xi}\|}{\sigma_{\max}^{1/2}(X)} \\
        &= -\bar{\alpha}_3(\|\eta\|) + ({\sf L}_{f_\eta^\circ}{\sf M}_{\xi^\star} + {\sf M}_{G_\eta J^+}{\sf M}_{\kappa^\star})\bar{\alpha}_4(\|\eta\|) - \left(\frac{{\sf w}}{\sigma_{\max}^{1/2}(X)} - \bar{\alpha}_4(\|\eta\|)({\sf L}_{f_\eta^\circ} + {\sf M}_{G_\eta J^+}\|K\|)\right)\|\tilde{\xi}\|
\end{align*}
for all ${\sf w}>0$.
Thus, for sufficiently large ${\sf w}>0$ and sufficiently small ${\sf M}_{\xi^\star}, {\sf M}_{\kappa^\star}>0$, one can always find ${\sf c}_1,{\sf c}_2>0$, $0<\delta<\delta_0$, and $0<\varepsilon<\alpha_2^{-1}(\alpha_1(\delta))$ satisfying
\begin{equation*}
    \dot{V} \le -W(\eta,\tilde{\xi}) \coloneqq -{\sf c}_1\bar{\alpha}_3(\|\eta\|) - {\sf c}_2\|\tilde{\xi}\|
\end{equation*}
for all $\varepsilon\le \|(\eta,\tilde{\xi})\|\le\delta$.
Therefore, by \cite[Theorem 4.18]{Khalil2015}, there exist class ${\cal KL}$ function $\beta$ and for every initial state $(\eta(0),\tilde{\xi}(0))$ satisfying $\|(\eta(0),\tilde{\xi}(0))\|<\alpha_2^{-1}(\alpha_1(\delta))$, there is $T_0>0$ (dependent on $(\eta(0),\tilde{\xi}(0))$ and $\varepsilon$) such that the solution of \eqref{eqn:closed_system_conventional_linear_form} satisfies
\begin{equation*}
    \|(\eta(t),\tilde{\xi}(t))\| \le \begin{dcases*} \beta(\|(\eta(0),\tilde{\xi}(0))\|,t), & $t\in[0,T_0]$ \\ \alpha_1^{-1}(\alpha_2(\varepsilon)), & $t\in[T_0,\infty)$. \end{dcases*}
\end{equation*}

\section{Proof of Lemma \ref{lem:solution_existence_of_differential_inclusion}}
\label{app:proof_of_lemma_about_solutionlem:solution_existence_of_differential_inclusion}

For simplicity, we write \eqref{eqn:differential_inclusion_of_prioritized_closed_system_on_omega_0} as $\dot{z}_0\in \mathcal{F}(t,z_0)$.
Since $\xi^\star(t)$ and $\kappa^\star(t)$ are bounded on $[0,\infty)$ and $(LQ)[z_0]$ and $\Gamma[z_0]$ are bounded on the compact set ${\mathcal{C}}_0$, $\mathcal{F}(t,z_0)$ is bounded on $[0,\infty)\times{\mathcal{C}}_0$.
Then, there exists a bounded upper semicontinuous set-valued map $\widetilde{\mathcal{F}}:\mathbb{R}\times\mathbb{R}^n\to2^{\mathbb{R}^n}$ with closed convex values satisfying $\widetilde{\mathcal{F}}(t,z_0) = \mathcal{F}(t,z_0)$ for all $(t,z_0)\in[0,\infty)\times {\mathcal{C}}_0$ \cite[Lemma 2.10]{Smirnov2002}.
It follows that for each $z_0^\circ \in \mathrm{int}({\mathcal{C}}_0)$, there exists an absolutely continuous function $z_0:[0,\infty)\to\mathbb{R}^n$ satisfying $z_0(0) = z_0^\circ$ and $\dot{z}_0(t)\in \widetilde{\mathcal{F}}(t,z_0(t))$ almost everywhere \cite[Chapter 4, Corollary 1.12, Exercise 1.14]{Clarke1998}, which concludes the proof. 

\section{Proof of Lemma \ref{lem:superset_of_differential_inclusion}}
\label{app:proof_of_lemma_about_superset_of_differential_inclusion}

First, define set-valued maps $\mathcal{M}_{ij}:{\Phi_0^{-1}({\cal C}_0)} \to 2^{\mathbb{R}^{m\times p_j}}$ for $1\le i,j\le k$ and $\mathcal{U}_{i_0+1:k+1}:[0,\infty)\times{\Phi_0^{-1}({\cal C}_0)} \to2^{\mathbb{R}^m}$ as:
\begin{align*}
    \mathcal{M}_{ij}(x) &\coloneqq \mathcal{K}[Q_i^TL_{ii}^T\Gamma_{ij}](x),\\
    \mathcal{U}_{i_0+1:k+1}(t,x) &\coloneqq \sum_{i=i_0+1}^k\sum_{j=1}^i\mathcal{M}_{ij}(x)(-K_j\tilde{\xi}_j(t,x) + \kappa_j^\star(t) + \kappa_j(x)) + N(x)\mathcal{U}_f.
\end{align*}
Since $N_0 = I_m - J_0^+J_0$ is smooth on ${\Phi_0^{-1}({\cal C}_0)}$ and $N_0 u_{i_0+1:k+1}^\pi = u_{i_0+1:k+1}^\pi$ for all $(t,x)$, we have $N_0 {\mathcal{U}}_{i_0+1:k+1} = {\mathcal{U}}_{i_0+1:k+1}$ on $[0,\infty)\times{\Phi_0^{-1}({\cal C}_0)}$.
Note also that ${\mathcal{U}}_{i_0+1:k+1}(t,x)$ is bounded on $[0,\infty)\times\Phi_0^{-1}({\mathcal{C}}_0)$ so that there exists a positive constant ${\sf M}_{\cal U}$ satisfying ${\cal U}_{i_0+1:k+1}(t,{\Phi_0^{-1}({\cal C}_0)})\subset{\sf M}_{\cal U}{\cal B}_m$ for all $t\in[0,\infty)$.
Thus, one can find  an upper semicontinuous set-valued map $\mathcal{U}:\Phi_0^{-1}({\mathcal{C}}_0)\to2^{\mathbb{R}^m}$ satisfying the following properties:
\begin{itemize}
    \item $\mathcal{U}(x)$ is a compact and convex neighborhood of 0 containing $\mathcal{U}_{i_0+1:k+1}(t,x)$ for all $(t,x)\in[0,\infty)\times\Phi_0^{-1}({\mathcal{C}}_0)$ (so that Item~(a) of the lemma holds); and
    \item Item~(b) of Lemma~\ref{lem:superset_of_differential_inclusion} is satisfied. 
\end{itemize}
While it is possible to define ${\cal U}$ as ${\cal U}(x) = {\sf M}_{\cal U}{\cal B}_m$ for all $x\in{\Phi_0^{-1}({\cal C}_0)}$, it should be noted that this choice may not satisfy Assumption \ref{asm:zero_dynamics_lyapunov}. Hence, careful consideration is required when selecting ${\cal U}$.
The proof will be completed if we show
\begin{equation*}
    \mathcal{K}[f_0^\circ + G_0^\circ(-K_0\tilde{\xi}_0 + \kappa_{0}^\star) + G_0 u_{i_0+1:k+1}^\pi] \subset f_0^\circ + G_0^\circ(-K_0 \tilde{\xi}_0+ \kappa_{0}^\star) + G_0\mathcal{U}_{i_0+1:k+1}
\end{equation*}
for all $(t,x)\in[0,\infty)\times\Phi_0^{-1}({\cal C}_0)$, by which Item~(c) of the lemma is derived.

Denote $f^* = f_0^\circ + G_0^\circ(-K_0\tilde{\xi}_0 + \kappa_0^\star) + G_0u_{i_0+1:k+1}^\pi$.
By observing
\begin{equation*}
    \overline{\mathrm{co}}\left(f^*\left(t,x+\frac{1}{a}\mathrm{cl}(\mathcal{B}_n)\right)\right)
        = \begin{dcases*} \bigcap_{1\le \delta<\infty}\overline{\mathrm{co}}(f^*(t,x+\delta\mathrm{cl}(\mathcal{B}_n))), & $a=1$ \\ \bigcap_{\frac{1}{a}\le \delta<\frac{1}{a-1}}\overline{\mathrm{co}}(f^*(t,x+\delta\mathrm{cl}(\mathcal{B}_n))), & $a\ge2$ \end{dcases*}
\end{equation*}
for $a\in\mathbb{N}$, we can rewrite the Krasovskii regularization of $f^*(t,x)$ as
\begin{equation*}
    {\cal K}[f^*](t,x) 
        = \bigcap_{\delta>0}\overline{\mathrm{co}}(f^*(t,x+\delta\mathrm{cl}(\mathcal{B}_n)))
        = \bigcap_{a\in\mathbb{N}}\overline{\mathrm{co}}\left(f^*\left(t,x+\frac{1}{a}\mathrm{cl}(\mathcal{B}_n)\right)\right).
\end{equation*}
Since ${\cal C}_0$ is compact and $\Phi_0^{-1}$ is continuous on $\Phi_0({\cal D}_0)$, $\Phi_0^{-1}({\cal C}_0)$ is compact \cite[Theorem 2.10]{Rudin1987}.
Since ${\cal D}_0$ is open and $\Phi_0^{-1}({\cal C}_0)\subset{\cal D}_0$, for every $x\in\Phi_0^{-1}({\cal C}_0)$ there exists $a_x\in\mathbb{N}$ such that $x+(1/a)\mathrm{cl}(\mathcal{B}_n) \subset {\cal D}_0$ for all $a\ge a_x$.
Since $\Gamma(x)$ is assumed to be bounded on every compact subset of ${\cal D}_0$ in Assumption \ref{asm:Gamma} and $x+(1/a)\mathrm{cl}(\mathcal{B}_n)\subset{\cal D}_0$ is compact, $f^*(t,x+(1/a)\mathrm{cl}(\mathcal{B}_n))$ is bounded for all $x\in\Phi_0^{-1}({\cal C}_0)$ and $a\ge a_x$.

\begin{lemma}
    \label{lem:convex_closure_equals_closure_of_convex_hull_for_bounded_set}
    If $A\subset\mathbb{R}^n$ is bounded, then $\overline{\rm co}(A) \coloneqq {\rm cl}({\rm co}(A)) = {\rm co}({\rm cl}(A))$.
    $\hfill\square$
\end{lemma}
\begin{proof}
    If $A\subset\mathbb{R}^n$ is bounded, then $\mathrm{cl}(A)$ is compact and so is $\mathrm{co}(\mathrm{cl}(A))$ \cite[Corollary 1.1]{Smirnov2002}.
    Then, $\overline{\mathrm{co}}(A)\subset{\rm cl}({\rm co}(\mathrm{cl}(A))) = \mathrm{co}(\mathrm{cl}(A))$.
    Since $\overline{\mathrm{co}}(A)$ is convex, we have $$A\subset\mathrm{co}(A) \implies \mathrm{cl}(A)\subset\overline{\mathrm{co}}(A) \implies \mathrm{co}(\mathrm{cl}(A))\subset\overline{\mathrm{co}}(A).$$
    Thus, $\overline{\mathrm{co}}(A) = \mathrm{co}(\mathrm{cl}(A))$ for every bounded $A\subset\mathbb{R}^n$
\end{proof}

By Lemma \ref{lem:convex_closure_equals_closure_of_convex_hull_for_bounded_set}, we have
\begin{align*}
    {\cal K}[f^*](t,x) 
        &= \bigcap_{a\in\mathbb{N}}\overline{\mathrm{co}}\left(f^*\left(t,x+\frac{1}{a}\mathrm{cl}(\mathcal{B}_n)\right)\right) \\
        &= \bigcap_{a\in\mathbb{N}}\mathrm{co}\left(\mathrm{cl}\left(f^*\left(t,x+\frac{1}{a}\mathrm{cl}(\mathcal{B}_n)\right)\right)\right) \\
        &= \bigcap_{a\in\mathbb{N}}\mathrm{co}\left(\left\{\lim_{b\to\infty}f^*(t,x_b) : x_b \in x+\frac{1}{a}\mathrm{cl}(\mathcal{B}_n)\right\}\right) \\
        &\downarrow\quad \text{\cite[Lemma A.1]{Paden1987}} \\
        &= \mathrm{co}\left(\bigcap_{a\in\mathbb{N}}\left\{\lim_{b\to\infty}f^*(t,x_b) : x_b \in x+\frac{1}{a}\mathrm{cl}(\mathcal{B}_n)\right\}\right) \\
        &= \mathrm{co}\left(\left\{\lim_{b\to\infty}f^*(t,x_b) : x_b\to x\right\}\right).
\end{align*}
Since $Q_i^TL_{ii}^T\Gamma_{ij}$ for $1\le i,j\le k$ is locally bounded on ${\cal D}_0$ and $f_0$, $G_0$, $\kappa$, $\tilde{\xi}$, $L_{ij}Q_j$, and $\Gamma_{ij}$ for $1\le i,j\le i_0$ are continuous on ${\cal D}_0$, for each sequence $x_b\to x$ satisfying that $\lim_{b\to\infty}f^*(t,x_b)$ exists, there exists a subsequence $x_{b_c}\to x$ satisfying that $\lim_{c\to\infty}(Q_i^TL_{ii}^T\Gamma_{ij})(x_{b_c})$ and $\lim_{c\to\infty}u_f(t,x_{b_c})$ exist for all $1\le i\le k$ and $1\le j\le i$ such that
\begin{align*}
    \lim_{b\to\infty}f^*(t,x_b)
        &= f_0(x) + G_0(x)\sum_{i=1}^k\sum_{j=1}^i\lim_{c\to\infty}(Q_i^TL_{ii}^T\Gamma_{ij})(x_{b_c})(v_j(t,x)-\kappa_j(x)) \\
        &\quad + G_0(x)N(x)\lim_{c\to\infty}u_f(t,x_{b_c}) \\
        &= f_0^\circ(x) + G_0^\circ(x)v_0(t,x) + G_0(x)\sum_{i=i_0+1}^k\sum_{j=1}^i\lim_{c\to\infty}(Q_i^TL_{ii}^T\Gamma_{ij})(x_{b_c})(v_j(t,x) - \kappa_j(x)) \\
        &\quad + G_0(x)N(x)\lim_{c\to\infty}u_f(t,x_{b_c})
\end{align*}
where $v_i = -K_i\tilde{\xi}_i + \kappa_i^\star$ and $v_0 = {\rm col}(v_1,\dots,v_{i_0}) = -K_0\tilde{\xi}_0 + \kappa_0^\star$.
It follows that
\begin{align*}
    {\cal K}[f^*](t,x)
        &= \mathrm{co}\left(\left\{\lim_{b\to\infty}f^*(t,x_b) : x_b\to x\right\}\right) \\
        &= \mathrm{co}\bigg(\bigg\{f_0^\circ(x) + G_0^\circ(x)v_0(t,x) + \sum_{i=i_0+1}^k\sum_{j=1}^iG_0(x)\lim_{c\to\infty}(Q_i^TL_{ii}^T\Gamma_{ij})(x_{b_c})(v_j(t,x)-\kappa_j(x)) \\
        &~~~~~~~~~~~~+ G_0(x)N(x)\lim_{c\to\infty}u_f(t,x_{b_c}) : x_b\to x,\,\{x_{b_c}\}\subset\{x_b\}\bigg\}\bigg) \\
        &\subset \mathrm{co}\bigg(\bigg\{f_0^\circ(x) + G_0^\circ(x)v_0(t,x) + \sum_{i=i_0+1}^k\sum_{j=1}^iG_0(x)\lim_{b\to\infty}(Q_i^TL_{ii}^T\Gamma_{ij})(x_b)(v_j(t,x)-\kappa_j(x)) \\
        &~~~~~~~~~~~~+ G_0(x)N(x)\lim_{b\to\infty}u_f(t,x_b) : x_b\to x\bigg\}\bigg) \\
        &\subset \mathrm{co}\bigg(f_0^\circ(x) + G_0^\circ(x)v_0(t,x) + \bigg\{\sum_{i=i_0+1}^k\sum_{j=1}^iG_0(x)\lim_{b\to\infty}(Q_i^TL_{ii}^T\Gamma_{ij})(x_b)(v_j(t,x)-\kappa_j(x)) \\
        &~~~~~~~~~~~~~~~~~~~~~~~~~~~~~~~~~~~~~~~~~~+ G_0(x)N(x)\lim_{b\to\infty}u_f(t,x_b) : x_b\to x\bigg\}\bigg) \\
        &= f_0^\circ(x) + G_0^\circ(x)v_0(t,x) + G_0(x)N(x){\rm co}\left(\left\{\lim_{b\to\infty}u_f(t,x_b) : x_b\to x\right\}\right) \\
        &~~~~~+ \sum_{i=i_0+1}^k\sum_{j=1}^iG_0(x)\mathrm{co}\left(\left\{\lim_{b\to\infty}(Q_i^TL_{ii}^T\Gamma_{ij})(x_b) : x_b\to x\right\}\right)(v_j(t,x)-\kappa_j(x))
\end{align*}
where the last equality stems from the fact that the convex hull of the sum of two sets is equal to the sum of convex hulls of sets.
By applying the same prosedure that shows
\begin{equation*}
    {\cal K}[f^*](t,x) = {\rm co}\left(\left\{\lim_{b\to\infty}f^*(t,x_b) : x_b\to x\right\}\right),
\end{equation*}
we can derive
\begin{align*}
    {\cal K}[Q_i^TL_{ii}^T\Gamma_{ij}](x) &= \mathrm{co}\left(\left\{\lim_{b\to\infty}(Q_i^TL_{ii}^T\Gamma_{ij})(x_b) : x_b\to x\right\}\right) = {\cal M}_{ij}(x) \\
    {\cal K}[u_f](t,x) &= {\rm co}\left(\left\{\lim_{b\to\infty}u_f(t,x_b) : x_b\to x\right\}\right) \subset {\cal U}_f.
\end{align*}
Therefore, we have
\begin{align*}
    {\cal K}[f^*](t,x)
        &\subset f_0^\circ(x) + G_0^\circ(x)v_0(t,x) + \sum_{i=i_0+1}^k\sum_{j=1}^iG_0(x){\cal M}_{ij}(x)(v_j(t,x)-\kappa_j(x)) + G_0(x)N(x){\cal U}_f \\
        &= f_0^\circ(x) + G_0^\circ(x)(-K_0\tilde{\xi}_0(t,x) + \kappa_0^\star(t)) + G_0(x){\cal U}_{i_0+1:k+1}(t,x)
\end{align*}
for all $(t,x)\in[0,\infty)\times\Phi_0^{-1}({\cal C}_0)$.

\section{Proof of the Last Step of Theorem \ref{thm:prioritized_output_tracking}}
\label{app:proof_of_footnote_about_find_ultimate_bound}

The proof of the last step of Theorem \ref{thm:prioritized_output_tracking} shares the same idea of the proof of \cite[Theorem 4.18]{Khalil2015}.
Nevertheless, we present it here in order to show how \cite[Theorem 4.18]{Khalil2015} can be adapted to a Lyapunov function candidate $V(\eta_0,\tilde{\xi}_0)$ whose time derivative exists only almost everywhere by using the fundamental theorem of calculus \cite[Theorem 7.20]{Rudin1987}.
Since 
\begin{equation*}
    0<\varepsilon,\,\|(\eta_0(0),\tilde{\xi}_0(0))\|<\alpha_2^{-1}(\alpha_1({\delta}))\le \delta,
\end{equation*}
we have
\begin{equation*}
    0 < \sigma \coloneqq \alpha_2(\varepsilon),\,\alpha_2(\|(\eta_0(0),\tilde{\xi}_0(0))\|) < \rho \coloneqq \alpha_1({\delta}) \le \alpha_2(\delta).
\end{equation*}
Define
\begin{equation*}
    \Omega_\bullet = \{(\eta_0,\tilde{\xi}_0)\in{\delta}\mathcal{B}_n : V(\eta_0,\tilde{\xi}_0)<\bullet\}
\end{equation*}
for $\bullet\in\{\sigma,\rho\}$.
Since $V(\eta_0,\tilde{\xi}_0)$ is continuous on $\delta{\cal B}_n$, $\Omega_\bullet$ is open and we have
\begin{align*}
    {\rm bd}(\Omega_\bullet) &= \{(\eta_0,\tilde{\xi}_0)\in\delta{\cal B}_n : V(\eta_0,\tilde{\xi}_0) = \bullet\} \\
    {\rm cl}(\Omega_\bullet) &= \{(\eta_0,\tilde{\xi}_0)\in\delta{\cal B}_n : V(\eta_0,\tilde{\xi}_0) \le \bullet\}
\end{align*}
for $\bullet\in\{\sigma,\rho\}$.
Then,
\begin{equation*}
    \varepsilon{\cal B}_n\subset\Omega_\sigma\subset{\rm cl}(\Omega_\sigma)\subset\Omega_\rho\subset\{\alpha_1(\|(\eta_0,\tilde{\xi}_0)\|)<\rho\}=\delta{\cal B}_n.
\end{equation*}
Define 
\begin{equation*}
    W(\eta_0,\tilde{\xi}_0) \coloneqq {\sf c}_1\bar{\alpha}_3(\|\eta_0\|) + {\sf c}_2\|\tilde{\xi}_0\|.
\end{equation*}
Since $W$ is continuous and positive definite on $\delta{\cal B}_n$, we have
\begin{equation*}
    c \coloneqq \min_{\varepsilon\le\|(\eta_0,\tilde{\xi}_0)\|\le{\delta}}W(\eta_0,\tilde{\xi}_0) > 0.
\end{equation*}

The sets $\Omega_\sigma$ and $\Omega_\rho$ have the property that every Krasovskii solution $(\eta_0(t),\tilde{\xi}_0(t))$ starting in ${\rm cl}(\Omega_\sigma)$ or $\Omega_\rho$ cannot leave ${\rm cl}(\Omega_\sigma)$ or $\Omega_\rho$, respectively.
To see this, firstly let $(\eta_0(0),\tilde{\xi}_0(0))\in{\rm cl}(\Omega_\sigma)$ and suppose that there exists $T>0$ satisfying $(\eta_0(T),\tilde{\xi}_0(T))\in\Omega_\rho\setminus{\rm cl}(\Omega_\sigma)$.
Then, there exist $0<t_1<t_2\le T$ such that $(\eta_0(t_1),\tilde{\xi}_0(t_1))\in{\rm bd}(\Omega_\sigma)$ and $(\eta_0(t),\tilde{\xi}_0(t))\in\Omega_\rho\setminus{\rm cl}(\Omega_\sigma)$ for all $t\in(t_1,t_2]$.
It follows that
\begin{equation*}
    \dot{V}(t,\eta_0(t),\tilde{\xi}_0(t)) \le -W(\eta_0(t),\tilde{\xi}_0(t)) \le -c
\end{equation*}
for almost all $t\in[t_1,t_2]$.
Thus, we have
\begin{align*}
    V(\eta_0(t),\tilde{\xi}_0(t)) &= V(\eta_0(t_1),\tilde{\xi}_0(t_1)) + \int_{t_1}^t\dot{V}(s,\eta_0(s),\tilde{\xi}_0(s))ds \\
        &\le V(\eta_0(t_1),\tilde{\xi}_0(t_1)) - c(t-t_1) \\
        &\le V(\eta_0(t_1),\tilde{\xi}_0(t_1)) \\
        &= \sigma
\end{align*}
and $(\eta_0(t),\tilde{\xi}_0(t))\in{\rm cl}(\Omega_\sigma)$ for all $t\in[t_1,t_2]$, a contradiction.
Therefore, $(\eta_0(t),\tilde{\xi}_0(t))\in{\rm cl}(\Omega_\sigma)$ for all $t\in[0,\infty)$.
Secondly, let $(\eta_0(0),\tilde{\xi}_0(0))\in\Omega_\rho$ and suppose that there exists $T>0$ satisfying $(\eta_0(T),\tilde{\xi}_0(T))\in{\rm bd}(\Omega_\rho)$.
Then, there exist $0<t_1<t_2\le T$ such that $(\eta_0(t),\tilde{\xi}_0(t))\in\Omega_\rho\setminus{\rm cl}(\Omega_\sigma)$ for all $t\in[t_1,t_2)$ and $(\eta_0(t_2),\tilde{\xi}_0(t_2))\in{\rm bd}(\Omega_\rho)$.
Similarly as above, we have
\begin{align*}
    \rho &= V(\eta(t_2),\tilde{\xi}(t_2)) \\
        &= V(\eta(t_1),\tilde{\xi}(t_1)) + \int_{t_1}^{t_2}\dot{V}(s,\eta(s),\tilde{\xi}(s))ds \\
        &\le V(\eta(t_1),\tilde{\xi}(t_1)) - c(t_2-t_1) \\
        &< V(\eta(t_1),\tilde{\xi}(t_1)) \\
        &<\rho,
\end{align*}
a contradiction.
Therefore, $(\eta_0(t),\tilde{\xi}_0(t))\in\Omega_\rho$ for all $t\in[0,\infty)$.

Since
\begin{equation*}
    \|(\eta_0(0),\tilde{\xi}_0(0))\| < \alpha_2^{-1}(\alpha_1({\delta})) \implies (\eta_0(0),\tilde{\xi}_0(0))\in\Omega_\rho,
\end{equation*}
we conclude that $(\eta_0(t),\tilde{\xi}_0(t))\in\Omega_\rho$ for all $t\ge0$.
Every Krasovskii solution starting in $\Omega_\rho\setminus{\rm cl}(\Omega_\sigma)$ must enter ${\rm cl}(\Omega_\sigma)$ in finite time because $\dot{V}\le -c<0$ on $\Omega_\rho\setminus{\rm cl}(\Omega_\sigma)$, so that
\begin{equation*}
    V(\eta_0(t),\tilde{\xi}_0(t)) \le V(\eta_0(0),\tilde{\xi}_0(0)) - ct \le \rho - ct
\end{equation*}
which shows that $V(\eta_0(t),\tilde{\xi}_0(t))$ reduces to $\sigma$ within the time interval $[0,(\rho-\sigma)/c]$.
For a Krasovskii solution starting inside ${\rm cl}(\Omega_\sigma)$, inequality
\begin{equation*}
    \|(\eta_0(t),\tilde{\xi}_0(t))\| \le \alpha_1^{-1}(\alpha_2(\varepsilon))
\end{equation*}
is satisfied for all $t\ge0$, because $\mathrm{cl}(\Omega_\sigma)\subset\{\alpha_1(\|(\eta_0,\tilde{\xi}_0)\|)\le\sigma = \alpha_2(\varepsilon)\}$.
For a Krasovskii solution starting inside $\Omega_\rho$, but outside $\mathrm{cl}(\Omega_\sigma)$, let $T_0$ be the first time it enters $\mathrm{cl}(\Omega_\sigma)$.
For almost all $t\in[0,T_0]$,
\begin{equation*}
    \dot{V} \le -W(\eta_0,\tilde{\xi}_0) \le -\alpha_3(\|(\eta_0,\tilde{\xi}_0)\|) \le -\alpha_3(\alpha_2^{-1}(V)) \coloneqq -\alpha(V)
\end{equation*}
where $\alpha_3$ and $\alpha$ are class $\mathcal{K}$ functions.
The existence of $\alpha_3$ follows from \cite[Lemma 4.3]{Khalil2015}.
Then, we have
\begin{align*}
    D^+V(\eta_0(t),\tilde{\xi}_0(t)) 
    &= \limsup_{h\downarrow0}\frac{V(\eta_0(t+h),\tilde{\xi}_0(t+h)) - V(\eta_0(t),\tilde{\xi}_0(t))}{h} \\
    &= \limsup_{h\downarrow0}\frac{1}{h}\int_t^{t+h}\dot{V}(s,\eta_0(s),\tilde{\xi}_0(s))ds \\
    &\le \limsup_{h\downarrow0}\frac{1}{h}\int_t^{t+h}(-\alpha(V(\eta_0(s),\tilde{\xi}_0(s))))ds \\
    &\downarrow\quad \dot{V}<0 \implies V(\eta_0(s),\tilde{\xi}_0(s)) \ge V(\eta_0(t+h),\tilde{\xi}_0(t+h)) \\
    &\le \limsup_{h\downarrow0}\frac{1}{h}\int_t^{t+h}(-\alpha(V(\eta_0(t+h),\tilde{\xi}_0(t+h))))ds \\
    &= \limsup_{h\downarrow0}\frac{1}{h}h(-\alpha(V(\eta_0(t+h),\tilde{\xi}_0(t+h)))) \\
    &= -\alpha(V(\eta_0(t),\tilde{\xi}_0(t)))
\end{align*}
for all $t\in[0,T_0)$.
Without loss of generality, we assume that $\alpha(\cdot)$ is locally Lipschitz.
Then, there exists a unique solution $v(t)$ satisfying
\begin{equation*}
    \dot{v} = -\alpha(v),\quad v(0) = V(\eta_0(0),\tilde{\xi}_0(0)).
\end{equation*}
By the comparison lemma \cite[Lemma 3.4]{Khalil2015},
\begin{equation*}
    V(\eta_0(t),\tilde{\xi}_0(t)) \le v(t),\quad \forall t\in[0,T_0].
\end{equation*}
By \cite[Lemma 4.4]{Khalil2015}, there exists a class $\mathcal{KL}$ function $\beta_0(r,s)$ defined on $[0,{\delta}]\times[0,\infty)$ such that
\begin{equation*}
    V(\eta_0(t),\tilde{\xi}_0(t)) \le \beta_0(V(\eta_0(0),\tilde{\xi}_0(0)),t),\quad\forall t\in[0,T_0].
\end{equation*}
Defining $\beta(r,s) = \alpha_1^{-1}(\beta_0(\alpha_2(r),s))$, we obtain
\begin{equation*}
    \|(\eta_0(t),\tilde{\xi}_0(t))\| \le \beta(\|(\eta_0(0),\tilde{\xi}_0(0))\|,t),\quad \forall t\in[0,T_0].
\end{equation*}

\section{Proof of the Omitted Part of Corollary \ref{cor:ultimate_bound_of_tracking_errors}}
\label{app:proof_of_the_omitted_part_of_ultimate_bound_of_tracking_errors}

Let
\begin{equation*}
    \tilde{\sf c}_i = \frac{\tilde{\sigma}_i}{\vartheta_i}\sum_{j=1}^i{\sf M}_{E_{ij}}^\circ\big({\sf L}_{\kappa_j}(\alpha_1^{-1}(\alpha_2(\varepsilon)) + {\sf M}_{\xi_0^\star}) + {\sf M}_{\kappa_j^\star}\big)
\end{equation*}
for $1\le i\le i_1$.
From \eqref{eqn:upper_bound_of_norm_of_tracking_error}, we have
\begin{equation*}
    \|\tilde{\xi}_1(t)\| \le \sqrt{\tilde{\sigma}_1}\|\tilde{\xi}_1(t_0)\|e^{-\vartheta_1(t-T_1)} + \tilde{\sf c}_1.
\end{equation*}
Thus, \eqref{eqn:upper_bound_of_zeta_i} holds for $i=1$ with ${\sf a}_{11} = \sqrt{\tilde{\sigma}_1}$, ${\sf b}_{11} = \vartheta_1$, and ${\sf c}_1 = \tilde{\sf c}_1$.
Assume that
\begin{equation*}
    \|\tilde{\xi}_j(t)\| \le \sum_{a=1}^j{\sf a}_{ja}\|\tilde{\xi}_a(T_1)\|e^{-{\sf b}_{ja}(t-T_1)} + {\sf c}_j
\end{equation*}
holds for all $t\in[T_1,\infty)$ and $1\le j < i$.
Let
\begin{equation*}
    {\sf c}_i = \frac{\tilde{\sigma}_i}{\vartheta_i}\sum_{j=1}^{i-1}{\sf M}_{E_{ij}}^\circ\|K_j\|{\sf c}_j + \tilde{\sf c}_i.
\end{equation*}
Then, we have from \eqref{eqn:upper_bound_of_norm_of_tracking_error}
\begin{align*}
    \|\tilde{\xi}_i(t)\| 
        &\le \sqrt{\tilde{\sigma}_i}\|\tilde{\xi}_i(T_1)\|e^{-\vartheta_i(t-T_1)} + \tilde{\sigma}_i\sum_{j=1}^{i-1}{\sf M}_{E_{ij}}^\circ\|K_j\|\int_{T_1}^t\left(\sum_{a=1}^j{\sf a}_{ja}\|\tilde{\xi}_a(T_1)\|e^{-{\sf b}_{ja}(s-T_1)} + {\sf c}_j\right)e^{-\vartheta_i(t-s)}ds + \tilde{\sf c}_i \\
        &\le \sqrt{\tilde{\sigma}_i}\|\tilde{\xi}_i(T_1)\|e^{-\vartheta_i(t-T_1)} + \tilde{\sigma}_i\sum_{j=1}^{i-1}{\sf M}_{E_{ij}}^\circ\|K_j\|\left(\sum_{a=1}^j{\sf a}_{ja}\|\tilde{\xi}_a(T_1)\|e^{{\sf b}_{ja}T_1-\vartheta_it}\int_{T_1}^te^{(\vartheta_i-{\sf b}_{ja})s}ds + \frac{{\sf c}_j}{\vartheta_i}\right) + \tilde{\sf c}_i \\
        &= \sqrt{\tilde{\sigma}_i}\|\tilde{\xi}_i(T_1)\|e^{-\vartheta_i(t-T_1)} + \tilde{\sigma}_i\sum_{j=1}^{i-1}{\sf M}_{E_{ij}}^\circ\|K_j\|\sum_{a=1}^j{\sf a}_{ja}\|\tilde{\xi}_a(T_1)\|e^{{\sf b}_{ja}T_1-\vartheta_it}\int_{T_1}^te^{(\vartheta_i-{\sf b}_{ja})s}ds + {\sf c}_i \\
        &= \sqrt{\tilde{\sigma}_i}\|\tilde{\xi}_i(T_1)\|e^{-\vartheta_i(t-T_1)} + \sum_{a=1}^{i-1}\sum_{j=1}^a\tilde{\sigma}_i{\sf M}_{E_{ia}}^\circ\|K_a\|{\sf a}_{aj}\|\tilde{\xi}_j(T_1)\|e^{{\sf b}_{aj}T_1-\vartheta_it}\int_{T_1}^te^{(\vartheta_i-{\sf b}_{aj})s}ds + {\sf c}_i \\
        &= \sqrt{\tilde{\sigma}_i}\|\tilde{\xi}_i(T_1)\|e^{-\vartheta_i(t-T_1)} + \sum_{j=1}^{i-1}\sum_{a=j}^{i-1}\tilde{\sigma}_i{\sf M}_{E_{ia}}^\circ\|K_a\|{\sf a}_{aj}\|\tilde{\xi}_j(T_1)\|e^{{\sf b}_{aj}T_1-\vartheta_it}\int_{T_1}^te^{(\vartheta_i-{\sf b}_{aj})s}ds + {\sf c}_i \\
\end{align*}
for all $t\in[T_1,\infty)$.
We can compute
\begin{equation*}
    e^{{\sf b}_{aj}T_1-\vartheta_it}\int_{T_1}^te^{(\vartheta_i-{\sf b}_{aj})s}ds = (t-T_1)e^{-\vartheta_i(t-T_1)} \le \left(\max_{t\in[T_1,\infty)}(t-T_1)e^{-(\vartheta_i/2)(t-T_1)}\right)e^{-(\vartheta_i/2)(t-T_1)}
\end{equation*}
if $\vartheta_i = {\sf b}_{aj}$ and
\begin{equation*}
    e^{{\sf b}_{aj}T_1-\vartheta_it}\int_{T_1}^te^{(\vartheta_i-{\sf b}_{aj})s}ds = \frac{e^{-{\sf b}_{aj}(t-T_1)} - e^{-\vartheta_i(t-T_1)}}{\vartheta_i-{\sf b}_{aj}} < \begin{dcases*} \frac{e^{-{\sf b}_{aj}(t-T_1)}}{\vartheta_i-{\sf b}_{aj}}, & $\vartheta_i>{\sf b}_{aj}$ \\ \frac{e^{-\vartheta_i(t-T_1)}}{{\sf b}_{aj}-\vartheta_i}, & $\vartheta_i < {\sf b}_{aj}$ \end{dcases*}
\end{equation*}
if $\vartheta_i \neq {\sf b}_{aj}$.
Therefore, there exist ${\sf a}_{ija}, {\sf b}_{ija}>0$ satisfying
\begin{equation*}
    e^{{\sf b}_{aj}T_1-\vartheta_it}\int_{T_1}^te^{(\vartheta_i-{\sf b}_{aj})s}ds \le {\sf a}_{ija}e^{-{\sf b}_{ija(t-T_1)}}
\end{equation*}
for all $t\in[T_1,\infty)$ and $1\le j\le a\le i-1$.
Let
\begin{align*}
    {\sf a}_{ij} &= \begin{dcases*} \sqrt{\tilde{\sigma}_i}, & $i = j$ \\ \tilde{\sigma}_i\sum_{a=j}^{i-1}{\sf M}_{E_{ia}}^\circ\|K_a\|{\sf a}_{aj}{\sf a}_{ija}, & $j<i$ \end{dcases*} \\
    {\sf b}_{ij} &= \begin{dcases*} \vartheta_i, & $i=j$ \\ \min\{{\sf b}_{ija} : j\le a\le i-1\}, & $j<i$ \end{dcases*}
\end{align*}
for $1\le j\le i$.
Then, we have
\begin{align*}
    \|\tilde{\xi}_i(t)\| 
        &\le \sqrt{\tilde{\sigma}_i}\|\tilde{\xi}_i(T_1)\|e^{-\vartheta_i(t-T_1)} + \sum_{j=1}^{i-1}\sum_{a=j}^{i-1}\tilde{\sigma}_i{\sf M}_{E_{ia}}^\circ\|K_a\|{\sf a}_{aj}\|\tilde{\xi}_j(T_1)\|{\sf a}_{ija}e^{-{\sf b}_{ija}(t-T_1)} + {\sf c}_i \\
        &\le {\sf a}_{ii}\|\tilde{\xi}_i(T_1)\|e^{-{\sf b}_{ii}(t-T_1)} + \sum_{j=1}^{i-1}\left(\tilde{\sigma}_i\sum_{a=j}^{i-1}{\sf M}_{E_{ia}}^\circ\|K_a\|{\sf a}_{aj}{\sf a}_{ija}\right)\|\tilde{\xi}_j(T_1)\|e^{-{\sf b}_{ij}(t-T_1)} + {\sf c}_i \\
        &= \sum_{j=1}^i{\sf a}_{ij}\|\tilde{\xi}_j(T_1)\|e^{-{\sf b}_{ij}(t-T_1)} + {\sf c}_i
\end{align*}
for all $t\in[T_1,\infty)$.

\bibliographystyle{IEEEtran}
\newcommand{\noop}[1]{}

\end{document}